\documentclass[a4paper,11pt]{article}
\usepackage{listings}
  \usepackage{mathrsfs}
       \usepackage[mathscr]{euscript}
 \usepackage{thm-restate}

\usepackage{vmargin}
\setmarginsrb{1in}{1in}{1in}{1in}{0pt}{0pt}{0pt}{6mm}
\usepackage{amsthm,amsmath,amssymb}
\usepackage{libertine}
\newtheorem{theorem}{Theorem}[section]
\newtheorem{lemma}[theorem]{Lemma}

\newtheorem{clm}[theorem]{Claim}

\newtheorem{definition}[theorem]{Definition}
\newtheorem{observation}[theorem]{Observation}
\newtheorem{proposition}[theorem]{Proposition}

\newtheorem{reduction rule}{Reduction Rule}
\newtheorem*{reduction rule*}{Reduction Rule}

\usepackage{todonotes}

\usepackage{algorithmicx}
\usepackage{algpseudocode}
\usepackage{mdframed}
 \usepackage{enumitem}

\usepackage{bm}
\usepackage[lined,boxed,commentsnumbered,ruled,vlined,noend,linesnumbered,boxed]{algorithm2e}
\usepackage{amssymb}
\usepackage{tcolorbox}
\tcbuselibrary{skins}
\usepackage{xcolor}
\usepackage{mathtools}
\usepackage{thmtools}
\usepackage{framed}
\usepackage{mdframed}
\usepackage{amsmath}
\usepackage{amssymb}

\newboolean{shortver}
\setboolean{shortver}{false}% for long version

\declaretheorem[name = Claim, numberwithin=theorem]{clm}
\makeatletter
\newcommand{\problemtitle}[1]{\gdef\@problemtitle{#1}}% Store problem title
\newcommand{\probleminput}[1]{\gdef\@probleminput{#1}}% Store problem input
\newcommand{\problemquestion}[1]{\gdef\@problemquestion{#1}}% Store problem question
\NewEnviron{problem}{
	\problemtitle{}\probleminput{}\problemquestion{}% Default input is empty
	\BODY% Parse input
	\par\addvspace{.5\baselineskip}
	\noindent
	\begin{tabularx}{\textwidth}{@{\hspace{\parindent}} l X c}
		\multicolumn{2}{@{\hspace{\parindent}}l}{\@problemtitle} \\% Title
		\textbf{Input:} & \@probleminput \\% Input
		\textbf{Goal:} & \@problemquestion% Question
	\end{tabularx}
	\par\addvspace{.5\baselineskip}
}
\colorlet{mix}{red!50!black}

\newcommand{\G}{\mathsf{VIG}}

\newcommand{\vis}{\mathsf{Vis}}

\newcommand{\reg}{\mathtt{Reg}}

\newcommand{\gsp}{{\sc MGS}$(\mathscr{P})$}

\newcommand{\uc}{\mathsf{uC}}
\newcommand{\lc}{\mathsf{\ell C}}
\newcommand{\bd}{\partial}
\newcommand{\inte}{\mathsf{iN}}
\newcommand{\ex}{\mathsf{eX}}

\newcommand{\pws}{\text{potential witness set}}

\newcommand{\wld}{\text{well-defined}}

\newcommand{\wts}{{\sc Witness Set}}

\newcommand{\ter}{\mathsf{Territory}}
\newcommand{\tra}{\mathsf{tRegion}}

\newcommand{\RA}{\mathsf{rAnchor}}
\newcommand{\LA}{\mathsf{\ell Anchor}}

\newcommand{\PC}{\mathsf{pChain}}
\newcommand{\PCU}{\mathsf{pChain}_{u}}
\newcommand{\PCL}{\mathsf{pChain}_{\ell}}
\newcommand{\XM}{\mathsf{xMax}}
\newcommand{\XMi}{\mathsf{xMin}}
\newcommand{\RC}{\mathsf{rChord}}
\newcommand{\LC}{\mathsf{\ell Chord}}
\newcommand{\IMr}{\mathsf{rImage}}
\newcommand{\IMl}{\mathsf{\ell Image}}
\newcommand{\Rb}{\mathsf{rBdry}}
\newcommand{\Lb}{\mathsf{\ell Bdry}}
\newcommand{\spr}{\mathsf{\Pi_{max}}}
\newcommand{\spl}{\mathsf{\Pi_{min}}}

\newcommand{\vig}{\mathsf{VIG}}
\newcommand{\str}{\mathsf{Str}}
\newcommand{\sig}{\mathsf{SIG}}
\newcommand{\infp}{\mathsf{Inf_{\delta}(P)}}

\newcommand{\po}{\mathscr{P}}
\newcommand{\mo}{\mathscr{M}}
\newcommand{\ro}{\mathcal{R}}
\newcommand{\zm}{\mathcal{Z}_{\mathtt{mid}}}

\newcommand{\OO}{\mathcal{O}}

	\newcommand{\fpt} {{\sf FPT}\xspace}
	\newcommand{\nph} {{\sf NP}-hard\xspace}

    	\newcommand{\wgb}{W^{\mathtt{good}}_{\mathtt{bdry}}}
		\newcommand{\wbb}{W^{\mathtt{bad}}_{\mathtt{bdry}}}
			\newcommand{\wint}{W_{\mathtt{int}}}

\usepackage{tcolorbox}
\tcbuselibrary{skins}
\usepackage{xcolor}

\usepackage{subcaption}
\colorlet{mix}{red!50!black}

\usepackage{nameref}
\definecolor{ForestGreen}{rgb}{0.1333,0.5451,0.1333}
\definecolor{DarkRed}{rgb}{0.8,0,0}
\definecolor{Red}{rgb}{1,0,0}
\usepackage[linktocpage=true,
pagebackref=true,colorlinks,
linkcolor=DarkRed,citecolor=ForestGreen,
bookmarks,bookmarksopen,bookmarksnumbered]
{hyperref}

 \usepackage{cleveref}
\usepackage{amsmath}

\makeatletter
\providecommand*{\cupdot}{%
  \mathbin{%
    \mathpalette\@cupdot{}%
  }%
}
\newcommand*{\@cupdot}[2]{%
  \ooalign{%
    $\m@th#1\cup$\cr
    \hidewidth$\m@th#1\cdot$\hidewidth
  }%
}
\makeatother

\newcounter{joe}

\newcounter{matya}

\newcounter{sasanka}

\newcounter{satya}

\newcounter{binayak}

\begin{document}

\title{Witness Set in Monotone Polygons: Exact and Approximate}

\author{
 Udvas Das \thanks{Indian Statistical Institute, Kolkata, India, \textrm{udvas.das@gmail.com}}
 \and 
 Binayak	Dutta\thanks{Christ University, Bangalore, India,  \textrm{binayak66@gmail.com}}
  % \and Satyabrata Jana \thanks{Corresponding Author (Satyabrata Jana)}  \hspace*{.3mm}\thanks{University of Warwick, UK, \textrm{satyamtma@gmail.com.}}
   \and Satyabrata Jana \thanks{University of Warwick, UK, \textrm{satyamtma@gmail.com.}}
  \and Debabrata Pal \thanks{Indian Statistical Institute, Kolkata, India, \textrm{debabratapal4521@gmail.com}}
\and Sasanka Roy \thanks{Indian Statistical Institute, Kolkata, India, \textrm{
sasanka.ro@gmail.com}}
}

\date{}

 \maketitle

%\thispagestyle{empty}

% 	\maketitle
% \sasanka{Why there is no affiliations? is it decided which journal to be submitted and who is supposed to do the formatting?}

% \satya{I will do. As all the previous comments based on this draft line number, I keep this format for now. When all comments will be done, I will change update the format.}

%  \sasanka{Do we need to reduce the abstract size??? I think it might be fine for journal version...}

%  \sasanka{I think we may be able to get a polytime algorithm for continuous witness problem for simple polygon which I have discussed with Debabrata. So it may be nice to submit it to CGTA and get it published. I don't think it will be wise to wait too long to publish all other results. Udvas should speed up the things, but it is difficult to find his time.}

	\begin{abstract}

	\noindent 	Given a simple polygon $\mathscr{P}$ on $n$ vertices, two points $x$ and $y$ within $\mathscr{P}$ are {\em visible} to each other
if the line segment between $x$ and $y$ is contained in $\mathscr{P}$. The {\em visibility region} of a point $x$ includes all points in $\mathscr{P}$ that are visible from $x$.  A point set $Q$ within a polygon $\mathscr{P}$ is said to be a \emph{witness set} for $\mathscr{P}$ if each point in $\mathscr{P}$ is visible from at most one point from $Q$. In other terms, the visibility region of any two points $s, t \in Q$ are non-intersecting within $\mathscr{P}$. The problem of finding the largest size witness set in a given polygon was introduced by Amit et al. [Int. J. Comput. Geom. Appl. 2010]. 
        Recently,  Daescu et al. [Comput. Geom. 2019] gave a linear-time algorithm for this problem on monotone mountains, a proper subclass of monotone polygons. In this study, we contribute to this field by obtaining the largest witness set within both continuous and discrete models.
        
        In the {\sc Witness Set (WS)} problem, the input is a polygon $\mathscr{P}$, and the goal is to find a maximum-sized witness set in $\mathscr{P}$. In the {\sc Discrete Witness Set (DisWS)} problem, one is given a finite set of points $S$ alongside $\mathscr{P}$, and the task is to find a witness set $Q \subseteq S$ that maximizes $|Q|$. We investigate {\sc DisWS} in simple polygons, but consider {\sc WS} specifically for monotone polygons. Our main contribution is as follows: (1) a polynomial time algorithm for {\sc DisWS} for general polygons and (2) the discretization of the {\sc WS} problem for monotone polygons, leading to several algorithmic developments. Our approach only employs basic geometric primitives and circumvents the need for algebraic tools, unlike the methods used in simple polygon guarding by Abrahamsen et al.[J. ACM. 2022], Efrat and Sariel Har-Peled [Inf. Process. Lett. 2006]. This implies that the Witness Problem admits discretization for monotone polygon and conjecture that, unlike polygon guarding as shown by Bonnet and  Miltzow [SoCG 2017], {\sc WS} admits a finite discretization for general polygon. Our discretization, along with its properties, may independently interest those solving {\sc WS} in general polygon contexts and polygon guarding scenarios like smooth or robust polygon guarding [Das et al., SoCG 2024,  Dobbins et al., arXiv:1811.01177]. Specifically, we determine the following for {\sc WS} when restricted to $n$-vertex monotone polygons with $r$ reflex vertices, where the size of the maximum witness set is $k$.
        \begin{itemize}
            \item Generate  a point set $Q$ with size $r^{\OO(k)} \cdot n$ such that $Q$ contains an optimal solution.
            
            \item An exact algorithm  running in time  $r^{\OO(k)} \cdot n^{\OO(1)}$. 
            
            \item A $(1+ \epsilon)$-approximation algorithm  with running time  $r^{\OO(1/\epsilon)} \cdot n^2$.

        \end{itemize} 
 {\sc DisWS} is equivalent to finding a maximum independent set (MIS) of visibility region intersection graphs, where each visibility region is a vertex and there exists an edge between a pair of vertices if and only if the corresponding visibility regions intersect. Thus, we obtain a polynomial-time algorithm for computing MIS in visibility region intersection graphs. This characterization motivates the study of which graph classes are subgraphs of visibility region intersection graphs. To corroborate our claim, we demonstrate that the visibility region intersection graph of a polygon can be classified as an outerstring graph. The string model we construct is consistent with the one outlined by Hengeveld and Miltzow [SoCG 2021]. For monotone polygons, these graphs turn into co-comparable graphs. This enables us to resolve the {\sc DisWS} problem in polynomial time for general polygons, and with enhanced running time for monotone polygons.

	\end{abstract}

\UseRawInputEncoding

	\section{Introduction}
 Given a simple polygon $\po$, we say that two points $p,q \in \po$ are {\em visible  to each other} 
if the line segment $\overline{pq}$ is contained in $\po$. 
A set of points $W \subseteq \po$ is said to be a {\em witness set} in  $\po$ if every point $p \in \po$ is visible from at most one witness $w \in W$.  The points of $W$ are called {\em witnesses}. It may happen that some points in $\po$ are not visible from any witness of $W$. A witness set of $\po$ is optimal if it is a maximum cardinality witness set of $\po$. In our work, we consider the {\sc Witness Set} problem: given a simple polygon $\po$, find a maximum-sized witness set in $\po$.

\begin{figure}[b]
    \centering
    \includegraphics[width=0.3\linewidth]{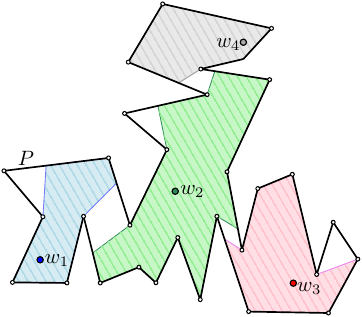}
    \caption{Example of a polygon having  four witnesses $w_1,w_2, w_3, w_4$ with their visibility regions.}
    \label{witfirst}
\end{figure}

% Few  related computational problems that arises from this notion and studied extensively in literature is referred as the  {\sc Art Gallery}, {\sc Hidden Set}, {\sc Convex Cover}.  problem: Given a simple polygon $\po$, find a minimum size guard set in $\po$. 
The concept of a {\sc Witness Set} is intricately related to several well-known problems in computational geometry, as extensively explored in the literature. These include the {\sc Art Gallery}  \cite{DBLP:journals/ijcga/AmitMP10}, {\sc Hidden Set}  \cite{DBLP:conf/focs/BrowneKMP23}, and  {\sc Convex Cover}  \cite{DBLP:conf/focs/Abrahamsen21}. In {\sc Art Gallery} problem, we are given a simple polygon $\po$ and our task is to find a minimum size point set ({\em guard set}) $G$  in $\po$ such that every point $p \in \po$ is visible from at least one guard $g \in G$.   Given a simple polygon $\po$, the {\sc Convex Cover}  problem seeks to cover $\po$ with the fewest
convex polygons ({\em convex cover set}) that lie within $\po$ whereas the {\sc Hidden Set} problem asks to place maximum number of points ({\em hidden set}) in $\po$ such that no pair of them is visible to each other.  Let $\mathtt{ws}(\po), \mathtt{gs}(\po), \mathtt{cc}(\po), \mathtt{hs}(\po)$ denote the size of a maximum witness set, minimum guard set,  minimum convex cover set, and maximum hidden set, respectively. It is easy to observe that since no two witnesses can be visible from a single point of $\po$, the cardinality of a maximum witness set is a lower bound on the minimum guard set in a polygon. Similarly, as two witnesses cannot be visible from each other, the cardinality of a maximum witness set is a lower bound on the maximum hidden set in a polygon. At the same time, since there is at
most one hidden point within any convex subset of $\po$, we must have the following basic inequalities:
 $\mathtt{ws}(\po) \leq  \mathtt{gs}(\po)$ as well as $\mathtt{ws}(\po) \leq  \mathtt{hs}(\po) \leq  \mathtt{cc}(\po)$.

\begin{tcolorbox}[enhanced,title={\color{black} {\sc Witness Set} ({\sc WS})}, colback=white, boxrule=0.4pt,
	attach boxed title to top center={xshift=-5.5cm, yshift*=-3.5mm},
	boxed title style={size=small,frame hidden,colback=white}]
	
	\textbf{Input:} A  polygon $ \po$.\\
	\textbf{Task:}  \hspace*{1mm} Find a maximum sized  set $ W$ of points in $ \po $ such that for every pair of points  in $ W $ their  visibility regions inside $ \po $ do not intersect.    
	
\end{tcolorbox}

The \wts~problem was first introduced by Amit et al.~\cite{DBLP:journals/ijcga/AmitMP10} in 2010. Although the authors in \cite{DBLP:journals/ijcga/AmitMP10} inadvertently referred to this problem as \nph, this claim appears to be folklore \footnote{ Joseph S. B. Mitchell strongly believe that this problem should have a polytime solution and our polynomial time solution for {\sc DisWS} indicates an assertion towards that belief.}.  So the computational complexity of this problem still remains open. There are only very few results in the literature on this problem. Recently, in 2019, Daescu et al.~\cite{DBLP:journals/comgeo/DaescuFMPS19} studied this problem on monotone mountains, a proper subclass of monotone polygons, and gave a linear-time algorithm. An $x$-monotone polygon is uni-monotone if one of its two chains is a single horizontal segment. Monotone mountains are uni-monotone polygons in which the segment-chain is not necessarily horizontal. These algorithmic results motivate us to look at the problems in other classes of polygons. We also explore this problem within discrete models, where we have given a set of points $S$ with $\po$ as input, and the goal is to identify a witness set $W \subseteq S$ that maximizes $|W|$.

 %The formal definition is provided below.

% \begin{definition}[$ \gs(\po), \gsp(\po),  \gr(\po)$] 
% 	We say a subset $ \gs(\po) \subseteq \po $ is a set of guards for $ \po $ if  $\po \subseteq \vis(\gs(\po))$. We use \gsp to denote the problem of finding a minimum size subset $\gs(\po)$ in $ \po $. We define   $ \gr(\po) $ to denote the size of a solution of   \gsp.
% \end{definition}	

    	\begin{tcolorbox}[enhanced,title={\color{black} {\sc Discrete Witness Set} ({\sc DisWS})}, colback=white, boxrule=0.4pt,
	attach boxed title to top center={xshift=-4.5cm, yshift*=-3.5mm},
	boxed title style={size=small,frame hidden,colback=white}]

	\textbf{Input:} A  polygon $ \po $, and a finite point set $ S \subseteq  \po$.\\
	\textbf{Task:}  \hspace*{1mm} Find a maximum sized $ W\subseteq S $ such that for every pair of points  in $ W$ their  visibility regions inside $ \po $ do not intersect. 
	
\end{tcolorbox}

     In this work, we mainly focus on two kinds of problems.

     \begin{mdframed}[backgroundcolor=gray!10,topline=false,bottomline=false,leftline=false,rightline=false] 
 \centering
 \textbf{Question 1.}  Can we solve the {\sc Discrete Witness Set} in polynomial  time? 
\end{mdframed}

We address Question 1 specifically for both simple polygons and monotone polygons and provide a positive answer. Following this, having solved Question 1, we proceed to explore a broader issue (discretization).

 \begin{mdframed}[backgroundcolor=gray!10,topline=false,bottomline=false,leftline=false,rightline=false] 
 \centering
 \textbf{Question 2.}  Given a polygon $\po$, can we obtain a finite point  set $S \subseteq \po$ in polynomial time  such that there exists a solution $W$ of the {\sc Witness Set} in $\po$ satisfying $W \subseteq S$?
\end{mdframed}

We focus primarily on Question 2 concerning monotone polygons and provide a positive answer. We hypothesize that this type of discretization is also applicable to a simple polygon. Consider a polygon $\po$. If (i) we resolve Question 2 in time $g(|S|,n)$, and (ii) we address Question 1 in time $f(|S|,n)$, then it follows that we can solve \wts~in a total time of $f(|S|,n)+ g(|S|,n)$. This forms the main theme of our work in this paper.

\paragraph{Motivations.} 
We use  \gsp~to denote the problem of finding a minimum-sized guard set in a given polygon $ \po $.
 There exists an algorithm for \gsp~given by Efrat and Har{-}Peled \cite{DBLP:journals/ipl/EfratH06} which is attributed to Micha Sharir. This algorithm uses heavy tools from algebraic geometry. Due to the recent breakthrough result by Abrahamsen et al.~\cite{DBLP:journals/jacm/AbrahamsenAM22} that the decision version of \gsp~is $\exists \mathbb{R}$-complete, it is clear that discretization (similar to one that we do in our work) for {the} \ gsp~might be possible to solve {the} \gsp. As quoted in \cite{DBLP:journals/jacm/AbrahamsenAM22}, As a corollary of our construction, we prove that for any real algebraic number $\alpha$, there is an instance of the {\sc Art Gallery} problem where one of the coordinates of the guards equals $\alpha$ in any guard set of minimum cardinality. That rules out many natural geometric approaches to the problem, as it shows that any approach based on constructing a finite set of candidate points for placing guards has to include points with coordinates as roots of polynomials of arbitrary degree. Considering \gsp,  Bonnet and Miltzow \cite{DBLP:conf/compgeom/BonnetM17} gave an $\OO(\log (\mathtt{opt}))$ factor approximate algorithm where $\mathtt{opt}$ is the size of the optimum solution. Their in-depth analysis of the problem led to the formulation of the necessary assumptions for polynomial time approximations for \gsp. 
 Furthermore, Bonnet and Miltzow \cite{DBLP:conf/compgeom/BonnetM17} have pointed out that \gsp is challenging even in highly restricted settings, such as when a polygon requires only a constant number of guards. Indeed, the problem is known to be difficult for as few as two guards \cite{Belleville91}. We hope that our characterizations for solving \wts~in monotone polygons would be able to bring new insights for \gsp.  Recently, Hengeveld and Miltzow \cite{DBLP:conf/compgeom/HengeveldM21} considered the problem of computing practical algorithms for computing \gsp~, and our characterizations may help devise better practical heuristic algorithms.  
  We hope that if we look for \gsp~through the lens of our algorithm for  \wts~in monotone polygons, then we may be able to devise approximation algorithms that need fewer assumptions than the one by Bonnet and Miltzow \cite{DBLP:conf/compgeom/BonnetM17} for guarding simple polygons. \gsp~has been well-studied in literature.  
  For details on early literature, see the classical books \cite{ghosh2007vis,ORourke87};
 for the recent literature one can look 
\cite{DBLP:journals/jacm/AbrahamsenAM22,DBLP:conf/compgeom/BonnetM17,DBLP:journals/comgeo/DaescuFMPS19,DBLP:conf/compgeom/HengeveldM21}. To the best of our knowledge, it is not known if one can avoid using the tools in \cite{DBLP:conf/compgeom/BonnetM17}  for obtaining an optimal solution of \gsp~using discretization even for the case of monotone polygons, although constant factor approximation for  \gsp~is known for monotone polygons due to Krohn and Nilsson \cite{DBLP:journals/algorithmica/KrohnN13}. To be more precise, is \gsp~$\exists \mathbb{R}$-complete? We anticipate that our geometric characterizations will prove useful in polygon guarding scenarios, such as smooth or robust guarding \cite{DBLP:conf/compgeom/DasFKM24, DBLP:journals/corr/abs-1811-01177}.

    %\sasanka{Gemini says to write the above sentence like "We anticipate that our geometric characterizations will prove useful in polygon guarding scenarios, such as smooth or robust guarding \cite{DBLP:conf/compgeom/DasFKM24, DBLP:journals/corr/abs-1811-01177}."}

    \paragraph{Our Results and Methods:}

	Our contributions are as follows.

    \begin{description}
        \item[\Cref{sec-outerstring}.] We consider the {\sc DisWS} problem in simple polygons. We first prove that the visibility intersection graph of a simple polygon is an outerstring graph. Since a maximum independent set in outerstring graphs can be determined in polynomial time, as established by \cite{DBLP:journals/comgeo/KeilMPV17}, we derive the subsequent conclusion.

          \begin{restatable}{theorem}{theo}\label{theo:disouter}
    {\sc Discrete Witness Set} is solvable in $\mathcal{O}(|S|^3 \cdot  |V(\po)|^3)$ time.
\end{restatable}

 \item[\Cref{sec-CocomparableGraph}.] We consider the {\sc DisWS} problem specifically for a monotone polygon. Initially, we show that the visibility intersection graph of a monotone polygon is co-comparable. Given that a maximum independent set can be solved in polynomial time for co-comparable graphs, we derive the following result.

  \begin{restatable}{theorem}{theoo}\label{theo-finite-witness-in-M}
    {\sc Discrete Witness Set} in a monotone polygon is solvable in $\OO(|S|^2 + |V(\po)||S|)$ time. 
\end{restatable}

\item[\Cref{sec-Discretization}.] We consider the \wts~problem in a monotone polygon with a suitable parameter. It is known that the {\sc Art Gallery} problem is W[1]-hard when parametrized by the number of guards, and Giannopoulos \cite{giannopoulos2016open} posed the open question of whether the {\sc Art Gallery} admits a fixed-parameter tractable (\fpt) algorithm when parameterized by $r$. Recently, Agrawal et al.~\cite{DBLP:journals/dcg/AgrawalKLSZ24} studied a variation called the {\sc Vertex-Vertex Art Gallery} and gave an \fpt algorithm with running time $r^{\OO(r^2)}n^{\OO(1)}$. Inspired by this, we explore the \wts~problem in a monotone polygon with $r$, the number of reflex vertices, as the parameter. 

Although our objective is to design an algorithm with a runtime of $f(r) n^{\OO(1)}$ for some computable function $f$, we have obtained an algorithm with an even better running time that is $r^{\OO(k)} n^{\OO(1)}$, where $k$ represents the size of the maximum witness set. This outcome shows that we can find the maximum witness set for a monotone polygon in polynomial time $(n^{\OO(1)})$ whenever $k \log(r)= \OO(\log n)$. As $k$ is at most $r$, the algorithm achieves a runtime of $f(r) n^{\OO(1)}$, for some computable function $f$.

   \begin{restatable}{theorem}{thfpt}\label{exact-algorithm}
    {\sc Witness Set} in a monotone polygon $ \mo $ with $n$ vertices is solvable in  $r^{\OO(k)} \cdot n^{\OO(1)}$ time,  where $r $ denotes the number of reflex vertices in $ \mo $ and $ k$ is the size of a maximum witness set in $\mo$.
\end{restatable}

The proof of \Cref{exact-algorithm} consists of two components. First, we try to find a discrete set of size $r^{\OO(k)}n^{\OO(1)}$ inside a monotone polygon $\mo$, that will suffice to contain a maximum witness set of  $\mo$. We were able to achieve this with the help of reflex vertices and a repeating {\em line arrangement} procedure $\OO(k)$ many times. The term "line arrangement" refers to connecting line segments between reflex vertices and a newly established set of points on $\mo$. Secondly, once we ascertain this discrete set, we apply \Cref{theo-finite-witness-in-M} to derive \Cref{exact-algorithm}.

\item[\Cref{sec-ptas}.] Using a clever {\em vertical decomposition} technique, we design   $(1+\epsilon)$-factor approximation algorithm ({\sf PTAS}) for the {\sc Witness Set} problem in monotone polygon running in time   $r^{\OO(1/\epsilon)} \cdot n^2$ time, where $ r$ and $n $ denotes the number of reflex vertices and vertices, respectively of input monotone polygon.

	\end{description}

\section{Preliminaries}\label{sec-preli}

\noindent{\bf Polygons.} In this paper, we only consider simple polygons. A {\em simple polygon} $\po$ is a flat shape consisting of $n$ straight, {\em non-intersecting} line segments that are joined pairwise to form a closed path. The line segments that make up a polygon, called {\em edges}, meet only at their endpoints, called {\em vertices}.
A simple polygon $\po$ encloses a region, called its {\em interior}, that has a measurable area. We do not consider the boundary of $\po$ as part of its interior. We use $ \bd(\po), ~\inte(\po), \ex(\po) $ to denote the boundary, interior, and exterior region of $\po$, respectively.  A vertex $v \in V(\po)$ is a {\em reflex}  vertex if the interior angle of $\po$ at $v$ is larger than 180 degrees, else it is called a {\em convex} vertex.   A polygon $\po$ in the plane is called {\em monotone} with respect to a straight line $L$, if every line orthogonal to $L$ intersects the boundary of $\po$ at most twice. Unless otherwise stated, we refer to a polygon that is monotone with respect to the line $y=0$ (or, the $x$-axis) as a monotone polygon. For any point $p$ in the plane,  $x(p)$ denotes the $x$-coordinate of $p$.

\ifthenelse{\boolean{shortver}}{}{
\begin{figure}[ht!]
    \centering
    \begin{subfigure}[b]{0.35\textwidth}
        \centering
        \includegraphics[width=\textwidth]{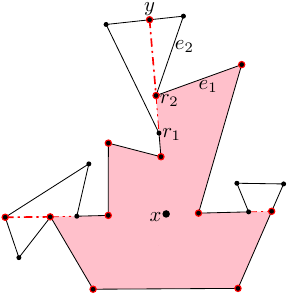}
        \subcaption{$x$ sees both reflex vertices $r_1, r_2$.}
        \label{polygonal-arm1}
    \end{subfigure}
    \hspace{8mm}
    \begin{subfigure}[b]{0.35\textwidth}
        \centering
        \includegraphics[width=\textwidth]{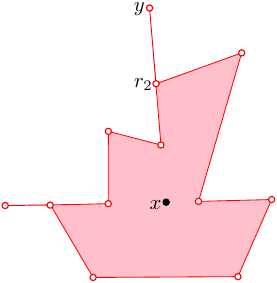}
        \subcaption{ $\overline{r_2y}$ is a \em{polygonal arm}.}
        \label{polygonal-arm2}
    \end{subfigure}
    \caption{An   instance where the visibility region of a point $x$  is not a  simple polygon.} 
    \label{visreg}
\end{figure}
}

\ifthenelse{\boolean{shortver}}{}{
We now motivate the definition of \textit{polygonal arms} of a visibility region $\vis(x)$. We will show an instance where the visibility region of a point $x$ in a simple polygon $\po$ might not be a simple polygon. In the following paragraph, we try to give an example of such a scenario. It is important to note that we are just trying to create an example, i.e., an example of a simple polygon $\po$ and point $x \in \po$ to show the origin of a \textit{polygonal arm} of a visibility region. It is not necessarily true for any arbitrary point $x$ inside a simple polygon $\po.$  So, let $x \in \po$ be a point such that it satisfies the following properties: (i) $x$ sees two distinct reflex vertices of $\po, ~r_1$ and $r_2$; (ii) the pairs of edges containing $r_1$ and $r_2$ lies on either side of the line joining $x, ~r_1$ and $r_2$ (See \Cref{visreg} for an illustration). Furthermore, $r_1$ and $r_2$ are the two closest reflex vertices respectively from $x$ on the line $\overline{xr_1r_2}$; and,(iii) Let the pair of edges that contain the reflex vertex $r_2$ be $e_1$ and $e_2$ respectively. Then neither $e_1$ nor $e_2$ is part of the extended line $\overline{xr_1r_2}$. Now suppose the line $\overline{xr_1r_2}$ hit the boundary of $\po$ on the point $y$ (See \Cref{polygonal-arm1}). Then, the segment $\overline{r_2y}$ is not a part of any polygon in $\vis(x)$. However, this line $\overline{r_2y}$ is indeed visible from $x$ and hence a part of $\vis(x)$. So, we observe that $\vis(x)$ doesn't always need to be a simple polygon in $\po$ for any $x \in \po$. Therefore, the above scenario leads us to the following definition. For a simple polygon $\po$ and a point $x \in \po$, if a line segment or an edge is visible from $x$ such that it is not a part of the simple polygon inside $\vis(x)$, we call that line segment/edge a \textit{polygonal arm} of $\vis(x)$. In \Cref{polygonal-arm2}, the line segment/edge $\overline{r_2y}$ is a \textit{polygonal arm} of $\vis(x)$.
}

%A {\em convex polygon} $\po$ is a simple polygon such that for every two points $p$ and $q$ on the boundary (or interior) of $\po$, no point of the line segment $\overline{pq}$ is strictly outside $\po$. In a convex polygon, all interior angles are less than or equal to 180 degrees, while in a strictly convex polygon all interior angles are less than 180 degrees. Given a non-convex polygon $\po=(V,E)$, we suppose w.l.o.g.~that $1\in V$ is a reflex~vertex.

\medskip
\noindent{\bf Visibility.} Let $\po=(V,E)$ be a simple polygon. We say that a point $p$ {\em sees} (or is {\em visible} to) a point $q$ if every point of the line segment $\overline{pq}$ belongs to  $\po$ (including the boundary). More generally, a set of points $S$ {\em sees} a set of points $Q$ if every point in $Q$ is seen by at least one point in $S$. Note that if a point $p$ sees a point $q$, then the point $q$ sees the point $p$ as well. Moreover, a vertex $v\in V$ necessarily sees itself and its two neighbors in $V(\po)$. For a point $x \in \po$ we use $\vis(x)$, to denote the visibility region for the point $ x $, defined by $\vis(x) = \{ q \in \po~|~q$ is visible from $x\}$. For a subset  $A \subset \po$, we define $\vis(A) \coloneqq \bigcup_{x \in A} \vis(x)$. For a set of points $ F $ in $ \po $, the \textit{visibility intersection graph} corresponding to $ F $ is denoted by $\vig(F)$ and defined as follows: the points in $F$ correspond to the  vertices of the graph, and there is an edge between a pair 
		of vertices $a,b \in F$ in $\vig(F)$ if and only if $\vis(a) \cap \vis(b) \neq \emptyset$.

\medskip 

\noindent Throughout the paper, we use $ \mo $ to denote an $ x$-monotone polygon.  The notation $\widetilde{G}$  denotes the complement graph of $G$, i.e.,  $V(\widetilde{G}) = V(G)$ and $uv \in E(\widetilde{G})$ if and only  $uv \notin E(G)$.

\section{{\sc Discrete Witness Set} in a Simple Polygon}\label{sec-outerstring}

This section aims to demonstrate that the {\sc Discrete Witness Set} within a simple polygon can be solved in polynomial time (\Cref{theo:disouter}). Initially, we prove that the visibility intersection graph in a simple polygon is an outerstring graph  \ifthenelse{\boolean{shortver}}{}{(\Cref{sigouterstring})}. With the use of the known result \ifthenelse{\boolean{shortver}}{}{(\Cref{prop:miso})} by Keil et al.~\cite{DBLP:journals/comgeo/KeilMPV17} that a maximum independent set can be found in polynomial time (in $|N|$) for outerstring graphs that is given with an intersection model consisting of polygonal arcs with a total of $|N|$ segments, we achieve our objective, which is the following.

\theo*

\ifthenelse{\boolean{shortver}}{}{
 We start with some basic definitions that we need.

\begin{definition}[Primary edge and Window]\label{def:pri}
	{\em Consider a polygon $\po$ and a point  $x \in \po$.  We call an edge of $\vis(x)$ a {\em primary edge} if that edge coincides entirely with the boundary of $\po$, otherwise we call the edge a \textit{window} of  $x$.  Among the two end-points of a window, the one that is closest to $x$ is called a \textit{base} and the other is called an \textit{end}.}
\end{definition} 

\noindent   For a point $x\in \po$, we define $\min_{pr}(x)$ to be the length of the shortest length primary edge(s) in $\vis(x)$. See  \Cref{primaryyyy} for an illustration of \Cref{def:pri}, where notice that  $e_3$ is a {\em polygonal arm} and $e_4, e_5$ have some parts which are not part of $\bd(\po)$. Observe that a polygonal arm is also a window of $x$. We say a point $\alpha$ is collinear to a line segment $l$, if any two points of $l$ and $\alpha$ are collinear.  i.e., if we extend $l$ then it passes through $\alpha$. It is easy to observe that every window of point $x$ is collinear with $ x $ (also mentioned in \cite{DBLP:journals/comgeo/BoseLM02}). However, a connection exists between a primary edge and a polygonal arm.

\begin{figure}[ht!]
	\centering
	\begin{subfigure}[b]{0.35\textwidth}
		\centering
		\includegraphics[width=\textwidth]{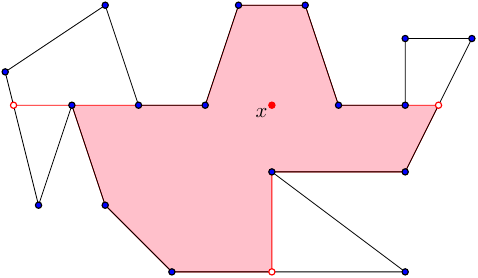}
		\subcaption{}
	\end{subfigure}
	\hspace{8mm}
	\begin{subfigure}[b]{0.35\textwidth}
		\centering
		\includegraphics[width=\textwidth]{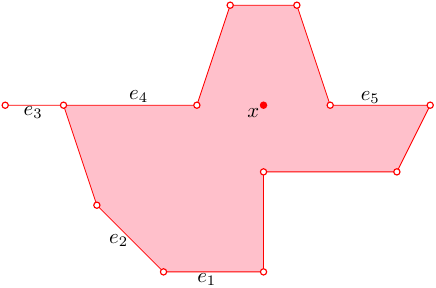}
		\subcaption{}
		\label{primaryyyy}
	\end{subfigure}
	\caption{Example of primary edges ($e_1, e_2$) in $\vis(x)$ and windows ($e_3, e_4, e_5$) of $x$.}
	\label{primaryedgeexample}
\end{figure}

\begin{observation}\label{primary_not_arm}
	For a point $x\in \po$, $\vis(x)$ contains a primary edge, i.e., not all edges of $\vis(x)$ can be windows.
\end{observation}
%\begin{figure}[ht!]
%	\centering
%	\includegraphics[width=0.3\linewidth]{fig/colin2.pdf}
%	\caption{$x$ is collinear to all the edges of $\vis(x)$.}
%	\label{colin22}
%\end{figure}

\ifthenelse{\boolean{shortver}}{}{
\begin{proof}
	Clearly $\vis(x)$ consists of a polygon (say $\po_x$) and maybe some polygonal arms (see \Cref{visreg}). $\po_x$ is a connected region with nonempty interior bounded by line segments in $\po$. So, it must contain at least three edges that, when extended, do not pass through the same point. But, since any window of $x$ is collinear to $x$, we must have an edge that is not a window, that is, a primary edge. So we can conclude that $\po_x$ must have a primary edge and no edge of $\po_x$ is a polygonal arm of $\vis(x)$.
\end{proof}}

% \begin{definition}[Visibility Intersection Graph]
	
	% 	{\em For a set of points $ F $ in $ \po $, the \textit{visibility intersection graph} corresponding to $ F $ is denoted by $\vig(F)$ and defined as follows: the points in $F$ correspond to the  vertices of the graph and there is an edge between a pair 
		% 	of vertices $a,b \in F$ in $\vig(F)$ if and only if $\vis(a) \cap \vis(b) \neq \emptyset$.}
	
	% \end{definition}

%\begin{definition}[$\bm{\min_{pr}L(w)}$]
%	{\em For a point $w\in \po$, we define $\min_{pr}(w)$ to be the length of the shortest length primary edge(s) in $\vis(w)$.}
%\end{definition}
\begin{definition}[String, String-like]
	{\em A {\em string} is a polygonal curve that does not cross itself. A {\em string-like} structure refers to a combination of strings $\mathcal{S}$ that satisfies a certain condition. Specifically, there exists a common string $\lambda \in \mathcal{S}$ such that none of the strings in the set $\mathcal{S} \setminus \{\lambda\}$ intersect with one another, and each string in $\mathcal{S} \setminus \{\lambda\}$ has exactly one endpoint in $\lambda$. In other words, by \textit{string-like}, we mean a structure that is a string or a string with branches.}  
\end{definition}

\paragraph{String-like $\boldsymbol{(\str^*(w)})$ Creation on $\boldsymbol{ \vis(w)}$.} For every input point $w \in \po$, our goal is to create a unique  string denoted by $ \str(w) $. To that end, we first create two structures consecutively associated with the point $w$, which are string-likes and we denote as $ \str^*(w) $. Let us first look at the visibility region $ \vis(w) $ and the topological boundary of the region, denoted by $ \bd(\vis(w)) $. We arbitrarily choose a primary edge, say $e$, of $\vis(w)$ and then we break $ \bd(\vis(w)) $ on $ e $ by cutting a very tiny open part of length $\epsilon$ where $0<\epsilon<\min_{pr}(w)$, so that it forms a  string-like $\str^*(w)$ corresponding to $ w $, starting at a point, say $v$, of $\bd(\vis(w))$ (see \Cref{viswist}). We call $v$  the origin of $\str^*(w)$, and by $\epsilon$-gap we mean the gap created by cutting that open segment of length $\epsilon$. We make a restriction here by not choosing any vertex of $\vis(w)$ as the origin of $\str^*(w)$, we can do this because the length of any primary edge is strictly greater than $\epsilon$. Also, for creating $\epsilon$-gap, we do not choose any polygonal arm of $\vis(w)$ to cut from, as the existence of a primary edge, which is not a polygonal arm, is ensured by \Cref{primary_not_arm}.

\begin{figure}[!ht]
	\centering
	\includegraphics[width=0.3\linewidth]{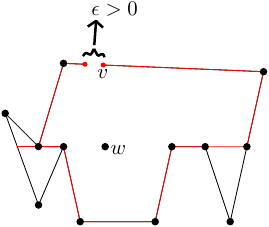}
	\caption{$\str^*(w)$ (red) is formed from $ \bd(\vis(w)) $ by cutting   at a primary edge by a tiny length $\epsilon$. }
	\label{viswist}
\end{figure}

\noindent Consider that we are given a simple polygon $ \po $ and a set $F=\{w_1,w_2, \cdots ,w_m\}$ of $ m $ points in $ \po $. Our main goal is to  construct  a unique  string corresponding to  each point in $ F $  such that the following property is satisfied: for each pair of points $ x,y \in F $, $$\vis(x) \cap \vis(y) \neq \emptyset~\text{if and only if} ~  \str(x) \cap \str(y) \neq \emptyset$$ 
To this end, we follow the three-step procedure described below.
\begin{description}
	\item[Step 1. (String-like formation):] 
	Our construction requires that each string-like originates at a unique point of $\bd(\po)$. In the worst-case scenario, we might have that all the string-likes originate at one particular common primary edge. So we want to choose a common $\epsilon$ in such a way that we have enough space to create multiple $\epsilon$-gaps from one various primary edge. Therefore, we choose $\epsilon$ so that $0 < \epsilon<\frac{1}{m}\min_{w\in F}(min_{pr}(w))$.
	
	\begin{figure}[ht!]
		\centering
		\begin{subfigure}[b]{0.3\textwidth}
			\centering
			\includegraphics[width=\textwidth]{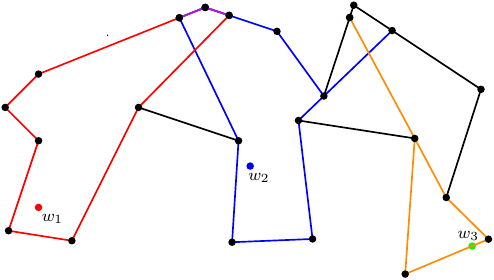}
			\subcaption{$\vis(w_i), i \in [3]$}
			\label{str1}
		\end{subfigure}
		\hspace{8mm}
		\begin{subfigure}[b]{0.5\textwidth}
			\centering
			\includegraphics[width=\textwidth]{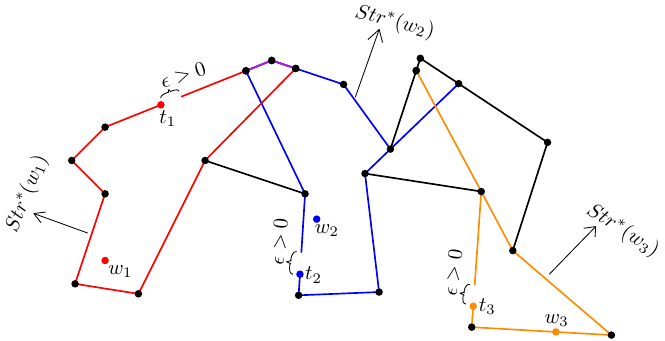}
			\subcaption{$\str^*(w_1), \str^*(w_2)$ and $\str^*(w_3)$.}
			\label{str2}
		\end{subfigure}
		\caption{Example of string-like  formations for $w_1, w_2, w_3$ from their visibility regions..}
		\label{str*formation}
	\end{figure}
	
	Now for each $w_i\in F$, we choose a primary edge of $\vis(w_i)$ and construct the string-like, $\str^*(w_i)$ by creating an $\epsilon$-gap on that edge (See \Cref{str*formation}). 
	The choice of $\epsilon$ ensures that it is possible to create all the string-likes together without overlapping any two $\epsilon$-gaps as $m\epsilon<\min_{w_{j}\in F}(min_{pr}(w_j))<min_{pr}(w_i)<$ \textit{length of any primary edge of} $\vis(w_i)$. We denote the  origin of $\str^*(w_i)$ by $t_i$. The following observation is immediate.
	
\end{description}

\begin{observation} \label{obs1stp1}
	For each $w \in F$ the length of any primary edge of $\vis(w)$ is more than $\epsilon.$
\end{observation}

Now define $S_i \coloneqq $ Set of all vertices of $\vis(w_i)$, $S=\bigcup_{i=1}^mS_i$, $T=\{t_1,...,t_m\}$. Let $d_{t_i}=\min_{s\in S}(d(t_i,s))$ and $d=\min_{t_i\in T}d_{t_i}$, where $d(x,y)$ denotes the length of  $\overline{xy}$. 
Since the choice of $\epsilon$ is in our hands, we choose $\epsilon$ such that the value of $\epsilon$ is even smaller than $d$. Let $0<\epsilon<\min\{d,\frac{1}{m}\min_{w\in F}(min_{pr}(w))\}$. The above restriction ensures that for any two points $w_i, w_j\in F$, the distance between the  origin $t_i$ of $\str^*(w_i)$ and any vertex of $\vis(w_j)$ is at least $\epsilon$. Thus, we have the following observation.
\begin{observation}\label{obs-novertexinsideepsilon}
	For any pair of  points, $w_i, w_j\in F$, $\vis(w_j)$ has no vertex which  is  in the $\epsilon$-gap beside the origin of $\str^*(w_i)$ on $\bd(\vis(w_i))$.
\end{observation}

Now we prove the following claim.

\begin{clm}\label{3rd_equivalance}
	%There is an edge between $\vis(w_j)$ and  $\vis(w_i)$ in $\vig(F)$ if and only if there is an edge between $\str(w_{i})$ and $\str(w_{j})$ in $\sig(Z)$. 
	For  $w_i,w_j\in F$, we have $\vis(w_i)\cap \vis(w_j)\neq \emptyset$ if and only if $\str^*(w_{i}) \cap \str^*(w_{j}) \neq \emptyset$.
\end{clm}
\begin{proof}
	In the forward direction, we show that $\vis(w_i) \cap \vis(w_j) \neq \emptyset$ implies $\str^*(w_{i}) \cap \str^*(w_{j}) \neq \emptyset$. 
	Let $\vis(w_i)\cap \vis(w_j)\neq \emptyset.$ Consider the case where $\str^*(w_{i}) \cap \str^*(w_{j}) = \emptyset.$
	It is easy to observe that $\vis(w_i)\cap \vis(w_j)\neq \emptyset$ implies $\bd(\vis(w_i))\cap\bd(\vis(w_j))\neq \emptyset.$ Now in this case, boundaries of $\vis(w_i)$ and $\vis(w_j)$ intersect but their corresponding string-likes do not. The source of this conflict must lie in the $\epsilon$-gaps introduced into the visibility regions during the construction of the string-likes. This implies that the boundaries of these regions can only intersect within these specific gaps. This can only happen when there exists a primary edge or a vertex of $\vis(w_j)$ strictly contained in the $\epsilon$-gap beside the origin of $\str^*(w_i)$ on $\bd(\vis(w_i))$ or, there exists a primary edge or a vertex of $\vis(w_i)$ strictly contained in the $\epsilon$-gap beside the origin of $\str^*(w_j)$ on $\bd(\vis(w_j))$. Now, there are only a finite number of such $(i,j)$ pairs. If there is a primary edge contained in the $\epsilon$-gap, then the length of that primary edge is less than $\epsilon$, which contradicts \Cref{obs1stp1}. If there is a vertex inside the $\epsilon$-gap, then it contradicts \Cref{obs-novertexinsideepsilon}. Therefore, $\str^*(w_{i}) \cap \str^*(w_{j}) \neq \emptyset$.
	
	In the reverse, let us assume for contradiction that $\vis(w_i) \cap \vis(w_j) = \emptyset$ however, $\str^*(w_{i}) \cap \str^*(w_{j}) \neq \emptyset$. Since these {\em string-likes} are the boundaries of the visibility regions with tiny parts of their edges removed, therefore if $\str^*(w_{i}) \cap \str^*(w_{j}) \neq \emptyset$ then there exists at least one point in the boundaries of $\vis(w_j)$ and  $\vis(w_i)$ where they intersect. Therefore, $\vis(w_i) \cap \vis(w_j) \neq \emptyset$, a contradiction.  This completes the proof.
\end{proof}

%\begin{observation}\label{0st_equivalance}
%$\vis(w_{i}) \cap \vis(w_{j}) \neq \emptyset$  if and only if $\str^*(w_{i}) \cap \str^*(w_{j}) \neq \emptyset$.
%\end{observation}
%Our algorithm starts with $w_1$, we first choose a primary edge of $\vis(w_1)$ and cut an open segment of length $\epsilon$ from it. That creates a {\em string-like}, originated at $s_1\in \bd(\vis(w_1))$ which is not a vertex of $\vis(w_1)$. Now at $i$-th step, if there exists a primary edge of $\vis(w_i)$ such that it has not been used before i.e., no $\epsilon${\em -gap} has not been made before from that edge till $i-1$-th step then we construct $\str^*(w_i)$ by cutting an open segment of length $\epsilon$ from that unused primary edge. And if there is no such unused primary edge left in $\vis(w_i)$ then we choose any primary edge of $\vis(w_i)$ and cut a line segment of length $\epsilon$ after adjusting the previous gaps made before on that edge. We can do the adjusting because all the $\epsilon${\em -gap}s can be made from any primary edge of any $\vis(w), w\in F$ as $m\epsilon<\min_{w\in F}(min_{pr}(w))<min_{pr}(w)$ for any $w\in F.$ That way we have created $\str^*(w_i)$ and let it is originated at $s_i.$

\begin{description}
	
	\item[Step 2. (Inflation of $\po$):]
Once the construction of the string-likes is done, we scale up (inflate)  the polygon $\po$ by a tiny factor, say $\delta$, without altering the string-likes. We call this modified polygon $\infp$. {The formal description of inflating the polygon $\po$ is given below.}
            \begin{itemize}
                \item {For each edge $e$ in $\po$, draw a perpendicular to the outside of $\po$.}
                \item  {On each of those perpendiculars drawn on an edge $e$ (say) of $\po$, draw another line $\hat{e}$ parallel to that corresponding edge, at a predetermined tiny distance, $\delta$.}
                \item {This set of new edges, $\{\hat{e} : e ~\text{is an edge in}~ \po \}$, are the edges of the inflated polygon, $\infp$, and the points where any two adjacent such edges meet are the vertices of $\infp$.}
            \end{itemize}
	Now, in $\infp$, for each $w_i\in F $, we modify $\str^*(w_i)$ by joining its  origin $t_i$ with the nearest point of $t_i$ on $\bd(\infp)$. Note that, by our construction, $t_i$ is not any vertex of $\po$, so $t_i$ must have a unique nearest point in $\bd(\infp)$ (the unique nearest point is the intersection point of the perpendicular on $t_i$, and $\bd(\infp)$). We denote the modified string-like by $\str^{**}(w)$ and the point on $\bd(\infp)$ connected to $t_i$ by $s_i$. Hereon, we call $s_i$ as the  origin of the string-like $\str^{**}(w_i).$
	\begin{figure}[ht!]
		\centering
		\includegraphics[width=0.5\linewidth]{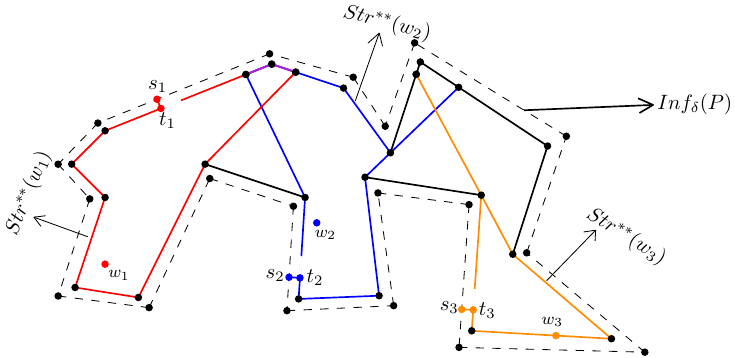}
		\caption{Inflation of $\po$ and $\str^{**}(w)$ formation.}
		\label{str**form}
	\end{figure}
	Observe that the procedure of modifying $\po$ presented above preserves the intersection of string-likes in $\po$, since the string-likes are not perturbed and only $\po$ is modified. The strings that initially intersected in $\po$ still intersected after inflating $\po$. The same goes for the string-likes that do not intersect each other. So we have the following observation, which is immediate.
	
\end{description}

\begin{observation}\label{1st_equivalance}
	$\str^*(w_{i}) \cap \str^*(w_{j}) \neq \emptyset$ if and only if $\str^{**}(w_{i}) \cap \str^{**}(w_{j}) \neq \emptyset$.
\end{observation}

\begin{description}
	
	\item[Step 3. (String formation):]
	In this step, for each $w_i\in F$, we create a  string, $\str(w_i)$ from the  string-like, $\str^{**}(w_i).$ Notice that, after creating the $\epsilon$-gap on $\bd(\vis(w_i))$, two end-points of $\str^*(w_i)$ were created. One of them was considered as the origin, $t_i$ and let the other end-point be denoted by $s_i'$. Then $s_i'$ is also an endpoint of $\str^{**}(w_i).$ Now, in the string-like $\str^{**}(w_i)$, there is a unique string (excluding the branches) joining the end-points $s_i$ and $s_i'$; we call it $\gamma_i$ (see \Cref{gamma2}). Essentially, $\gamma_i \cup \{branches ~of ~\str^{**}(w_i)\} = \str^{**}(w_i)$. These branches are nothing but the polygonal arms of $\vis(w_i)$. Let $l_i^1, l_i^2, \ldots , l_i^k$ be the  branches of $\str^{**}(w_i)$ and $l_i^j$ is joined with $\gamma_i$ at the point $a_i^j$. Let the other endpoint of $l_i^j$ be $b_i^j$. We denote the branch $l_i^j$ by $\overline{a_i^jb_i^j}$ (since a branch is a polygonal arm of a visibility region, it is a line segment).
	\begin{figure}[ht!]
		\centering
		\begin{subfigure}[b]{0.2\textwidth}
			\centering
			\includegraphics[width=\textwidth]{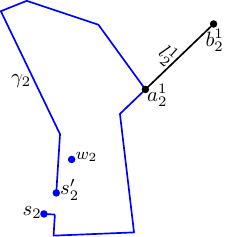}
			\subcaption{$\gamma_2$ (blue string).}
			\label{gamma2}
		\end{subfigure}
		\hspace{8mm}
		\begin{subfigure}[b]{0.5\textwidth}
			\centering
			\includegraphics[width=\textwidth]{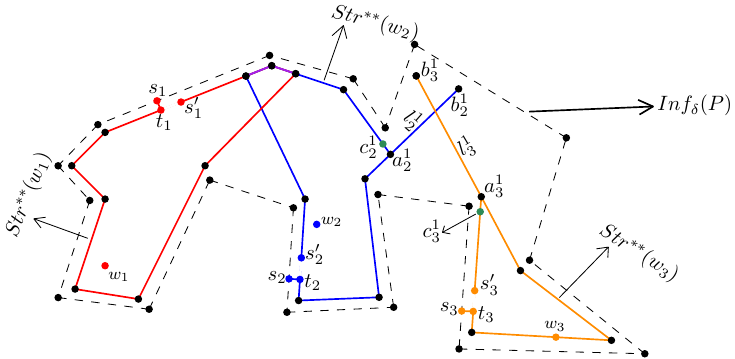}
			\subcaption{$\str^{**}(w_1), \str^{**}(w_2)$ and $\str^{**}(w_3)$.}
			\label{str**withlabeling}
		\end{subfigure}
		\caption{Construction of $\str^{}(w)$ from $\str^{**}(w)$.}
	\end{figure}
	
	Observe that, as there are only a finite number of string-like structures and hence a finite number of branches in total, there exists an open neighborhood $U_i^j$ around $a_i^j$ (the point of intersection of a branch $l_i^j$ and $\gamma_i$) such that for any point $c_i^j \in U_i^j \cap \gamma_i$ (See \Cref{str**withlabeling}), we can replace $\overline{a_i^jc_i^j}$ by $\overline{c_i^jb_i^j}$ (See \Cref{strreplacement}) to create a new string $\str(w_i)$ such that it satisfies the following properties: (i) for all $l$ such that $\str^{**}(w_i) \cap \str^{**}(w_l) \neq \emptyset$ we have $\str(w_i) \cap \str(w_l) \neq \emptyset$; (ii) for all $l$ such that $\str^{**}(w_i) \cap \str^{**}(w_l) = \emptyset$ we have $\str(w_i) \cap \str(w_l) = \emptyset$. We can ensure property (i) because we only have a finite number of string-like structures that intersect with $\str^{**}(w_i)$. So, it cannot happen that for any arbitrarily small neighborhood $U_i^j$, by adding the line $\overline{c_i^jb_i^j}$, it always intersects with a new string-like $\str^{**}(w_{i'})$. Also, we can ensure property (ii) by choosing a small neighborhood $U_i^j$ such that the segment $\overline{a_i^jc_i^j}$ does not intersect with any other string-like structure. In ensuring both properties, the fundamental idea was that the number of string-like structures is finite. We formally describe the construction of $\str(w_i)$ from $\str^{**}(w_i)$ below (see \Cref{strreplacement}). 
	\begin{itemize}
		\item For each $j$, add $\overline{c_i^j b_i^j}$ with $\str^{**}(w_i)$.
		\item For each $j$, delete $\overline{a_i^j c_i^j}$ from $\str^{**}(w_i)$ except the points $a_i^j \text{ and } c_i^j$, i.e., delete the open segment $(a_i^jc_i^j)$.
	\end{itemize}
	
	\begin{figure}[ht!]
		\centering
		\includegraphics[width=0.5\linewidth]{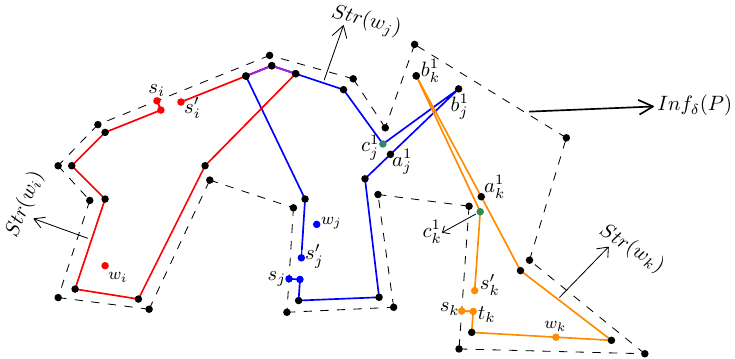}
		\caption{Replacing $\overline{a_i^jc_i^j}$ by $\overline{c_i^jb_i^j}$.}
		\label{strreplacement}
	\end{figure}

	\begin{figure}[ht!]
		\centering
		\includegraphics[width=0.5\linewidth]{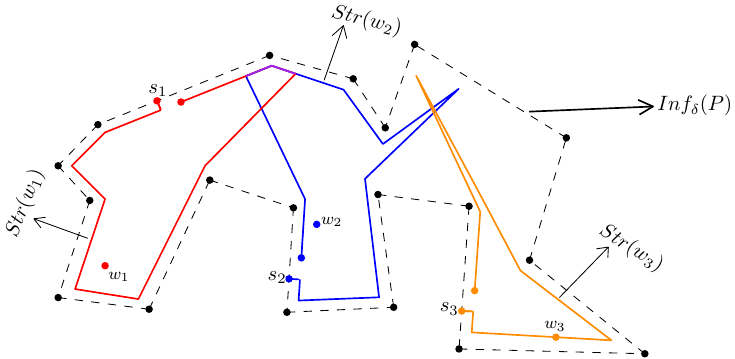}
		\caption{$\str(w_1), \str(w_2)$ and $\str(w_3)$.}
	\end{figure}
	
	Observe that $\str(w_i)$ is a string as it does not have any branches anymore. $s_i$ is the origin of $\str(w_i)$. Since the construction of $\str(w_i)$ for each $w_i\in F$ respects the intersection-invariant property (i.e., two strings intersect if and only if their respective string-likes also intersect, and the procedure described for constructing the strings does not change their respective intersections), we have the following observation.
	
\end{description}

\begin{observation}\label{2nd_equivalance}
	$\str^{**}(w_{i}) \cap \str^{**}(w_{j}) \neq \emptyset$ if and only if $\str(w_{i}) \cap \str(w_{j}) \neq \emptyset$.
\end{observation}

Combining  \Cref{3rd_equivalance}, Observations  \ref{1st_equivalance} and  \ref{2nd_equivalance} we obtain the following.

\begin{lemma}\label{lem:equi}
	Let  $F=\{w_1,w_2, \cdots ,w_m\}$ be a set of $ m $ points in  a simple polygon $ \po $. For each pair of points $ w_i,w_j \in F $, $\vis(w_i) \cap \vis(w_j) \neq \emptyset~\text{if and only if} ~  \str(w_i) \cap \str(w_j) \neq \emptyset$.
\end{lemma}

\begin{definition}[String Intersection Graph]
	{\em For a set $ \mathcal{S} $ of strings, the string intersection graph corresponding to $ \mathcal{S} $ is denoted by $\sig(\mathcal{S})$; and is defined as follows: each string in $\mathcal{S}$ corresponds to a vertex of $\sig(\mathcal{S})$ and there is an edge between a pair of vertices in $\sig(\mathcal{S})$ if and only if their corresponding strings intersect.}
\end{definition}
\medskip
Let $\mathcal{S}=\{\str(w_1), \ldots , \str(w_m)\}$. Notice  that for each point $ w\in F $, there exists a corresponding $\vis(w)$ in $ \po $ and $\str(w)$ in $ \infp $. Also, $\vis(w)$ and $\str(w)$ correspond to a vertex in $\vig(F)$ and $\sig(\mathcal{S})$, respectively. With the help of  \Cref{lem:equi}, we have the following.

\begin{lemma}\label{vigsiglemma}
	The graphs $\vig(F)$ and $\sig(\mathcal{S})$ are isomorphic.
\end{lemma}

Now we show that the string intersection graph corresponding to $ \mathcal{S} $ is an outerstring graph. 
The {\em outerstring graphs} are the intersection graphs of curves in the plane that lie inside a circle such that each curve intersects	the boundary of the circle at one of its endpoints. However, here, instead of a circle, what we have is a simple polygon. Even then, the intersection graph of curves forms an outerstring graph. The reason is that we can translate the underlying polygon to a circle such that the pairwise intersections among the curves are preserved in either case.

\begin{clm}\label{sigouterstring}
	$\sig(\mathcal{S})$ is an outerstring graph.
\end{clm}

\begin{proof}
	Every vertex of $\sig(\mathcal{S})$ corresponds to a point $w \in F$ where $w$ corresponds to a string $\str(w)$. Now the string $\str(w)$ has one endpoint at $ \bd(\po) $ and another endpoint at the interior of $\po$. Therefore, the graph $\sig(Z)$ is an outerstring graph. Hence, the claim.
\end{proof}

Now we show that any independent set on the graph $\sig(\mathcal{S})$  corresponds to a witness set on $\po$ of the same size.

\begin{lemma} \label{indwi}
	Any independent set in $\sig(\mathcal{S})$ gives a witness set in $\po$ with the same size. 
\end{lemma}

\begin{proof}	
	Any independent set in the graph $\sig(\mathcal{S})$ gives us the set of strings in $\po$ that do not intersect each other. Now each string has a corresponding visibility region in $\po$. Since the strings returned by the independent set do not intersect each other, their corresponding visibility regions do not intersect. Hence, the proof follows.
\end{proof}

\begin{proposition}[\cite{DBLP:journals/comgeo/KeilMPV17}]  \label{prop:miso}
	Given the geometric representation $ (\po, S)$ of a weighted outerstring graph $G$,  an independent set with maximum weight for $G$ can be found in $\OO(n^3)$ time, where $n$ is the number of segments
used to represent the strings of $S$ and the polygon $\po$.
\end{proposition}

%\theo*

% Given a set of $ m $ points in a simple polygon having $ n $ vertices, the maximum witness set problem in the polygon  can be solved in $O(m^3 n^3)$ time. 
%\end{theorem}

\begin{proof}[\textbf{Proof of \Cref{theo:disouter}}]
	Given an outerstring graph $G$ and its polygonal representation in a polygon $\po$, we can compute the maximum independent set in $\mathcal{O}(N^3)$ time \cite{DBLP:journals/comgeo/KeilMPV17}, where $N$ is the number of segments used to represent the strings of $G$ and the polygon $\po$. In our problem, each individual string may have $\mathcal{O}(n)$ segments and there are $m$ strings in total, where $n$ is the number of vertices in $\po$ and $m$ is the number of input points.  Hence $N=\OO(mn)$. Now applying  \Cref{indwi}, computing the maximum independent set in the graph $\sig(\mathcal{S})$ returns the maximum witness set in $\po$. Therefore, the maximum witness set problem in a simple polygon $\po$ for a set of $m$ points can be solved in $\mathcal{O}(m^3n^3)$ time. 
\end{proof}}

	\section{{\sc Discrete Witness Set} in a Monotone Polygon} \label{sec-CocomparableGraph} 

This section aims to improve the running time of the result already obtained in the previous section (\Cref{theo:disouter}), for the class of monotone polygons. Initially, we prove that the visibility intersection graph of a monotone polygon is co-comparable \ifthenelse{\boolean{shortver}}{}{(\Cref{lem-cocomparable})}.  Then, with the use of a known result \ifthenelse{\boolean{shortver}}{}{(\Cref{prop:mis})}  that a maximum independent set can be found in polynomial time for co-comparable graphs, we achieve our objective, which is as follows.

\theoo*

\ifthenelse{\boolean{shortver}}{}{

 \begin{definition}[Co-comparable Graph]
     {\em A {\em comparability} graph is an undirected graph that connects pairs of elements that are comparable to each other in a partial order. A graph $G$ is {\em co-comparable} if the complement of the graph $G$ is a comparability graph.} 
 \end{definition}

  \noindent \textbf{Remark.} There is an equivalent definition of a comparability graph, which is as follows.  A comparability graph is a graph that has a transitive orientation, i.e., an assignment of directions to the edges of the graph (i.e. an orientation of the graph) such that the adjacency relation of the resulting directed graph is transitive: whenever there exist directed edges $(x,y)$ and $(y,z)$, there must exist an edge $(x,z)$. 

 \begin{proposition}[\cite{DBLP:journals/dm/McConnellS99}]  \label{prop:mis}
      A maximum independent set in a co-comparable graph with $n$ vertices can be found in $\OO(n^2)$ time. 
 \end{proposition}

  Let $F$ represent a set of $f$ points within a monotone polygon $\mo$. Initially, we construct the visibility intersection graph $\G(F)$ in the following way. The vertices of $\G(F)$ are the points in $F$ and there is an edge between two vertices $a,b \in F$ if and only if $\vis(a) \cap \vis(b) \neq \emptyset$. Subsequently, we demonstrate that the graph $\G(F)$ is co-comparable. Next, we introduce some notations and definitions. Some of which will be used in the subsequent section. 
	
	\begin{figure}[ht!]
		\centering
		\includegraphics[scale=1.1]{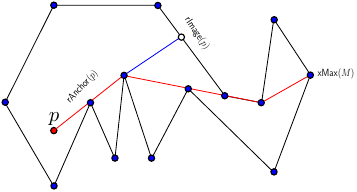}
		\caption{$\spr(p)$ is denoted by the red path. $\Rb(p)$ is the chord between $\RA(p), \IMr(p)$. } \label{fig-Intersect}
	\end{figure}

%	\begin{definition}[Reflex and Convex Vertex] \label{def-reflex}
%		{\em A vertex $ v $ of a polygon is called {\em reflex} when the internal angle at $ v $ exceeds 180 degrees; otherwise, it is referred to as {\em convex}.}
%	\end{definition}
	
	\begin{definition}[$ \XM() $ and $ \XMi() $]
		{\em For a region $ A $, $\XM(A)$ (resp., $\XMi(A)$) denotes a point having maximum (resp., minimum) $x$-coordinate among all the points in  $A$. }
	\end{definition}
    
	\noindent Observe that there may be multiple points in the set $A$ which possess the highest (or lowest) $x$-coordinate. In such situations, we arbitrarily choose $\XM(A)$. A similar rationale applies to $\XMi(A)$. We represent a monotone polygon by $\mo$. For a point $p \in \mo$, a new edge is created when the visibility of $p$ is blocked by a reflex vertex $c$ of $\mo$. The other endpoint of this newly formed edge, which lies on the boundary of $\mo$, is called the \textit{image of $p$ through $\mathit{c}$}. We will also carry forward all the definitions and notations defined in this section to the following sections.

	\begin{definition}[$ \spr() $ and $ \spl() $]\label{def:one}
		{\em For a point $ p \in \mo $, $\spr(p)$ (resp., $\spl(p)$) denotes a  shortest path in $\mo$ from $p$ to   $\XM(\mo)$ (resp., $\XMi(\mo)$).}
	\end{definition}
	
	\begin{definition}[$ \RA() $ and $ \LA() $]\label{def:two}
		{\em For a point $p \in \mo$, we define a vertex  $c$ of $\mo$ as $\RA(p)$  if $c \neq \XM(\mo)$ and $c$ is the leftmost  vertex (having smallest $x$-coordinate)  of $\mo$ other than $ p $ that is in $\spr(p)$.   Similarly,  we define a vertex  $c$ of $\mo$ as $\LA(p)$ if $c \neq \XMi(\mo)$ and $c$ is the rightmost (having largest $x$-coordinate) vertex  of $\mo$ other than $ p $ that is in $\spl(p)$.} 
	\end{definition}
	
	\begin{definition}[$ \IMr() $ and $ \IMl() $]\label{def:thr}
		{\em For a point $p\in \mo$, if there exist $\RA(p)$ and $\LA(p)$, then we call the points, \textit{image of $p$ through $\mathit{\RA(p)}$} and \textit{image of $p$ through $\mathit{\LA(p)}$} as $\IMr(p)$ and $\IMl(p)$, respectively.}
	\end{definition}

      This leads to the following observations.

    \begin{observation}\label{obs:lara}
        \em{For a point $p \in \mo$, if there exists no $\RA(p)$,  then $\XM(\vis(p))=\XM(\mo)$. Similarly, if there exists no $\LA(p)$  then $\XMi(\vis(p))=\XMi(\mo)$.}
    \end{observation}
    
    \begin{observation}\label{obs-existence_of_LA_and_RA}
      {\em Given a monotone polygon $\mo$ and a set of witnesses $w_1, w_2, \ldots, w_k$  with $x(w_1) < x(w_2) < \ldots < x(w_k)$, both $\LA()$ and $\RA()$ are present for each witness $w_i$ where $1<i<k$. However, this may not hold for the extreme points $w_1$ and $w_k$, which are the leftmost and rightmost witnesses, respectively (due to \Cref{obs:lara}).}
    \end{observation}

	\begin{definition}[$ \Rb() $ and $ \Lb() $]\label{def:fr}
		{\em For a point $p \in \mo$, we define $\Rb(p)$  as the chord with endpoints at  $\RA(p)$  and $\IMr(p)$. Similarly,  we define  $\Lb(p)$ as  the chord with endpoints  $\LA(p)$) and  $\IMl(p)$.} 
	\end{definition}

    For an illustration of  Definitions \ref{def:one}, \ref{def:two}, \ref{def:thr}, and \ref{def:fr}  see \Cref{fig-Intersect}. Observe that it might happen that for some point $p$, there is no $ \Rb(p) $ or  $ \Lb(p) $(due to \Cref{obs:lara}).

	\begin{definition}\label{def:chord}
		{\em We define $\ell(p)$ as the vertical chord passing through a point $p \in \mo$.}
	\end{definition}

 The following observation is immediate.

	\begin{observation}  \label{fact-Chord}
		For any point $p \in \mo$, we have $\ell(p) \subset \vis(p)$ . 
	\end{observation}

	\noindent In the following \Cref{lem-cl1} and  \Cref{lem-cl2}, we show that for every pair of points $a,b \in \mo$, we can preprocess visibility regions $\vis(a)$ and $\vis(b)$  during the computation of $\vis(a)$ and $\vis(b)$ that helps us to check if $\vis(a) \cap \vis(b)=\emptyset$ in $\OO(1)$ time. 
	
	% Using the  \Cref{fact-Chord}, we have the following lemmas for a pair of points $a, b \in \mo$ such that $x(a) < x(b)$. So it is clear that $\XM(\vis(a))\leq x(a) < x(b)$.	
	
	\begin{observation} \label{lem-cl1}
		Consider a pair of points $a, b \in \mo$ with  $x(a) < x(b)$. If $x(b) \le x(\XM(\vis(a))$, then $\vis(a) \cap \vis(b) \neq \emptyset$
	\end{observation}

\begin{proof}
    If $x(b) \le x(\XM(\vis(a))$, then $l(b)\cap\vis(a) \neq \emptyset$ and hence, $\vis(b)\cap\vis(a)\neq \emptyset$ as from \Cref{fact-Chord} we have that $\ell(b) \subset \vis(b)$. 
\end{proof}

	\begin{lemma} \label{lem-cl2}
		Consider  a pair of points $a, b \in \mo$ with  $x(a) < x(b)$ and  $x(b)>x(\XM(\vis(a))$. Then $\vis(a)\cap \vis(b)\neq \emptyset$ if and only if $\Rb(a)$ intersects with $\Lb(b)$.
	\end{lemma}
	
	\begin{proof}
            First, let us assume that $\vis(a) \cap \vis(b) \neq \emptyset$. Then, for a pair of points $a, b \in \mo$ such that $x(a)<x(b)$ and $x(b)>x(\XM(\vis(a))$, the line segment $\Rb(a)$ partitions $\mo$ into two subpolygons $\mo_a^{in}, \mo_a^{out}$ such that $a \in \mo_a^{in}$ and $b \in \mo_a^{out}$. Also observe that, this $\Rb(a)$ forms the boundary between the two partitions $\mo_a^{in}$ and $\mo_a^{out}$. We consider $\Rb(a)$ to belong to the partition $\mo_a^{in}$. Similarly, we can partition $\mo$ into two subpolygons $\mo_b^{in}$ and $\mo_b^{out}$, by $\Lb(b)$ as well, where we consider $\Lb(b)$ to belong to the partition $\mo_b^{in}$. We see that $\vis(a) \subseteq \mo_a^{in}$ and $\vis(b) \subseteq \mo_b^{in}$. This is true because if there exists a point $y\in \vis(a) \cap \mo_a^{out}$ then, the line of visibility from $a$ to $y$ lies entirely in $\mo$ and hence must intersect $\Rb(a)$ (as it partitions $\mo$). This can only happen if $y$ lies in $\Rb(a)$, as $\Rb(a)$ is collinear to $a$ (see \Cref{containedinpartion}). But this leads to a contradiction as $\Rb(a) \subseteq \mo_a^{in}$. A similar argument shows that $\vis(b) \subseteq \mo_b^{in}$. Now, as $\vis(a) \cap \vis(b) \neq \emptyset$, it implies that $\mo_a^{in} \cap \mo_b^{in} \neq \emptyset$. But as $\mo_a^{in}$ and $\mo_b^{in}$ are two partitions of the polygon $\mo$ whose boundaries are $\Rb(a)$ and $\Lb(b)$ respectively, they must intersect. (See  \Cref{fig-cintersect}).
            
            Conversely, let us assume that $\Rb(a)$ intersects with $\Lb(b)$ at $c\in \mo$. Now, observe that the point $c$ is visible from both $a$ and $b$. Hence, we  conclude that $\vis(a)\cap\vis(b)\neq \emptyset$.
\begin{figure}[ht!]
		\centering
		\begin{subfigure}[b]{0.45\textwidth}
			\centering
			\includegraphics[width=\textwidth]{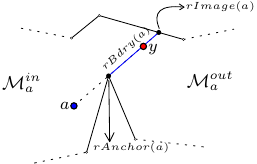}
			\subcaption{$a$ sees no point in the partition $\mo_a^{out}$.}
			\label{containedinpartion}
		\end{subfigure}
		\hspace{8mm}
		\begin{subfigure}[b]{0.45\textwidth}
			\centering
			\includegraphics[width=\textwidth]{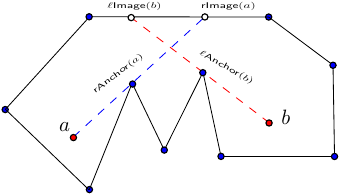}
			\subcaption{Intersection of $\Rb(a)$ and $\Lb(b)$.}
			\label{fig-cintersect}
		\end{subfigure}
        \caption{Necessary and sufficient condition for two visibility regions to intersect.}
	\end{figure}
  %update the image with respect the proof.
	\end{proof}
    
Therefore, we have established a necessary and sufficient condition to determine whether $\vis(a)\cap\vis(b)$ is empty. 
    
%     Now, due to the property of monotonicity and \Cref{lem-cl2}, the following observation is immediate.

% \begin{observation}\label{obs-leftrightcheckingenough}
% Suppose we are given a monotone polygon $\mo$ and a  witness set $W = \{w_1, w_2, \ldots, w_k\}$  in $\mo$. Then, for any witness point $w_i$ where $i \in \{2, \ldots, k-1\}$, if we can find another point $q$ such that $x(w_{i-1}) < x(q) < x(w_{i+1})$ with $\vis(w_{i-1}) \cap \vis(q) \neq \emptyset$ and  $\vis(w_{i+1}) \cap \vis(q) \neq \emptyset$ then $W \smallsetminus \{w_i\} \cup \{q\}$ is also a witness set in $\mo$.
% 		\end{observation}

	\begin{lemma} \label{lem-intersectionGconstruction}
		$\G(F)$ can be constructed in $\OO(n|F| +|F|^2)$ time.
	\end{lemma}
	
	\begin{proof} Let $|F|=f$. The visibility polygon for a point $p \in \mo$, denoted as $\vis(p)$, can be determined in $\OO(n)$ time, as demonstrated by \cite{ghosh2007vis}. Consequently, computing visibility regions for all the points in $F$ requires $\OO(fn)$ time, which is also the time complexity for constructing the vertices of $\G(F)$. While determining $\vis(p)$, we generate $\Rb(p)$ and $\Lb(p)$. By utilizing \Cref{lem-cl1} and \Cref{lem-cl2}, we can check  whether $\vis(a) \cap \vis(b)$ is empty in $\OO(1)$ time. Therefore, computing the edges of $\G(F)$ for $F$ takes $\OO(f^2)$ time. Thus, the graph $\G(F)$ can be fully constructed within $\OO(nf + f^2)$ time.
	\end{proof}

	\begin{observation}\label{lem-transitive}
		Consider three distinct points $p, q, r$ in a monotone polygon $\mo$ with $x(p) < x(q) <x(r)$. If $\vis(p) \cap \vis(q)= \emptyset$ and $\vis(q) \cap \vis(r) = \emptyset$ then $\vis(p) \cap \vis(r) = \emptyset$.
	    
	\end{observation}

	\begin{proof}
		Since $\vis(p) \cap \vis(q) = \emptyset$ and $\vis(q) \cap \vis(r) = \emptyset$, the line $l(q)$ divides the polygon $\mo$ into two separate regions, where one contains $p$ and the other contains  $r$. Thus, for $\vis(p)$ to have an intersection with $\vis(r)$, one of them must intersect $l(q)$, as $l(q)$ is responsible for partitioning $\mo$ into distinct regions, each containing either $p$ or $r$. Hence the proof.
	\end{proof}

   \Cref{lem-transitive} directly implies the following. Consider  a monotone polygon $\mo$ and a  witness set $W = \{w_1, w_2, \ldots, w_k\}$  in $\mo$. Then, for any witness point $w_i$ where $i \in \{2, \ldots, k-1\}$, if we can find another point $q$ such that $x(w_{i-1}) < x(q) < x(w_{i+1})$ with $\vis(w_{i-1}) \cap \vis(q) = \emptyset$ and  $\vis(w_{i+1}) \cap \vis(q) = \emptyset$ then $W \smallsetminus \{w_i\} \cup \{q\}$ is also a witness set in $\mo$. Below, we show that the visibility intersection graph of a monotone polygon is co-comparable.

	%Note that $x(p) = x(q)$ then $\vis(p) \cap \vis(q) \neq \emptyset$ and we don't need to apply \Cref{lem-transitive}.

	%Do we need this following observation?
	%\begin{observation}
		%There exists a configuration of a monotone polygon $\mo$ with a set $F$ of four vertices $a,b,c,d$ such that $\G(a,b,c,d)$ form a chordless four cycle(see  \Cref{fig-fourc}).
	%\end{observation}

	%\begin{figure}[h]
	%	\centering
	%	\includegraphics[scale=1.3]{fig/fourcycle.pdf}
	%	\caption{ $\G(a,b,c,d)$ contain a chordless four cycle. The vertices $a,b,c,d$ of $\mo$ are denoted by red disks and the rest of the vertices of $\mo$ are denoted by black disks. The images of the vertices are denoted by blue squares. $\vis(a)$ is defined by the polygon induced by $a,p,r,s,y,b$; $\vis(b)$ is defined by the polygon induced by $b,y,s,r,p,a$; $\vis(c)$ is defined by the polygon induced by $c,d,u,w,t$; $\vis(d)$ is defined by the polygon induced $d,u,q,r,t,c$. The vertices of visibility regions that are not vertices of $\mo$ are denoted by blue squares.}
	%	\label{fig-fourc}
	%\end{figure}

	\begin{lemma} \label{lem-cocomparable}
		Let $F$ be a set of points in a monotone polygon. Then $\G(F)$ is a co-comparable graph.
	\end{lemma}

	\begin{proof}
		
		We show that the complement graph $\widetilde{\G(F)}$ of $\G(F)$ is comparable. To that end, we prove that the edges of  $\widetilde{\G(F)}$ has a transitive orientation, i.e.,   an assignment of directions to the edges of the graph (i.e. an orientation of the graph) such that the adjacency relation of the resulting directed graph is transitive: whenever there exist directed edges $(p,q)$ and $(q,r)$, there must exists an edge $(p,r)$ where $ p,q,r $ are the vertices in $\widetilde{\G(F)}$.

  % We say the edges of a graph $G$ is transitive orientable if there is an assignment of the directions of the edges in $G$ 
		% in the following way: If there is a directed edge from  vertex $p$ to  vertex $q$ and directed edge vertex $q$ to  vertex $r$, then there must be  a directed edge from  vertex $p$ to  vertex $r$, where $ p,q,r $ are the vertices in $G$. 
		%In our case, we show, in $\widetilde{\G(F)}$, if $\vis(p) \cap \vis(q)= \emptyset$ and $\vis(q) \cap \vis(r) = \emptyset$ then $\vis(p) \cap \vis(r) = \emptyset$.
		
		We assign a transitive orientation to the edges of $\widetilde{\G(F)}$ in the following method: Begin by ordering the vertices in $\widetilde{\G(F)}$ based on sorting the points in $F$ according to their $x$-coordinates, with a vertex $p$ preceding a vertex $q$ if $x(p) > x(q)$. When an edge exists between any two vertices $p$ and $q$ where $p$ comes before $q$, we direct the edge from $p$ to $q$. Following this directional assignment, if there is a directed edge from $p$ to $q$ and another from $q$ to $r$ in $\widetilde{\G(F)}$, it follows that $x(p) > x(q) > x(r)$ and $\vis(p) \cap \vis(q)= \emptyset$ as well as $\vis(q) \cap \vis(r) = \emptyset$. By applying \Cref{lem-transitive}, we deduce $\vis(p) \cap \vis(r) = \emptyset$, necessitating a directed edge from $p$ to $r$ in the orientation. Because $p, q, r$ are chosen arbitrarily, we can conclude that $\widetilde{\G(F)}$ forms a comparability graph. 
		%Note that $x(p) = x(q)$ then $\vis(p) \cap \vis(q) \neq \emptyset$ and we don't need to apply  \Cref{lem-transitive}. So the vertices with same $x$-coordinate can be ordered arbitrarily in the sorted order of $F$. \qed
	\end{proof}

	\begin{proof}[\textbf{Proof of \Cref{theo-finite-witness-in-M}}]
	By \Cref{lem-intersectionGconstruction}, the graph $\G(S)$ can be formed in $\OO(|V(\po)||S| + |S|^2)$ time. Given that $\G(S)$ is co-comparable (due to \Cref{lem-cocomparable}), \Cref{prop:mis} implies that a maximum independent set in $\G(Q)$ can be determined in $\OO(|S|^2)$ time \cite{DBLP:journals/dm/McConnellS99}. Thus, the overall running time is  $\OO(|S|^2 + |V(\po)||S|)$. 
	\end{proof}

}

\section{{\sc Witness Set} in a Monotone Polygon}

\paragraph{Overview of this section.} In this section, our aim is to find an algorithm for the {\sc  Witness Set} in monotone polygons. To do that, we first build the preliminaries needed to develop our algorithm. We try to find a discrete set inside a monotone polygon $\mo$ that will suffice to contain a solution of the {\sc Witness Set} in $\mo$. We propose an algorithm (\Cref{algo_1})for the {\sc Witness Set} in $\mo$ running in time $r^{\OO(k)}n^{\OO(1)}$,  accordingly it creates a $\pws$ of $\mo$, i.e., it generates a set of points which contain the set of vertices of $\mo$, the midpoints of the line segments joining any two reflex vertices of $\mo$, and a new set of points on $\mo$, which we will {denote by $V(\mo),\ro_{mid}$ and $\zm$, respectively} (defined later). Note that for a witness set $X= \{w_1, \ldots, w_{\ell} \}$ with $x(w_1) < \ldots < x(w_{\ell})$, a point $w \notin X$ to check whether $X \cup \{w\}$ becomes a witness set, we only need to check the intersection of $\vis(w)$ with one pair of witnesses in $X$, which are the left and right points of $w$ in $X$, that is, $w_i$ and $w_{i+1}$, where $x(w_i) < x(w) < x(w_{i+1})$ \ifthenelse{\boolean{shortver}}{}{(due to \Cref{lem-transitive})}. We then prove the correctness of our algorithm. After obtaining a $\pws$ of $\mo$, we use the result from the previous sections to obtain the visibility intersection graph $\G(\mo)$. Since $\G(\mo)$ is an outerstring graph \ifthenelse{\boolean{shortver}}{}{(due to \Cref{vigsiglemma} and \Cref{sigouterstring})}, we can obtain maximum independent set via an algorithm presented {by the authors} in \cite{DBLP:journals/comgeo/KeilMPV17}, to obtain an optimal Witness set of $\mo$ \ifthenelse{\boolean{shortver}}{}{(due to \Cref{indwi})}. Finally, we conclude with the main result of this section.

%\todo{Udvas modify this}

  \thfpt*

% \todo{Using the phrase "In this section..." twice. Needs update @Satya Da}
 In the following, we formally describe our procedure for finding a maximum witness set in a monotone polygon $ \mo $. 
	Let $ W= \{w_1, \ldots, w_k\} $ be a solution of the {\sc Witness Set} in $\mo$ where $\mathtt{ws}(\mo)=k $ and $ x(w_1) < \ldots < x (w_k) $ (due to $x-$monotonicity). We characterize a region of special importance for each witness in $ W $ (called $\tra(\cdot,\cdot)$, formally defined later, \Cref{def-trap}), which in turn will help us to obtain a discretization of $\mo$ into a finite point set $Q$ such that the size of a maximum witness set $Z \subseteq Q$ is $k$.

 \subsection*{Step 1: Discretization} \label{sec-Discretization}
 
We start by introducing some essential definitions and notation for our purposes. 

\ifthenelse{\boolean{shortver}}{}{To start with, we define a set $\ter(w, W)$ and a region denoted as $\tra(w)$ for a witness $w \in W$, as described below.}

	\begin{definition}[$\ter(\cdot,\cdot)$] \label{def-territory}
		{\em Let $W$ be a witness set of a monotone polygon $\mo$. For   $w \in W$, $\ter(w, W)$ is a non-empty connected region in  $\mo$ that satisfies
  \begin{enumerate}
      \item $ \vis(w)  \subseteq \ter(w, W)$
      \item $\ter(w, W) \cap \vis(w')=\emptyset$ for every $w' \in W \smallsetminus \{w\}$
  \end{enumerate}}

	\end{definition}
	
	\noindent In \Cref{def-territory}, condition (1) requires that $\ter(w, W)$ fully contains $\vis(w)$, and condition (2) ensures it does not intersect $\vis(w')$ for any $w' \in W \setminus \{w\}$. Note that $\vis(w)$ itself satisfies both conditions, so we treat $\ter(w, W) = \vis(w)$. For each $w \in W$, we define a connected region $\tra(w, W)$, and in \Cref{sec-trap}, we show that such a region always exists.

%In \Cref{def-territory}, condition (i) specifies that the region $\ter(w, W)$ must completely contain $\vis(w)$, while condition (ii) ensures that $\ter(w, W)$ does not intersect with the visibility regions of any other witnesses in $W$ except for $w$. Observe that, {for any witness $w \in W$,} $\vis(w)$ itself is an example of $\ter(w,W)$. From now on, we will consider $\vis(w)$ to be the $\ter(w,W)$ for simplicity. Now, for each $ w \in W$, we define a {\em connected region}, referred  as $\tra(w,W)$. Subsequently, in \Cref{sec-trap}, we demonstrate that for every $w \in W$, the region $\tra(w,W)$ always exists. 
	
\begin{sloppypar}
\begin{definition}[$\tra(\cdot,\cdot)$] \label{def-trap}		
		{\em Let $W$ be a witness set of a monotone polygon $\mo$. For   $w \in W$,  $\tra(w, W)  $ is a connected region in $\ter(w, W)$  such that $\vis(\tra(w, W)) \subseteq  \ter(w,W)$.}
	\end{definition}
    \end{sloppypar}
	
	The definition mentioned above leads to the following observation, which is easy to follow.
	
	\begin{observation}\label{obs-wit_repre}
		For a witness set $W$ of a monotone polygon $\mo$ and $w \in W$, if $\tra(w,W) \neq \emptyset$  then for any point $ z (\neq w) \in  \tra(w,W)$, $ (W \smallsetminus \{w\}) \cup \{ z\} $ is also a witness set in $\mo$.
		
	\end{observation}

	\subsection{Characterizations of \texorpdfstring{\boldmath$\tra()$}{Lg}}\label{sec-trap}

	In this section, we characterize and define the $ \tra(w,W) $ for every $ w \in W $, {for which both $\LA(w)$ and $\RA(w)$ exists,} in detail. {The case for the leftmost and the rightmost witness will be dealt with later.} Recall that for any point $z \in \tra(w,W)$, $ \vis(z) \subseteq \ter(w, W)$ whereas  $\tra(w,W) $ is a connected region that is fully contained in $\ter(w, W)$. Below we define a notion of $ \RC(p) $ and $ \LC(p) $ for any point $p$ in $\mo$. 
    
    %\todo{should we update the definition of $\tra(\cdot,\cdot)$ and totally omit the concept of $\ter(\cdot,\cdot)$?}
	
	\begin{definition}[$ \RC() $ and $ \LC() $]\label{def-RCLC}
		{\em For a point $p$ in a monotone polygon  $\mo$, the maximal  chord in $\mo$ passing through the vertices $ p$ and  $\IMr(p)$ is defined as $\RC(p)$. Similarly, the maximal chord in $\mo$ passing through the vertices $ p$ and  $\IMl(p)$ is defined as $\LC(p)$.} 
	\end{definition}

    \ifthenelse{\boolean{shortver}}{}{

    \begin{definition}[${ \PCU()}$ and ${ \PCL()}$] \label{pChain}
		{\em For a pair of points $ a, b \in \uc(\mo)$ (resp, $\in \lc(\mo)$), we use $ \PC_{u}(a,b)$ (resp,  $ \PC_{\ell}(a,b)$ )  to denote the path from $a$ to $b$ along $\uc(\mo)$ (resp, $ \lc(\mo) $)}.
		
	\end{definition}

    \begin{definition}[hill]\label{hill}
       {\em For a pair of distinct points $x,y$ in a polygon $\po$, which cannot see each other, i.e., $y \notin \vis(x)$, and vice-versa, the line joining $x$ and $y$, must intersect the exterior $\ex(\po)$ of $\po$. Specifically, $\overline{xy} \cap \ex(\po) = \{\ell_1, \cdots, \ell_j\}$, where $\ell_i'$s are open line segments\footnote{By an open line segment $(a,b)$ we mean the line segment $\overline{ab}$ except the two endpoints $a$ and $b$.} with their endpoints on $\bd(\po)$. For any line segment $\ell_i$, we name its two endpoints as $h_i^1$ and $h_i^2$. For any $i = 1, \cdots,j$, the polygonal chain (part of $\bd(\po)$) from $h_i^1$ to $h_i^2$ is called a \textit{hill} $\mathcal{H_i}$ corresponding to a pair of mutually invisible points $x$ and $y$ (Refer to \Cref{fig-reflex} for an illustration).}
    \end{definition}

\begin{figure}[ht!]
		\centering
		\includegraphics[scale=1.0]{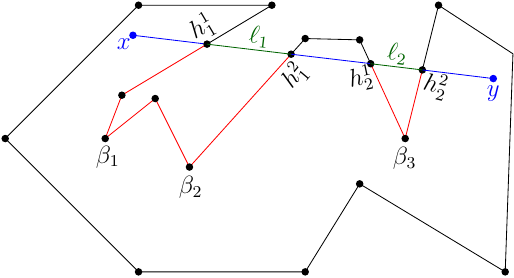}
		\caption{The points $x, y \in \po$ cannot see each other. $\ell_1 \cup \ell_2 = \overline{xy} \cap \bd(\po)$. The polygonal chains from $h_1^1$ to $h_1^2$ and $h_2^1$ to $h_2^2$ are two \textit{hills} $\mathcal{H_1},\mathcal{H_2}$ corresponding to $x,y$. $\beta_1,\beta_2$ are two reflex vertices present on $\mathcal{H_1}$ and $\beta_3$ is another reflex vertex present on the hill $\mathcal{H_2}$.} \label{fig-reflex}
	\end{figure}

    \noindent The immediate observation concerns {when} two points {in $\po$} cannot see each other. 
    
\begin{observation}\label{obs-hill}
     For a pair of {distinct} points $x,y$ in a polygon $\po$ such that $y \notin \vis(x)$, there exists at least one hill corresponding to $x,y$ (note that there might exist multiple hills corresponding to $x,y$). Any such hill (which is a chain in $\bd(\po)$) contains at least one reflex vertex (See \Cref{fig-reflex}).
\end{observation}

\begin{proof}
    The two endpoints $h_i^1,h_i^2$ of any hill $\mathcal{H_i}$ are two points on $\bd(\po)$ which cannot see each other. If the polygonal chain between them consists of only convex vertices, it means that $h_i^1$ and $h_i^2$ can see each other, which is a contradiction.
\end{proof}

}

   We now propose a claim on $\tra(\cdot,\cdot)$ which will be crucial in our paper hereon.

    \begin{clm}\label{clm-tregion}
       For a witness $w \in W$, there exists a non-empty $\tra(w,W) \setminus \{w\}$ such that it is either an open line segment or $\tra(w,W)$ has a non-empty intersection with $\bd(\mo)$.
    \end{clm}

    \ifthenelse{\boolean{shortver}}{}{
    
    \begin{proof}
       Observe that, for any witness $w$, the intersection of its visibility region $\vis(w)$ with another visibility region of a witness $w'$ depends on its $\Rb(w)$ and $\Lb(w)$ (Due to \Cref{lem-cl2}). Now, we consider the following two cases.

\begin{sloppypar}

    \begin{description}
        \item[{Case 1: \textbf{The witness $w, \RA(w)$ and $\LA(w)$ are collinear.}}] 
        {In this case, we have that the $\RC(w)$ and $\LC(w)$ (as in \Cref{def-RCLC}) are essentially the same line (See Figure \hyperref[fig-wLARACollinear]{17} for an example). So, if we move the witness $w$ to any other point $w'$ along the line segment $\overline{\LA(w)\RA(w)}$ (except the two points $\RA(w)$ and $\LA(w)$), then $\RA(w')$ and $\LA(w')$ remain unchanged and therefore their respective $\Rb(w')$ and $\Lb(w')$ remains the same as in the case of the witness $w$. Thus, since the intersection of $\vis(w')$ with any other visibility region of a witness depends on the position of its $\Rb(w')$ and $\Lb(w')$ (due to \Cref{lem-cl2}), $\vis(w')$ does not intersect with any other visibility region of another witness. So, by \Cref{def-territory}, $\vis(w') \subseteq \ter(w,W)$, which means that $w' \in \tra(w,W)$. This shows that in this case the open line segment $\overline{\LA(w)\RA(w)} = \overline{\RA(w')\LA(w')} \subseteq \tra(w,W)$.}

        \item[{Case 2: \textbf{The witness $w, \RA(w)$ and $\LA(w)$ are not collinear.}}]
       {In this case, we will show that either $w \in \bd(\mo)$ (in which case we need not show anything in particular) or $\tra(w,W)$ is non-empty and that $\tra(w, W) \cap \bd(\mo) \neq \emptyset$, or an open line segment $(r_1,r_2) \subseteq \tra(w,W)$, where $r_1$ and $r_2$ are two reflex vertices of $\mo$. As $w, \RA(w)$ and $\LA(w)$ are not collinear, the respective $\LC(w)$ and $\RC(w)$ are two distinct lines intersecting at $w$. Let the two endpoints of $\LC(w)$ and $\RC(w)$ other than $\IMl(w)$ and $\IMr(w)$ be $a$ and $b$, respectively. Now, observe that if we move $w$ to any other point $w^*$ along the $\LC(w)$ between $\overline{\LA(w)a}$ (similarly along $\RC(w)$ between $\overline{\RA(w)b}$, resp.), then the visibility region changes but the $\LC(w^*)$ ($\RC(w^*)$, resp.) does not. Let us look at one of the above specific cases to understand the substitution of $w$ by $w^*$ in a better way. If we move $w$ to a point $w^*$ along $\LC(w)$ towards $a$, then the visibility region may increase, or decrease on the right side of $\vis(w)$, but it remains unchanged on the left side as the $\LC(w)$ remains unchanged. Similarly, this works if we move $w$ towards $\LA(w)$. However, the movement of $w$ in these two opposite directions results in two different outcomes. If the visibility region increases while moving $w$ towards $a$, then the visibility region decreases if we move $w$ towards $\LA(w)$, and vice-versa. This is because the change in visibility region while moving $w$ along $\overline{\LA(w)a}$ is actually affected by the rotation of $\RC(w)$ pivoting about the reflex vertex $\RA(w)$ (See Figure \hyperref[fig-tRegionProof4]{13}). Let us name the movement of $w$ along the $\LC(w)$ in the two opposite directions by $\ell Cdir_1$ and $\ell Cdir_2$. In a similar manner, let us name the two opposite directions of movement (one movement towards $\RA(w)$ and another movement is towards $b$) of $w$ along the $\RC(w)$ in between $\overline{\RA(w)b}$ as $rCdir_1$ and $rCdir_2$. Essentially, we are denoting the \textit{direction of movements} of $w$ towards any four of the points $\LA(w), a$ (along $\LC(w)$) and $\RA(w),b$ (along $\RC(w)$) by $\ell Cdir_1$, $\ell Cdir_2$, $rCdir_1$, and $rCdir_2$, respectively (See Figure \hyperref[fig-tRegionProof1]{14}). Note that if we choose $w^*$ by moving $w$ towards any of these four directions, exactly for two of the directions the visibility region $\vis(w^*)$ increases, and for exactly two other directions, $\vis(w^*)$ decreases. To be precise, moving $w$ towards the directions $\ell Cdir_1$ and $\ell Cdir_2$ results in opposite outcomes for the change in the visibility area, $\vis(w)$. The same goes for the two directions $rCdir_1$ and $rCdir_2$. So, if we can choose $w^*$ in such a direction that the visibility region decreases, then we have that $\vis(w^*) \subseteq \vis(w) \subseteq \ter(w,W)$. So, by \Cref{def-trap}, $w^* \in \tra(w,W)$, which means that $\tra(w,W) \neq \emptyset$. \\}

        \begin{figure}[ht!]\label{fig-tRegionProof4}
		\centering
		\includegraphics[scale=1]{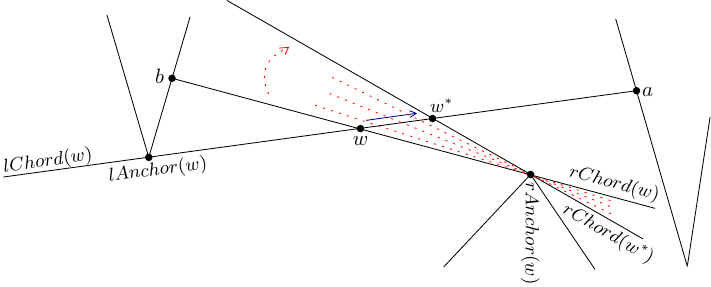}
		\caption{$w$ is the witness. The change in visibility area when $w$ is moved along $\LC(w)$ to another point $w^*$ is depicted.} 
	\end{figure}

        \begin{figure}[ht!]\label{fig-tRegionProof1}
		\centering
		\includegraphics[scale=1]{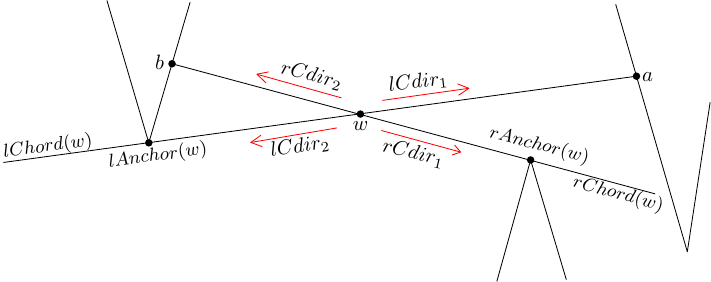}
		\caption{$w$ is the witness. $a, b$ are the endpoints of $\LC(w)$ and $\RC(w)$ respectively, other than the endpoints $\IMl(w)$ and $\IMr(w)$. The four directions of moving $w$ along either $\LC(w)$ or $\RC(w)$ are shown.} 
	\end{figure}
        
        {We now move forward to our next part of the proof, which is to show that either an open segment $(r_1,r_2) \subseteq \tra(w,W)$ where $r_1$ and $r_2$ are two reflex vertices of $\mo$, or $\tra(w,W) \cap \bd(\mo) \neq \emptyset$. If our choice of $w^* \in \tra(w,W)$ lies on either of $\overline{wa}$ or $\overline{wb}$, then we can choose $w^*$ to be the endpoints $a$ or $b$ itself, as choosing $a$ reduces the visibility region on the right side of $\vis(w)$ without changing the $\LC(w)$ and choosing $b$ reduces the visibility region on the left side of $\vis(w)$ without changing the $\RC(w)$. So, in that case, as both $a, b \in \bd(\mo)$, we have that $\tra(w,W) \cap \bd(\mo) \neq \emptyset$. On the other hand, if our choice of $w^* \in \tra(w,W)$ was restricted to the lines $\overline{w\LA(w)}$ and $\overline{w\RA(w)}$, then we cannot directly choose $w^*$ to be $\LA(w)$ or $\RA(w)$, as then, $\LA(w^*) = \LA(\LA(w))$ would not be same as $\LA(w)$ and also $\RA(w^*) = \RA(\RA(w))$ would not be the same as $\RA(w)$. This could result in an increase in the visibility region of $w^*$, $\vis(w^*)$, which will not help in our existence and construction of the $\tra(w,W)$. To resolve this, we consider two sub-cases. }

        \begin{itemize}
            \item {\textbf{The open segment $(\LA(w),\RA(w))$ lies totally inside $\mo$} (See Figure \hyperref[noncollinear1]{$15(a)$}). In this case, any point on the open segment $(\LA(w),\RA(w))$ can be considered as $w^*$, where $w^*$ is in $\tra(w,W)$. In other words, the open segment $(\LA(w),\RA(w)) \subseteq \tra(w,W)$. A short proof of this is as follows: \\ Let $w^*$ be any point on $(\LA(w),\RA(w))$. Then, we draw two lines $\ell_1, \ell_2$ from $w^*$, one of which is parallel to $\overline{w\LA(w)}$ and another one is parallel to $\overline{w\RA(w)}$, respectively. Now let $\ell_1 \cap \overline{w\RA(w)} = p_1$ and $\ell_2 \cap \overline{w\LA(w)} = p_2$, respectively. Then, by moving our initial witness $w$ to $p_1$, we get that $\vis(p_1) \subseteq \vis(w)$. Again, by moving $p_1$ to $w^*$, we are essentially moving $p_1$ in the direction parallel to $\LA(w)$. So, $\vis(w^*) \subseteq \vis(p_1)$. Therefore, $\vis(w^*) \subseteq \vis(w)$, which implies that $w^* \in \tra(w,W)$, as required.} 

            \item {\textbf{The open segment $(\LA(w),\RA(w))$ does not lie totally inside $\mo$, i.e., $\LA(w)$ and $\RA(w)$ is not visible to each other} (See Figure \hyperref[noncollinear2]{$15(b)$}). In this case, due to \Cref{obs-hill}, there must exist a hill containing a reflex vertex $r$ that intersects the line $\overline{\LA(w)\RA(w)}$. Observe that the orientation of this reflex vertex forces it to lie inside the triangle $\triangle{\LA(w)w\RA(w)}$. Then, this reflex vertex $r$, which also lies on $\bd(\mo)$, will lie in $\tra(w,W)$. This is again by the same argument as provided in the case above. $\vis(r) \subseteq \vis(p_1) \subseteq \vis(w)$. So, $r \in \bd(\mo)$ lies in $\tra(w,W)$, which means that $\tra(w,W) \cap \bd(\mo) \neq \emptyset$.}
        \end{itemize}

    \end{description}   

    \end{sloppypar}

        \begin{figure}[ht!]\label{fig-wLARAnonCollinear}
    \centering
    \begin{subfigure}[b]{0.75\textwidth}
        \centering
        \includegraphics[width=\textwidth]{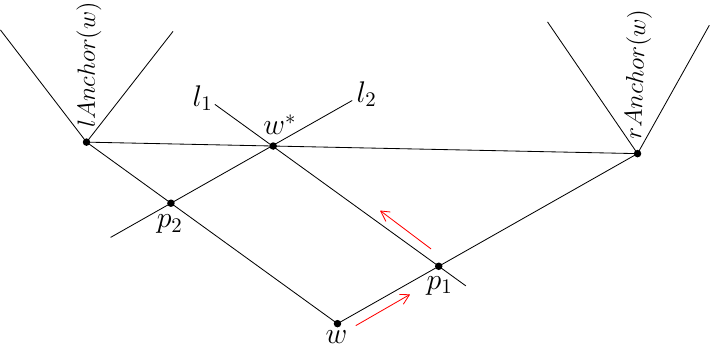}
        \subcaption{$\overline{\LA(w)\RA(w)}$ lies totally inside the polygon $\mo$.}
        \label{noncollinear1}
    \end{subfigure}
    \hspace{8mm}
    \begin{subfigure}[b]{0.75\textwidth}
        \centering
        \includegraphics[width=\textwidth]{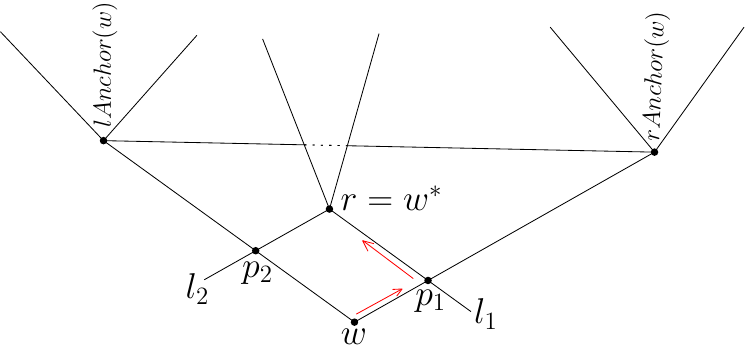}
        \subcaption{ $\overline{\LA(w)\RA(w)}$ does not lie totally inside the polygon $\mo$.}
        \label{noncollinear2}
    \end{subfigure}
    \caption{Cases where $ \LA(w), \RA(w) $ and $ \bm{w} $ are not collinear and by moving $\bm(w)$ towards $\LA(w)$ and $\RA(w)$ reduces the visibility region.}
\end{figure}
        
    \end{proof}
    }

    \ifthenelse{\boolean{shortver}}{}{
	\noindent Thus far, we can see that for a witness $w \in W$, the $\tra(w,W)$ can be classified into two categories based on the relative positions of $\LA(w)$, $\RA(w)$, and $w$: (i) open line segments and (ii) connected regions with non-empty interior. Below, we provide a more detailed examples of construction of these classifications. 

    \begin{sloppypar}

	\begin{description}
		\item[Case 1. \textbf{Both the vertices} $ \bm{\LA(w)}$ and $ \bm{\RA(w)} $ \textbf{belongs to same chain}.] 
       
        There are four sub-cases.
		\begin{figure}[ht!]
				\centering
				\begin{subfigure}[b]{0.4\textwidth}
					\centering
					\includegraphics[width=\textwidth]{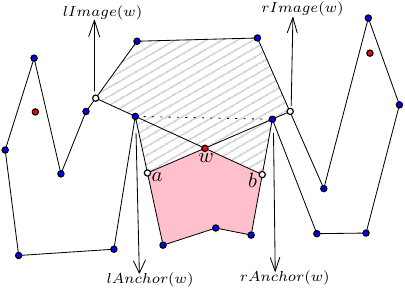}
					\subcaption{Case \hyperref[c1]{1.1}}
					\label{trg-llb}
				\end{subfigure}
				\hspace{8mm}
				\begin{subfigure}[b]{0.4\textwidth}
					\centering
					\includegraphics[width=\textwidth]{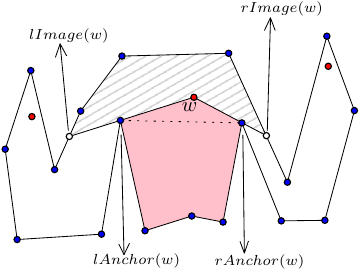}
					\subcaption{Case \hyperref[c2]{1.2}}
					\label{trg-llu}
				\end{subfigure}
                \hspace{8mm}
				\begin{subfigure}[b]{0.4\textwidth}
					\centering
					\includegraphics[width=\textwidth]{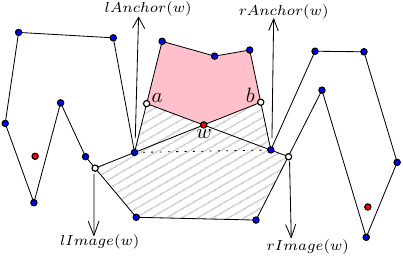}
					\subcaption{Case \hyperref[c3]{1.3}}
					\label{trg-uuu}
				\end{subfigure}
                \hspace{8mm}
				\begin{subfigure}[b]{0.4\textwidth}
					\centering
					\includegraphics[width=\textwidth]{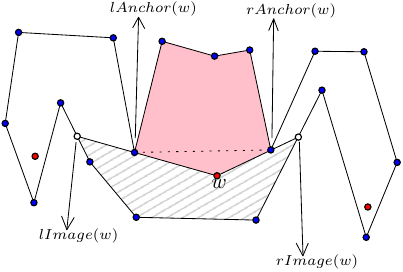}
					\subcaption{Case \hyperref[c4]{1.4}}
					\label{trg-uub}
				\end{subfigure}
			\caption{Red points are set of witnesses. Here $ \LA(w)$ and $ \RA(w) $ belongs to same chain. Pink region is defined for $ \tra(w) $ which is a part of $ \ter(w) $.} \label{fig-trap3}
		\end{figure}
        \begin{description}
             
            \item[1.1\label{c1}] Here we consider that  both belongs to the lower chain $ \lc(\mo) $, and the witness $w$ lies inside the region bounded by $\PCL(\LA(w),\RA(w))$, $\overline{\RA(w)\LA(w)}$. Let $ a $ and $ b $ be the endpoints of   $ \RC(w) $ and $ \LC(w) $, respectively where $ a \neq \IMr(w) $ and $ b \neq \IMl(w) $. Here $ \tra(w,W) $	 is the connected region defined by the boundaries  $   \overline{wb}, ~\PCL(b,a)  $ and $ \overline{a w} $.  See  \Cref{trg-llb} for an illustration.
\smallskip 
            \item[1.2\label{c2}] Here we consider that  both belongs to the lower chain $ \lc(\mo) $, but the witness $w$ lies outside the region bounded by $\PCL(\LA(w),\RA(w))$, $\overline{\RA(w)\LA(w)}$. Then, the $\tra(w,W)$ is the connected region defined by the boundaries $\PCL(\RA(w),\LA(w), ~ \overline{\LA(w)w}, ~ \overline{w\RA(w)}$. See \Cref{trg-llu} for an illustration.
\smallskip
            \item[1.3\label{c3}] Now we consider that  both belongs to the upper chain $ \uc(\mo) $ and the witness $w$ lies inside the region bounded by $\PCU(\LA(w),\RA(w))$, $\overline{\RA(w)\LA(w)}$. Let $ a $ and $ b $ be the endpoints of   $ \RC(w) $ and $ \LC(w) $, respectively where $ a \neq \IMr(w) $ and $ b \neq \IMl(w) $. Here $ \tra(w,W) $	 is the connected region defined by the boundaries  $   \overline{wb}, ~\PCU(b,a)  $ and $ \overline{a w} $.    See \Cref{trg-uuu} for an illustration.
\smallskip
            \item[1.4\label{c4}] Now we consider that  both belongs to the upper chain $ \uc(\mo) $ and the witness $w$ lies outside the region bounded by $\PCU(\LA(w),\RA(w))$, $\overline{\RA(w)\LA(w)}$. Then, the $\tra(w,W)$ is the connected region defined by the boundaries $\PCU(\RA(w),\LA(w), ~ \overline{\LA(w)w}, ~ \overline{w\RA(w)}$. See \Cref{trg-uub} for an illustration.

        \end{description}
        
         \item[Case 2.  $ \bm{\LA(w)}$ and $ \bm{\RA(w)} $ \textbf{belong to different chains}.] There are two subcases:
         
         %\todo{needs picture alignment for better understanding. Also how to label these subcases?}
       % $~~$

        \begin{description}
            \item[2.1. $ \bm{\LA(w), \RA(w) }$ \textbf{and} $ \bm{w} $ \textbf{are collinear.}\label{c21}] First we consider the situation that $  \LA(w) \in \uc(\mo)$ and $ \RA(w) \in \lc(\mo)$.  Here $ \tra(w,W) $	is defined  by the open line segment $ (\LA(w), \RA(w)) $ (See  \Cref{wint1ul} for an illustration). The argument for  $  \LA(w) \in \lc(\mo)$ and $ \RA(w) \in \uc(\mo)$  in this case is exactly same (See \Cref{wint2lu}).

\begin{figure}[ht!]\label{fig-wLARACollinear}
    \centering
    \begin{subfigure}[b]{0.45\textwidth}
        \centering
        \includegraphics[width=\textwidth]{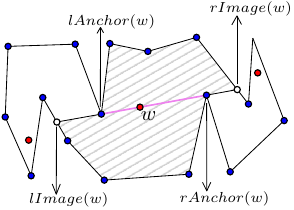}
        \subcaption{$  \LA(w) \in \uc(\mo), \RA(w) \in \lc(\mo)$}
        \label{wint1ul}
    \end{subfigure}
    \hspace{8mm}
    \begin{subfigure}[b]{0.45\textwidth}
        \centering
        \includegraphics[width=\textwidth]{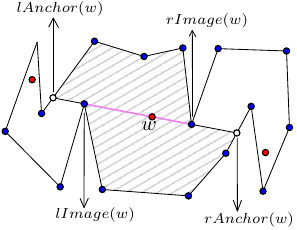}
        \subcaption{ $  \LA(w) \in \lc(\mo), \RA(w) \in \uc(\mo)$}
        \label{wint2lu}
    \end{subfigure}
    \caption{Case \hyperref[c21]{2.1}. $ \LA(w), \RA(w) $ and $ \bm{w} $ are collinear.}
    \label{winttt}
\end{figure}
            
 \item[2.2. $ \bm{\LA(w), \RA(w) }$ \textbf{and} $ \bm{w} $ \textbf{are not collinear.}\label{c22}] We will consider four distinct subcases that depend on the location of the witness and the placements of $\LA(w)$ and $\RA(w)$. Although these cases are symmetric with respect to reflection and rotation, we will present each one individually for the sake of clarity and completeness.

            \smallskip

    \begin{description}        

        \item[2.2.a\label{c23}]  Here, we consider $\LA(w)$ belongs to the upper chain $\uc(\mo)$ and $\RA(w)$ belongs to the lower chain $\lc(\mo)$; and the witness $w$  lies inside the region defined by the boundaries $\overline{\LA(w)\RA(w)},~\PCL(\RA(w),\XMi(\mo))$ and $ \PCU(\XMi(\mo),\LA(w))$. Let $ c $ and $ d $ denote the endpoints of   $ \RC(w) $ and $ \LC(w) $, respectively where $ x(c) <x(w) $ and $ x(d) >x(w) $. Then, the $ \tra(w,W) $	 is the connected region defined by the boundaries   $   \overline{wd},~ \PCL(d,\RA(w)) $ and $ \overline{\RA(w) w} $, except the point $\RA(w)$. (See  \Cref{trg-ulb} for an illustration).
\smallskip
        \item[2.2.b\label{c24}] Here, we consider $\LA(w)$ belongs to the upper chain $\uc(\mo)$ and $\RA(w)$ belongs to the lower chain $\lc(\mo)$; and the witness $w$ lies inside the region defined by the boundaries $\overline{\RA(w)\LA(w)},~\PCU(\LA(w),\XM(\mo))$ and $ \PCL(\XM(\mo),\RA(w))$. Let $ c $ denote the endpoint of   $ \RC(w)$ where $x(c) < x(w)$. Then, the $ \tra(w,W) $	 is the connected region defined by the boundaries $   \overline{wc},~ \PCL(c,\LA(w)) $ and $ \overline{\LA(w) w} $, except the point $\LA(w)$. (See  \Cref{trg-ulu} for an illustration).
\smallskip
        \item[2.2.c\label{c25}] Here, we consider $\LA(w)$ belongs to the lower chain $\lc(\mo)$ and $\RA(w)$ belongs to the upper chain $\uc(\mo)$; and the witness $w$ lies inside the region defined by the boundaries $\overline{\RA(w)\LA(w)},~\PCL(\LA(w),\XM(\mo))$ and $ \PCU(\XM(\mo),\RA(w))$. Let $ c $ denote the endpoint of   $ \RC(w) $  where $ x(c) <x(w) $. Then, the $ \tra(w,W) $ is the connected region defined by the boundaries   $   \overline{wc}, ~ \PCL(c,\LA(w)) $ and $ \overline{\LA(w) w} $, except the point $\LA(w)$. (See  \Cref{trg-lub} for an illustration).
\smallskip
        \item[2.2.d\label{c26}] Here, we consider $\LA(w)$ belongs to the lower chain $\lc(\mo)$ and $\RA(w)$ belongs to the upper chain $\uc(\mo)$; and the witness $w$ lies inside the region defined by the boundaries $\overline{\LA(w)\RA(w)}, \PCU(\RA(w),\XM(\mo))$ and $ \PCL(\XM(\mo),\LA(w))$. Let $c$ and $d$ denote the endpoints of   $ \RC(w) $ and $ \LC(w) $, respectively where $ x(c) <x(w) $ and $ x(d) >x(w) $. Then, the $ \tra(w,W) $	 is the connected region defined by the boundaries $\overline{wd},~ \PCL(d,\RA(w)) $ and $ \overline{\RA(w) w} $, except the point $\RA(w)$.(See  \Cref{trg-luu} for an illustration).
    
    \end{description}
   
    \end{description}

    \end{description}

\begin{figure}[ht!]
\centering
\begin{subfigure}[b]{0.45\textwidth}
\centering
\includegraphics[width=\textwidth]{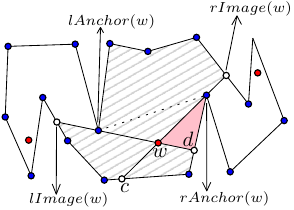}
\subcaption{Case \hyperref[c21]{2.2.a}}
\label{trg-ulb}
\end{subfigure}
\hspace{8mm}
\begin{subfigure}[b]{0.45\textwidth}
\centering
\includegraphics[width=\textwidth]{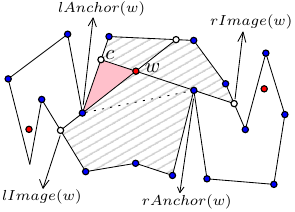}
\subcaption{Case \hyperref[c23]{2.2.b}}
\label{trg-ulu}
\end{subfigure}
\hspace{8mm}
\begin{subfigure}[b]{0.45\textwidth}
\centering
\includegraphics[width=\textwidth]{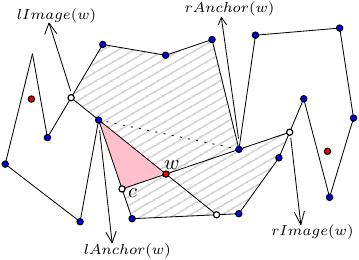}
\subcaption{Case \hyperref[c24]{2.2.c}}
\label{trg-lub}
\end{subfigure}
\hspace{8mm}
\begin{subfigure}[b]{0.45\textwidth}
\centering
\includegraphics[width=\textwidth]{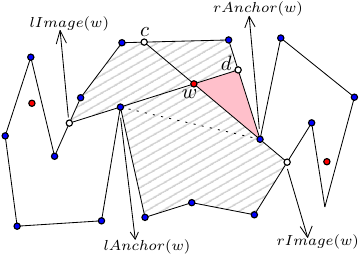}
\subcaption{Case \hyperref[c25]{2.2.d}}
\label{trg-luu}
\end{subfigure}
\caption{ Case \hyperref[c22]{2.2}. $ \LA(w), \RA(w) $ and $ \bm{w} $ are collinear. The red points are a set of witnesses. Here $ \LA(w)$ and $ \RA(w) $ belong to the same chain. Pink region is defined for $ \tra(w,W) $.} \label{fig-trap4}
\end{figure}
\end{sloppypar}		

            \begin{mdframed}[backgroundcolor=red!10,topline=false,bottomline=false,leftline=false,rightline=false] 
	%\centering
	Observe that for any witness $w \in W $, the region $ \tra(w,W) $ is either of following three types;(1) It contains some vertex from $V(\mo)$ (see Figures \ref{fig-trap3}, \ref{trg-ulb}, and \ref{trg-luu}), (2) It contains no vertex from $V(\mo)$ (see (Figures \ref{trg-ulu} and \ref{trg-lub}), and (3) It contains all points of the segment $\overline{\beta \beta'}$ except the points $\{\beta,\beta'\}$ where $\beta $ and $\beta' $ is a pairwise reflex vertices that is visible to each other (See \Cref{winttt}). 
\end{mdframed}    

%\todo{Probably not needed if we update the construction of $\ter(\cdot,\cdot)$ and $\tra(\cdot,\cdot)$}
}

\noindent Based on the above discussion, we partition the set $W$ into  $\wint$, $\wgb$, and $\wbb$.
		
	\begin{sloppypar}
	\begin{definition}[$\wint$, $\wgb$ and $\wbb$]\label{def:witness}
	{\em We partition the witness set $W$ into three  subsets, labeled $\wint$, $\wgb$ and $\wbb$, {based on the nature of their respective $\tra(\cdot,\cdot)$.} Each witness $w$ in $W$ is sorted into either of these subsets.
	\begin{description}
		\item[1. $\wint$:]{Set of the witnesses $w \in W $ for which $w, \LA(w)$ and $\RA(w)$ are collinear and $\LA(w)$ and $\RA(w)$ belongs to the different chains $\uc(\mo)$ and $\lc(\mo)$. So, if $w \in \wint$,} then $ \tra(w,W) $ is obtained  by the open line segment $ (\LA(w), \RA(w)) $. 
	\item[2. $\wgb$:] Set of the witnesses $w$ in $ W $ for which $ \tra(w,W)$ intersects $ V(\mo)$.
	
		\item[3. $\wbb$:] Set of the witnesses $w$ in $ W $ for which $ \tra(w,W)$ has a non-empty intersection with $ \bd(\mo)$ but contains no vertex from $V(\mo)$.

	\end{description}}
    \end{definition}
           \end{sloppypar}

           \ifthenelse{\boolean{shortver}}{

   \begin{figure}[ht!]
\centering
\begin{subfigure}[b]{0.45\textwidth}
\centering
\includegraphics[width=\textwidth]{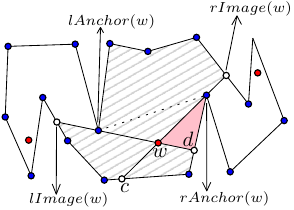}
%\subcaption{Case \hyperref[c21]{2.2.a}}
\label{trg-ulb}
\end{subfigure}
\hspace{8mm}
\begin{subfigure}[b]{0.45\textwidth}
\centering
\includegraphics[width=\textwidth]{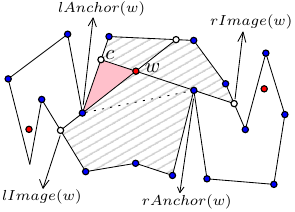}
%\subcaption{Case \hyperref[c23]{2.2.b}}
\label{trg-ulu}
\end{subfigure}
\hspace{8mm}
\begin{subfigure}[b]{0.45\textwidth}
\centering
\includegraphics[width=\textwidth]{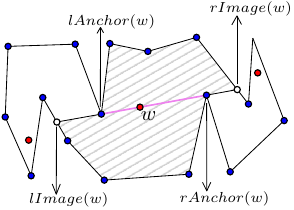}
%\subcaption{Case \hyperref[c24]{2.2.c}}
\label{trg-lub}
\end{subfigure}
\hspace{8mm}
\begin{subfigure}[b]{0.45\textwidth}
\centering
\includegraphics[width=\textwidth]{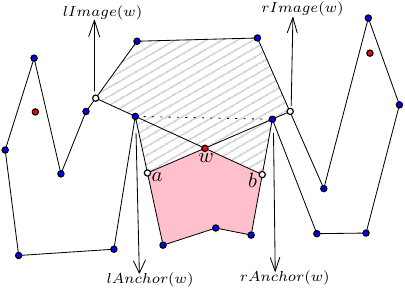}
%\subcaption{Case \hyperref[c25]{2.2.d}}
\label{trg-luu}
\end{subfigure}
\caption{Pink region is defined for $ \tra(w,W) $.} \label{fig-trap4}
\end{figure}

 \begin{mdframed}[backgroundcolor=red!10,topline=false,bottomline=false,leftline=false,rightline=false] 
	%\centering
	Observe that for any witness $w \in W $, the region $ \tra(w,W) $ is either of following three types;(1) It contains some vertex from $V(\mo)$, (2) It contains no vertex from $V(\mo)$  and (3) It contains all points of the segment $\overline{\beta \beta'}$ except the points $\{\beta,\beta'\}$ where $\beta $ and $\beta' $ is a pairwise reflex vertices that is visible to each other (See \Cref{fig-trap4}). 
\end{mdframed}
}{}

 Now, we introduce the concept of {\em potential witness set}, which contains a solution to the {\sc Witness Set}.  
	
	\begin{definition}[\bm{$ \pws $}] \label{def-potential}
		{\em For a monotone polygon $\mo$, a finite point set $ Q $ in $ \mo $ is said to be a {\em potential witness set} of $ \mo $, if there exists a subset $ Q' \subseteq Q$ such that $ Q' $ is a solution of the {\sc  Witness Set} in $\mo$. }
	\end{definition}
	
	\begin{definition}[{$ (W,S)\text{-}\pws $}] \label{def-potential1}
		{\em Consider a witness set $W$ and a subset $S \subseteq W$. A set of points $Z$ is said to be a  $ (W,S)\text{-}\pws $ if there exists a subset $Z' \subseteq Z$ satisfying  (i) $ (W \smallsetminus S) \cup Z' $ is a witness set of $\mo$, and (ii) $| (W \smallsetminus S) \cup Z' | = |W|$.}

       % $Let $ W $ be an optimal solution of $ \wsp(\mo) $. 	For a subset $ A \subseteq W $, a set $ A' $ of points is set to be $ \PW $ of $ A $ if there exist a  subset $ A'' \subseteq A'$ such that $ (W \smallsetminus A) \cup A'' $ is a  solution of $ \wsp(\mo) $. \hl{Why is this defn needed?}
	\end{definition}

Note that if $W$ is a solution to the {\sc Witness Set} within $\mo$ and $S=W$, then $ (W,S)\text{-}\pws $ corresponds precisely to the $\pws$ of $\mo$. Thus, \Cref{def-potential1} can be viewed as a generalization of \Cref{def-potential}. Now we define the notion of a potential witness set for $W$ in $\mo$. 
\begin{definition}[{$ W\text{-}\pws $}] \label{def-potential2}
		{\em Consider a witness set $W$ in $\mo$. A set of points $Z$ is said to be a  $ W\text{-}\pws $ in $\mo$ if for each witness $w$ in $W$, there exists a $q \in Z$ such that $(W \smallsetminus w) \cup \{q\}$ is a witness set in $\mo$.}

        \end{definition}

In the next section, we create both $ \wint$-\pws ~and $ \wbb$-\pws. 
	
	\subsection{Generating a Potential Witness Set of $\mo$} \label{sec-point}
	Let $\ro$ denote the set of all reflex vertices within $\mo$, and let $r$ be its cardinality. The main objective of a potential witness set $Q$ for $\mo$ is described as follows: for any solution $W$ of the {\sc Witness Set} problem, each witness $w$ in $W$ should have an alternative $q$ in $Q$ such that substituting $w$ with $q$ in $W$ results in a valid solution, specifically, $(W \smallsetminus \{w\}) \cup \{q\}$ remains a solution for the {\sc Witness Set} in $\mo$, due to \Cref{obs-wit_repre}.

 	% \subsubsection{ \texorpdfstring{\boldmath $ \PWS $}{Lg} of \texorpdfstring{\boldmath $ \wint $}{Lg}}

\paragraph{Few Conventions.} Consider the witness set $W= \{w_1, \ldots, w_k\}$ with a specific order such that $x(w_1) < \ldots < x(w_k)$. To demonstrate that $Q$ is a $W$-potential witness set of $\mo$, we proceed as follows: for each $w_i$  $ i \in [k]$, we find a substitute $w_i^* \in Q$ such that  $(W \smallsetminus \{w_i\}) \cup \{w_i^*\}$ is a  witness set within $\mo$. If $w_i \in Q$ then there is nothing to show; in that case $w_i= w^*_i$. This implies every witness $w \in W$ can be appropriately replaced by a $w^* \in Q$. \ifthenelse{\boolean{shortver}}{}{Within our proof, we adhere to the following convention: For any two points $a_1$, $a_2$ within $\mo$, we assert that $a_1$   lies to the left of $a_2$ or equivalently $a_2$  lies to the right of $a_1$ if $x(a_1) < x(a_2)$.  We will now present two important definitions --  {\em visibility region of a line segment} and {\em clone} that we use often in our proof.}

 	\begin{definition}[Visibility of a line segment from another line segment]\label{def_linevisibility} 
 		{\em Consider a pair of  line segments, $L_1$ and $L_2$, both contained entirely within $\mo$. The notation $\vis(L_1)\vert_{L_2}$ stands for the visibility region of $L_1$ when restricted to $L_2$, i.e.,  $\vis(L_1) \cap L_2$, {with respect to the subspace topology on $L_2$. Here, $\vis(L_1)$ denotes the visibility region of $L_1 \subseteq \mo$ as defined in \Cref{sec-preli}.} Essentially this  captures the portion of $L_2$ that can be seen from $L_1$ in $\mo$.}  
 	\end{definition}
 	
 	The following observation is immediate from the above definition.
 	
 	\begin{observation}\label{obs:con}
 		For any two line segments $L_1, L_2$, the region $ \vis(L_1)\vert_{L_2}$ is closed and indeed a connected segment.
 	\end{observation}
 	
 	\begin{definition}[Clone of a vertex/point]\label{def_clone}
 		{\em Let $v$ be a vertex or a point  of $\bd(\mo)$. We denote $v^+$ (resp., $v^-$) to be an {\em infinitesimally close point}  near $v$ on $\bd(\mo)$ with $x(v^+)>x(v)$ (resp., $x(v^-)<x(v)$). We call $v^+$ as the {\em right clone} of $v$ (resp., $v^-$ as the {\em left clone} of $v$). }  
 	\end{definition}
 
  	We refer to the clone of a vertex as {\em clone point}. We now present a useful claim.

 \begin{clm}
 	\label{edgevis}
 	Consider a witness $w_i \in W$ with $w_i ~{\notin \wint} ~(1 < i <k$). {Then, as $\tra(w_i) \cap \bd(\mo) \neq \emptyset$, we can assume, w.l.o.g. that $w_i$ lies on $\bd(\mo)$.} Let $e_i$ be the edge of $\mo$ on which the witness $w_i$ lies. Then $\vis(\Rb(w_{i-1}))\vert_{e_i} \cup \vis(\Lb(w_{i+1}))\vert_{e_i} \neq e_i$.
 \end{clm}
 \ifthenelse{\boolean{shortver}}{}{
 \begin{proof}
 	Suppose for the sake of contradiction that $\vis(\Rb(w_{i-1}))\vert_{e_i} \cup \vis(\Lb(w_{i+1}))\vert_{e_i} = e_i$. This indicates that $w_i \in \vis(\Rb(w_{i-1}))\vert_{e_i} \cup \vis(\Lb(w_{i+1}))\vert_{e_i}$, suggesting that $w_i$ is visible from $\Rb(w_{i-1})$ or $\Lb(w_{i+1})$. {Consequently, this implies that there exists a point $x \in \Rb(w_{i-1})$ or a point $y \in \Lb(w_{i+1})$ such that $x$ (or $y$) is visible from $w_i$. Also, every point on $\Rb(w_{i-1})$ is visible from $w_{i-1}$ and every point on $\Lb(w_{i+1})$ is visible from $w_{i+1}$. So, this means that $(x \in)\vis(w_{i-1}) \cap \vis(w_i) \neq \emptyset$ or $(y \in)\vis(w_i) \cap \vis(w_{i+1}) \neq \emptyset$, which is}  a contradiction in either case.
 \end{proof}
 }
	\begin{definition}[$\reg(e),e_i,\reg_i, \reg^c_i $]\label{def:rege}
	{\em For an edge $e$ in $\mo$, $\reg(e)$ denotes the region determined by all the points lying {on} the edge $e$. For each witness $w_i$, if it is indeed an internal point of some edge in $\mo$, we use $e_i$ to denote such an edge that contains the witness $w_i$. For each integer $i \in \{2, \ldots, k-1\}$, we define two regions $ \reg_i $ and $\reg^c_i $ as follows.
	\begin{itemize}
			\item $\reg_i \coloneqq \vis(\Rb(w_{i-1}))\vert_{e_i} \cup \vis(\Lb(w_{i+1}))\vert_{e_i}$.
			
			\item $\reg^c_i \coloneqq \reg(e_i) \setminus \reg_i$.
	\end{itemize} }  
\end{definition}

  As $\tra(w,W)$ of each witness $w \in \wgb$ intersects $V(\mo)$,  the following holds.
  
  \begin{clm} \label{Type-4-witness}
  	$V(\mo)$ is a  $(W, \wgb)$-\pws.
  \end{clm}

Now we show that when $|W| \geq 2$, the size of the set $\wgb$ is at least two. To be precise, both the vertices $w_1$ and $w_k$ belong to $\wgb$. We prove this fact formally in \Cref{wonek}.

\begin{clm}\label{wonek}
	$\{w_1, w_k\} \subseteq  \wgb$
\end{clm}

\ifthenelse{\boolean{shortver}}{}{
\begin{proof}
	 First, we demonstrate that $w_1$ can be substituted with $w_1^*$, where $w_1^*$ belongs to $V(\mo)$. Symmetrically, this argument applies to $w_k$, the last witness from left to right. If $w_1$ is already in $V(\mo)$, we have nothing to prove and $w_1 = w_1^*$. Otherwise, if $w_1 \notin V(\mo)$, we determine $w_1^*$ through a two-step process.
	
	\begin{description}
		\item[Step 1. Finding a replacement of $w_1$ on $\bd(\mo)$.] If $w_1 \in \bd(\mo)$, then we are done. Otherwise, consider the chord in $\mo$ that passes through $w_1, \RA(w_1)$ and has one endpoint at $\IMr(w_1)$. Let $\ro(w_1)$ represent the other endpoint.  Clearly, $\ro(w_1) \in \bd(\mo)$. It is apparent that $\vis(\ro(w_1)) \subseteq \ter(w_1,W)$, as $w_1$ is the leftmost witness and $\Rb(w_1) = \Rb(\ro(w_1))$. Thus, $w_1$ has a substitute on $\bd(\mo)$ concerning $W$, and $\ro(w_1)$ is one such substitute.
		
		\smallskip 
		
		\item[Step 2. Finding a replacement of $\ro(w_1)$ on $V(\mo)$.] Consider $\ro(w_1)$ positioned on the edge $e_1 = \overline{u_1v_1}$ of $\mo$. We examine $\vis(\Lb(w_2))\vert_{e_1}$ and assert that it cannot {be} equal {to $\reg(e_1)$}. If it did, it would imply $w_1 \in \vis(\Lb(w_2))\vert_{e_1}$, {thereby} obtaining {(at least) one} point $x$ on $\Lb(w_2)$ {which sees $w_1$, and vice-versa. Now,} every point on $\Lb(w_2)$ is visible from $w_2$. Thus, this would mean $(x\in)\vis(w_1)\cap\vis(w_2)\neq \emptyset$, which would lead to a contradiction. Based on \Cref{obs:con}, $\vis(\Lb(w_2))\vert_{e_1}$ is a connected line segment, and since $\vis(\Lb(w_2))\vert_{e_1} \neq e_1$, {either} one of $u_1$ or $v_1$ must not be visible from $\Lb(w_2)$. Assume $v_1$ is not visible from $\Lb(w_2)$ (refer to Figure \hyperref[w1invertex]{12}). Consequently, a replacement for $\ro(w_1)$ in $V(\mo)$ is obtained, which serves as a substitute for $w_1$ in $\mo$ with respect to $W$.
         \end{description}
	Hence the proof.
	\end{proof}
}

\ifthenelse{\boolean{shortver}}{}{
	\begin{figure}[ht!]
		\centering
		\includegraphics[scale=.8]{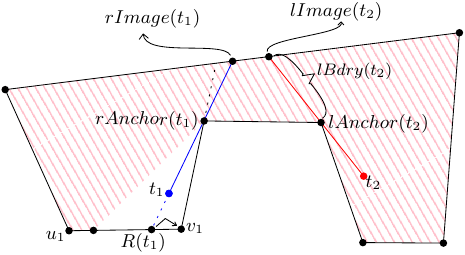}
		\caption{Illustration of the proof of \Cref{wonek}. Replacing $t_1$ by  $v_1\in V(\mo)$. The pink region denotes  visibility region of $\Lb(t_2)$.} \label{w1invertex} 
	\end{figure}
	
	}
  
  \subsubsection{Construction of  a $ (W,\wint)$-Potential Witness Set} The $\tra(w,W)$ of each witness $w \in \wint$ consists of an open line segment. It is important to note that there can be at most $r^2$ of these segments, with each $w \in \wint$ lying within one of these $r^2$ regions. Recall  that the $\tra(w,W)$ is associated with a pair of visible reflex vertices where $\LA(w), \RA(w)$, and $w$ are collinear, and $\LA(w), ~ \RA(w)$ belong to the different chains, {$\uc(\mo)$ and $\lc(\mo)$}. We know that $\vis(\tra(w,W)) \subseteq \ter(w, W)$, implying that for any point $z \in \tra(w)$, $ \vis(z) \subseteq \ter(w, W)$. We will now construct a point set, denoted as $\ro_{mid}$, in the following manner:

  $$\ro_{mid} = \bigl\{ p~:~p~\text{is the midpoint of}~\overline{ss'}~\text{for all}~s, s' \in \ro~\text{with}~s' \in \vis(s)\bigr\} $$

  Let $\ro_{mid}$ denote the set of all midpoints of segments connecting visible reflex vertex pairs. The $\ro_{mid}$ can be computed in $n^{\mathcal{O}(1)}$ time. The following observation is immediate.

% In essence, $\ro_{mid}$ comprises all possible midpoints of the line segments connecting every pair of visible reflex vertices. {The set $\ro_{mid}$ can be determined in $n^{\OO(1)}$ time.} The following observation is immediate.

\begin{observation}\label{obs:r2}
    $|\ro_{mid}| \leq r^2$.
\end{observation}

In the following \Cref{Type-1-witness}, we demonstrate that the collection $\ro_{mid}$ is indeed a $(W, \wint)$-\pws. Specifically, there exists a subset $Q' \subseteq \ro_{mid}$ such that: (i)  $(W \smallsetminus \wint) \cup Q'$ forms a witness set in $\mo$, and (ii)  $|(W \smallsetminus \wint) \cup Q'|$ does not exceed $|W|$. 
	
	\begin{clm} \label{Type-1-witness}
		 $\ro_{mid}$ is a  $(W, \wint)$-\pws.
	\end{clm}

\ifthenelse{\boolean{shortver}}{}{
\begin{proof}
 {We show that for any $w \in \wint$, we can find a substitute $w^* \in \ro_{mid}$ such that $(W \setminus \{w\} \cup w^*$ is a valid witness set of $\mo$. So, now let $w \in \wint$ be an arbitrary witness. We have that $\tra(w,W)$ is an open line segment $(\LA(w)\RA(w))$. In particular, the mid-point $p$ of $\overline{\LA(w)\RA(w)}$ is in $\tra(w,W)$. Therefore, we can successfully substitute $w$ by the point $p (= w^*$ in this case) to get a witness set of $\mo$ without changing the cardinality of $W$. This proves our claim.}
	\end{proof}
}

Using \Cref{obs:r2} and \Cref{Type-1-witness}, we arrive at the following lemma.

    \begin{lemma} \label{lem:intwitness}
		Consider a monotone polygon $\mo$ with $n$ vertices. For any hypothetical witness set $W$ in $\mo$ and  subset $\wint \subseteq W$ as specified in \Cref{def:witness}, there exists a $(W, \wint)$-\pws~of size at most $r^2$. Moreover, this set can be found in $n^{\OO(1)}$ time.
	\end{lemma}

    % \todo{Can we make this a Proposition or a Theorem?     Seems like an end result and I'm finding it difficult to refer     this elsewhere}

  \begin{mdframed}[backgroundcolor=red!10,topline=false,bottomline=false,leftline=false,rightline=false] 
	%\centering
	Observe that if there is a solution $W$ for the {\sc Witness Set} in $\mo$ such that $\wbb = \emptyset$, then we can solve the {\sc Witness Set} for $\mo$ in $n^{\OO(1)}$  time. This is true due to the fact that $Z \coloneqq V(\mo) \cup \ro_{mid}$ is a $W$-\pws~of $\mo$ (due to \Cref{Type-4-witness} and \Cref{Type-1-witness}) with $|Z| \leq n+r^2$ (due to \Cref{obs:r2}). 
	
\end{mdframed}

As every witness in $\wgb \cup \wint$ has a potential replacement in $V(\mo) \cup \ro_{mid}$, we can assume that for each witness $w_i \in W$ where $i \in [k]$, we have exactly one of the following. 
\begin{enumerate}
	\item $w_i \in V(\mo) \cup \ro_{mid}$, or 
	
	\item $w_i \in \wbb$, i.e., it has no substitute in $V(\mo) \cup \ro_{mid}$, moreover $w_i \in \reg(e_i) \subseteq \bd(\mo)$ {(as $\tra(w_i,W) \cap \bd(\mo) \neq \emptyset$)}. 
\end{enumerate}

\paragraph{Transforming into Sub-instances.}
Initially, we have a given monotone polygon $\mo$. Then we consider an hypothetical solution {of the {\sc Witness Set}} $W= \{w_1, \ldots, w_k\}$. We {have already shown} that if $\wbb = \emptyset$ we can solve the problem in  $n^{\OO(1)}$  time. Now we consider $\wbb \neq \emptyset$, and $ |\wgb \cup \wint| =(\ell <k) $ and $\wgb \cup \wint \coloneqq \{w_{i_1}, w_{i_2}, \ldots, w_{i_\ell}\}  $ where $i_1 < i_2 < \ldots < i_{\ell}$. It is clear from our above discussion that $\wint \cup \wgb \neq \emptyset$ as $w_1, w_k \in V(\mo)$ (by \Cref{wonek}), so $i_1=1$ and $i_{\ell}=k$. {So, a chain of witnesses $w_{i_j+1}, w_{i_j+2}, \ldots , w_{i_{j+1}-1}$ between two consecutive $w_{i_j},w_{i_{j+1}} \in \wgb \cup \wint$ lies in $\wbb$ for any $1 \leq j \leq l-1$.}
	
In the following, we give an algorithm that produces a set of {points} called $\zm$ such that $\zm$ serves as a $(W, \wbb)$-Potential Witness Set. {We show that (later in \Cref{sec:correctness}) for any such chain of witnesses $w_{i_j+1}, w_{i_j+2}, \ldots , w_{i_{j+1}-1} \in \wbb$, we can find a suitable replacement for them in $\zm$.} Basically we are trying to find a witness set $W'=\{w_i': i \in [k]\}$ such that $w_i' \in \zm$ when  $i' \notin \{i_1, i_2,  \ldots, i_{\ell}\}$ and {$w_i' \in V(\mo) \cup \ro_{mid}$} when  $i' \in \{i_1, i_2,  \ldots, i_{\ell}\}$. {Thus, so far, we have that our witness set in $\mo$ satisfies the following properties:}
  \begin{mdframed}[backgroundcolor=black!10,topline=false,bottomline=false,leftline=false,rightline=false] 
	%\centering
	 $W= \{w_1, \ldots, w_k\}$ is a  witness set in $\mo$ such that {$\{w_{i_1}, w_{i_2}, \ldots , w_{i_l}\} \subseteq V(\mo) \cup \ro_{mid}$ and each $w_i$,  $i \notin \{i_1, \ldots, i_l\}$} the point $w_i$ is in  $\reg(e_i)$ (as in \Cref{def:rege}) where $w_i$ could not be replaced by any vertex in   $ V(\mo) \cup \ro_{mid} $. {So $\wbb= W \smallsetminus \{w_{i_1}, \ldots , w_{i_l}\}$. From now on, we choose any particular chain of witnesses $w_{i_j+1}, w_{i_j+2}, \ldots , w_{i_{j+1}-1} \in \wbb$, and show that we can find a suitable replacement for each of them in $\zm$.} 
	
\end{mdframed}

 	%\subsubsection{ \texorpdfstring{\boldmath $ \PWS $}{Lg} of \texorpdfstring{\boldmath $ \wbb $}{Lg}}

\subsubsection{Construction of  a $(W, \wbb)$-Potential Witness Set} We have that each witness $w_i \in \wbb$ lies on $\bd(\mo)$. Below in \Cref{algo_1} we describe in more detail how to find the set $\zm$ which is a $(W,\wbb)$-\pws.
	
	\paragraph{Overview of \Cref{algo_1}.} Here we give a short description of \Cref{algo_1}. Let  $A_i$ denote the arrangement\footnote{Here, {by} arrangement, {we} mean adding some line segments in our instance in a specific way.} of $i$\textsuperscript{th} iteration and  $ Q_i $ denotes a set of points added  on $\bd(\mo)$. At the beginning, $A_0= \bd(\mo)$ and  $Q_0 = V(\mo)$. Finally we return the point set $Q_{2k}$ as $\zm$.  For each $i \in [2k]$, we compute the arrangement $A_i$ by creating chords that connect every possible pair of visible {points} $v$ and $\beta$, where $v \in Q_{i-1}$ and $\beta \in \ro$ (the set of reflex vertices in $\mo$). Define the set $Q_{i-1}^{\mathtt{chord}}$ in this manner: for each chord $L$ drawn in $A_i$, we include the intersection point $L \cap \bd(\mo)$ into $Q_{i-1}^{\mathtt{chord}}$. Let $Q_{i-1}^{\mathtt{midpt}}$ represent the set of all  midpoints (the exact central point of the line segment) between any two points $u$ and $v$ in $Q_{i-1}$, where both $u$ and $v$ are on the same edge of $\mo$ and are contiguous (i.e., the line segment $\overline{uv}$ does not contain any other point of $Q_{i-1}$). The set $Q_i$ is defined in the following way (see  \Cref{algo_1}).
    
    %(\Cref{fig-induction} for an illustration of the  iterations). 
    
    %\todo{Udvas modify this}
    
$$ Q_i= Q_{i-1} \cup Q_{i-1}^{\mathtt{midpt}} \cup Q_{i-1}^{\mathtt{chord}}$$

 %    We run this iteration $ 2k $ number of times, where $k= \wit(\po) $ (we guess this value)
	
 %     We will call all points (except the vertices of $\mo$) that are generated in the next iterations as \textit{induced vertices} of $\mo$, \hl{and the mid-points generated in the $i-th$ iteration as mid-points. In the $(i+1)th$ iteration we call the previous set of mid-points as induced vertices as they gets added to $Q_i$}. The procedure to obtain $ Q'' $ is described in  \Cref{algo_1} (\Cref{fig-induction} depicts an example of the iterations. \hl{Need to change this figure, the ordering of the figure, lemma and lemma proof}).
	% \todo{In algorithm 1,  $A_i$ is not justified.}

	\IncMargin{1em}
	\begin{algorithm}[ht!]
		\SetKwData{Left}{left}\SetKwData{This}{this}\SetKwData{Up}{up}
		\SetKwFunction{Union}{Union}\SetKwFunction{FindCompress}{FindCompress}
		\SetKwInOut{Input}{Input}\SetKwInOut{Output}{Output}
		\Input{A monotone polygon $ \mo $ and a non-negative integer  $ k $}
		\Output{A point set $\zm$}
		\BlankLine
		Initialization:  $ A_0 = \bd(\mo),  Q_0 = V(\mo)$ \;
		\For{$i\leftarrow 1$ \KwTo $2k$}{

         $Q_{i-1}^{\mathtt{midpt}} \coloneqq $ set of  all mid-points between a pair  adjacent points of $Q_{i-1}$ that lie on the same edge of $\mo$. \\
         $Q_{i-1}^{\mathtt{chord}} = \emptyset$\\ 
			For each possible pair of vertices $ v$ and $\beta $ where $ v \in Q_{i-1} $ and $ \beta \in \ro$ \tcp*[r]{\ \color{blue} $\ro$ denotes the set of all reflex vertices in $ V(\mo) $.} 
			{
				\If{\ $v$ is visible to $\beta $ }{\label{lt}
					{\ Add a maximal chord $ L $ in $ \mo $ joining $ v$ and $\beta $ in $A_i$} \tcp*[r]{\ \color{blue} Maximal means no other chord in $ \mo $ contains $ L $.}
					{\ Add the point $ L \cap \bd(\mo) $ into $ Q_{i-1}^{\mathtt{chord}} $.}
					%	{\ $ Q_i = Q_i \cup  \clo(Q_i)$}
					
				}
				\lElse{do nothing}

                 $Q_i= Q_{i-1} \cup Q_{i-1}^{\mathtt{midpt}} \cup Q_{i-1}^{\mathtt{chord}}$
			}
			
		}
		
		\Return $ \zm=Q _{2k} $
		\caption{\texttt{WitGen}$(\mo,k)$: Generation of $ \zm $}\label{algo_1}
	\end{algorithm}
	\DecMargin{1em}

% \begin{figure}[ht!]\label{}
%      \centering
%      \begin{subfigure}[b]{0.45\textwidth}
%          \centering
%          \includegraphics[width=\textwidth]{fig/initial poly.pdf}
%          \subcaption{$r_1,r_2$ are two reflex vertices.}
%          \label{step1alg}
%      \end{subfigure}
%      \hfill
%      \begin{subfigure}[b]{0.5\textwidth}
%          \centering
%          \includegraphics[width=\textwidth]{fig/processing step.pdf}
%          \subcaption{$Q_0^{midpt}$ consists of the purple points and $Q_0^{midpt}$ consists of the generated red points.}
%          \label{step2alg}
%      \end{subfigure}
%      \hfill
%      \begin{subfigure}[b]{0.7\textwidth}
%          \centering
%          \includegraphics[width=\textwidth]{fig/1stinterationfinal.pdf}
%          \subcaption{$Q_1$ consists of all the black points.}
%          \label{step3alg}
%      \end{subfigure}
%         \caption{}
%         \label{1stiteration}
% \end{figure}
%\todo[inline]{update the figure}
	
	\begin{observation}\label{obs:r3}
    $|\zm| \leq  |V(\mo)| \cdot (2+r)^{2k},$ where $r = |\ro|$. 
\end{observation}

	\ifthenelse{\boolean{shortver}}{}{
	\begin{proof}
		As $Q_{i}= Q_{i-1} \cup Q_{i-1}^{\mathtt{midpt}} \cup Q_{i-1}^{\mathtt{chord}}$. Now $|Q_{i-1}^{\mathtt{midpt}}| \leq |Q_{i-1}|$ and $|Q_{i-1}^{\mathtt{chord}}| \leq |Q_{i-1}| \cdot r$. This implies $|Q_i| \leq 2|Q_{i-1}| + |Q_{i-1}| \cdot r= |Q_{i-1}| \cdot (2+r)$. As $\zm= Q_{2k}$ and $|Q_0|= V(\mo)$, hence $\zm=  |V(\mo)| \cdot (2+r)^{2k}$.
	\end{proof}
    }
\begin{sloppypar}
    We will now demonstrate that the collection $\zm$ is indeed a  $(W, \wbb)$-\pws. Specifically, there exists a subset $Q \subseteq \zm$ such that: (i) the union $(W \smallsetminus \wbb) \cup Q$ forms a witness set in $\mo$, and (ii)  $| \wbb|=| Q|$ (as per \Cref{def-potential1}). The lemma is stated below, while the proof will be presented in the subsequent section (\Cref{sec:correctness}). We want to mention that we prove the following lemma as follows: for each {witness} $w \in \wbb$ there is a replacement ({which will also be referred to as "substitute")} $q$ of $w$ in $Q$, i.e., $(W \smallsetminus \{w\}) \cup \{q\}$ is a witness set.

    \end{sloppypar}
	
	\begin{lemma} \label{Type-2-witness}
		 $\zm$ is a  $(W, \wbb)$-\pws.
	\end{lemma}

Using \Cref{obs:r3} and \Cref{Type-2-witness}, we arrive at the following lemma.

    \begin{lemma} \label{lem:bdrywitness}
		Consider a monotone polygon $\mo$ with $n$ vertices. In  $  r^{\OO(k)} n^{\OO(1)}$ time we can construct a point set $\zm$ of size {at most} $ n \cdot (2+r)^{2k-1}$ such that for any  witness set $W$ in $\mo$ and its subset $\wbb \subseteq W$ as specified in \Cref{def:witness}, $\zm$ can serve as  a $(W, \wbb)$-\pws. Here $r$ denotes the number of reflex vertices in {$\mo$}.
	\end{lemma}
    
    %\todo{Can we make this a Proposition or a Theorem?     Seems like an end result and I'm finding it difficult to refer     this elsewhere}

  \begin{mdframed}[backgroundcolor=black!10,topline=false,bottomline=false,leftline=false,rightline=false] 
{In the next subsection},   we show that $\zm$ is $(W, \wbb)$-\pws. That is,  for each witness $w$ in $W$, there exists a $q \in \zm$ such that $(W \smallsetminus w) \cup \{q\}$ is a witness set in $\mo$. Hereafter, every point in $\zm \smallsetminus V(\mo)$ will be referred to as a {\textit{vertex} of} $\zm$. Unless {stated otherwise}, we will refer to each vertex in $V(\mo)$ simply as a {\em vertex}.
	
\end{mdframed}

	\subsubsection{Correctness}\label{sec:correctness}

\begin{sloppypar}

    \paragraph{Overview of the proof of \Cref{Type-2-witness}.}
    
    We first consider an Optimal Witness Set $W = \{w_1, w_2,  \ldots ,w_k\}$ of size $k$ in $\mo$ (Here, $x(w_1) < x(w_2) < \ldots  < x(w_k)$). Then, we try to show that for every witness $w \in W$, we can obtain a substitute $w^* \in Q = \ro_{mid} \cup \zm$ {($\zm$ contains $V(\mo)$)}. {Observe that, for all witnesses in $\wint \cup \wgb$, we have already found a suitable replacement in $V(\mo) \cup \ro_{mid}$ in $n^{\OO(1)}$ time. Let these witnesses be denoted by the set \{$w_{i_1}, w_{i_2}, \ldots , w_{i_l}$\}. Then, we pick a chain of witnesses $w_{i_j+1}, w_{i_j+2}, \ldots , w_{i_{j+1}-1} \in \wbb$ for some $1 \leq j \leq l-1$. For notational simplicity, from now onward, we will denote these witnesses by $t_2, t_3, \ldots t_{k'-1}$, respectively. Note that $t_1=w_{i_j}$ and $t_k'=w_{i_{j+1}}$ lies in $V(\mo) \cup \ro_{mid}$. In the proof provided below, we will show that we can find a replacement for $t_{i} \in \wbb$ in $\zm$ for all $1 < i <k'$. Since the choice of this {\em chain} of witnesses in $\wbb$ was done arbitrarily, our proof will hold for any such {\em chain} of witnesses in $\wbb$.} \\In {the beginning of our proof,} we introduce a crucial concept of \textit{line visibility}, which means the part of a line which is visible from another line of $\mo$ (\Cref{def_linevisibility}). {We then} proceed in the following way: {For finding the substitutes of the witnesses $t_i \in \wbb$ in $\zm$, we might not be able to find them directly in all cases. So, in those cases, we will create an intermediary set of points, known as {\em clone points} (\Cref{def_clone}), which will temporarily suffice as substitutes for the witnesses $t_2,\ldots,t_{k'-1}$.} We first show that a substitute of $t_2$ is present in {a clone point of a vertex in $\zm$, assuming the fact that $t_1 \in V(\mo) \cup \ro_{mid}$}. Next, we proceed to the next set of witnesses from left to right,{$t_3, t_4$}, and so on. {After finding the substitutes of the witnesses in these intermediate set of points, i.e., clone points, we} then proceed our arguments similarly again, by moving from right to left this time, {again assuming the fact that $t_k' \in V(\mo) \cup \ro_{mid}$}. That is, we try to find substitutes for $t_{k'-1}, t_{k'-2},$ and so on. We then show that, for those particular cases where we had chosen clone points as substitutes, we can find a substitute for the witness by choosing any point from an open line segment, the endpoints of which {lie on the set $\zm$}. So, we then substitute those witnesses by the mid-points of the open segments, which were also considered in our \Cref{algo_1}, {and get added to our final output $\zm = Q_{2k}$}. Therefore, this would prove that we can successfully find a substitute for any witness $t_i \in \wbb$ in $\zm$. We then conclude that our algorithm is indeed correct. The details of the proof are given below. {Here onward, unless mentioned otherwise, we will assume that as for any $t_i \in \wbb, ~\tra(t_i,W) \cap \bd(\mo) \neq \emptyset,$ $t_i$ lies on the edge $e_i = \overline{u_iv_i}$ of $\mo$.}

    \end{sloppypar}

    \smallskip
    
	\ifthenelse{\boolean{shortver}}{}{
We now present the detailed proof of \Cref{Type-2-witness}.

	%\begin{proof}

          \paragraph{Determining Substitutes for $\boldsymbol{t_2}$.}
          
       Here, we show that $t_2$ can be substituted with a clone of a vertex from $Q_2 $. By our definition,  $e_2 = \overline{u_2v_2}$ is the edge of  $\mo$ on which the witness $t_2$  lies. Consider the pair of  segments $\Rb(t_1)$ and $\Lb(t_3)$, and their corresponding visibility regions $\vis(\Rb(t_1))$ and $\vis(\Lb(t_3))$. We begin by noting a few observations as outlined below.

        \begin{observation}\label{obs-bothendptsvisible}
            $|\vis(\Rb(t_1)) \cap \{u_2, v_2\}| \leq 1$ and $|\vis(\Lb(t_3)) \cap \{u_2, v_2\}| \leq 1$.  
        \end{observation}

\ifthenelse{\boolean{shortver}}{}{
        \begin{proof}
           Our objective is to show that {at most} one of $u_2$ and $v_2$ can be visible from $\Rb(t_1)$ (or equivalently, from $\Lb(t_3)$). For the sake of contradiction, let's assume that $|\vis(\Rb(t_1)) \cap \{u_2, v_2\}| = 2$. We have established that for any two lines, $L_1$ and $L_2$, the visibility $\vis(L_1)\vert_{L_2}$ remains connected, according to \Cref{obs:con}. Consequently, this implies {that} every point on $e_2$ is visible from $\Rb(t_1)$, or in notation, $\vis(\Rb(t_1))\vert_{e_2}=\reg(e_2)$, contradicting \Cref{edgevis}. The proof to show $|\vis(\Lb(t_3)) \cap \{u_2, v_2\}| \leq 1$ follows a similar {argument}.
        \end{proof}
}

        \begin{observation}\label{obs-openconnected}
             If both the points $u_2$ and $v_2$ lies in the region $\reg_2$ (as per  \Cref{def:rege}) then the region $\reg^c_{2}$ is a non-empty, {open} and connected subregion of $\reg(e_2)$. 
        \end{observation}

\ifthenelse{\boolean{shortver}}{}{
        \begin{proof}
          Recall {in \Cref{def:rege}} that $\reg_2= \vis(\Rb(t_1))\vert_{e_2} \cup \vis(\Lb(t_3))\vert_{e_2}$. According to \Cref{obs:con}, both $\vis(\Rb(t_1))\vert_{e_2}$ and $\vis(\Lb(t_3))\vert_{e_2}$ are closed and connected segments. The region $\reg^c_{2}$ is formed by subtracting two closed connected segments, one containing $u_2$ and the other containing $v_2$ (both of them simultaneously cannot be contained in a single closed connected segment due to \Cref{obs-bothendptsvisible}), where $u_2$ and $v_2$ are the endpoints of $e_2$, resulting in an open, connected segment within {$\reg(e_2)$}. Furthermore, \Cref{edgevis} indicates that $\reg^c_{2}$ is not empty.
        \end{proof}
}

        \begin{observation}\label{open-connectedw2}
             For every  point $x\in \reg^c_{2}$ we have $\vis(x)\cap\vis(t_i)=\emptyset$ where $i \in \{1,3\}$.
        \end{observation}

\ifthenelse{\boolean{shortver}}{}{
        \begin{proof}
          If $\vis(x) \cap \vis(t_1) \neq \emptyset$, it follows from \Cref{lem-cl2} that $\Rb(t_1)$ and $\Lb(x)$ intersect. This indicates that $x$ is visible from $\Rb(t_1)$ along the line $\LC(x)$, meaning $x \in \vis(\Rb(t_1)) \vert_{e_2}$. Consequently, $x \notin \reg^c_{2}$, leading to a contradiction. An analogous argument applies to $t_3$ to prove that $\vis(x)\cap\vis(t_3)=\emptyset$.
        \end{proof}
    }
    
%    \begin{figure}[ht!]              
%    	\centering  
%    	\includegraphics[width=0.8\linewidth]{fig/atmostonevisible.pdf} 
%    	\caption{At most one of $\{u_2, v_2\}$ is visible from $\Rb(t_1)$ and $\Lb(t_3)$. Here, the vertex $v_2$ is not visible from either $\Rb(t_1)$ or $\Lb(t_3)$.}
%    	\label{atmostonevisible}            
%    \end{figure}

      According to \Cref{open-connectedw2}, $t_2$ can be substituted with any point from $\reg^c_{2}$. Our primary objective is to demonstrate the existence of a point $q \in Q_2 $ that allows $t_2$ to be replaced by a clone of $q$. It should be noted that our ultimate objective is to show that $t_2$ has a replacement within $\zm$. Note that we are dealing with the situation when both  $u_2$ and $ v_2$ are contained in the region defined by union of $\vis(\Rb(t_1))\vert_{e_2} $ and $ \vis(\Lb(t_3))\vert_{e_2}$. Here  $\Rb(t_1)$ sees one vertex of $\{u_2, v_2\}$  whereas $\Lb(t_3)$ sees the other. So, w.l.o.g., assume that $u_2 \in \vis(\Rb(t_1))\vert_{e_2}$ and $v_2 \in \vis(\Lb(t_3))\vert_{e_2}$. Recall that by \Cref{obs-openconnected}, the region $\reg^c_{2}$ is a non-empty connected open segment. Now we deal with  the following two subcases.

        \begin{description}
			\item [\textbf{Case I.} $\vis(\Rb(t_1))\vert_{e_2} = \{u_2\}$.\label{CaseI}] Here the  only point  of the edge $e_2$ that is visible from $\Rb(t_1)$ is $u_2$. Consequently, there is an open {connected} segment in $\reg(e_2)$ {beside\footnote{{by this, we mean that the open, connected segment $\reg_2^c$ has one of its endpoints as $u_2$.}} $u_2$} within which $t_2$ can be substituted with any point. Thus, $t_2$ can be swapped with either $u_2^+$ or $u_2^-$ (i.e., clone point), depending on which is positioned on the edge $e_2$.(See \Cref{singlepointvisible})

            \medskip

          \item [\textbf{Case II.} $\{u_2\} \subset \vis(\Rb(t_1))\vert_{e_2} $.\label{CaseII}] Let $\vis(\Rb(t_1))\vert_{e_2}=\overline{u_2u_2'}$, where $u_2' \neq u_2$ and is located in $e_2$. In this situation, there is an open {connected} segment within $\reg(e_2)$ {beside $u_2'$} that allows $t_2$ to be substituted by any point within this segment. Therefore, we have the option to substitute $t_2$ with either $u_2'^+$ or $u_2'^-$ (a clone point), {one of which} does not fall within $\vis(\Rb(t_1))\vert_{e_2}$. Now, w.l.o.g., assume $x(u_2')<x(u_2)$. For this,  we can substitute $t_2$ with $u_2'^-$ (see  \Cref{segmentvisible}). {Otherwise, if $x(u_2) < x(u_2')$, then we can substitute $t_2$ with $u_2'^+$}. 
		\end{description}
	
	\begin{figure}[ht!]\label{w2-vertexreplacement}
		\centering
		\begin{subfigure}[b]{0.45\textwidth}
			\centering
			\includegraphics[width=\textwidth]{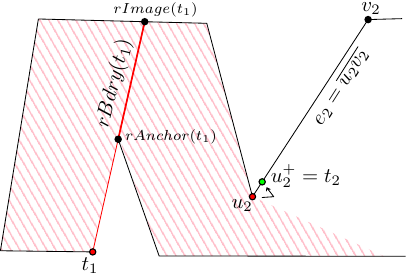}
			\subcaption{\hyperref[CaseI]{Case I} :  $u_2^+$ is a replacement of $t_2$.}
			\label{singlepointvisible}
		\end{subfigure}
		\hspace{8mm}
		\begin{subfigure}[b]{0.45\textwidth}
			\centering
			\includegraphics[width=\textwidth]{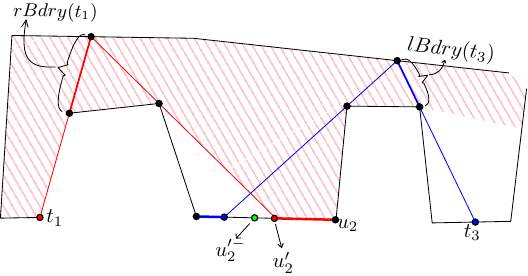}
			\subcaption{\hyperref[CaseII]{Case II}:  $u_2'^-$ is a replacement of $t_2$.}
			\label{segmentvisible}
		\end{subfigure}
		\caption{Pink regions are the visibility regions of $\Rb(t_1)$.}
	\end{figure}

       In the following (\Cref{claim:u2}), we demonstrate that the {point} $u_2'$ introduced in Case \hyperref[CaseII]{II} is indeed a vertex of $\zm$ (more specifically, a point of $Q_2$). This will establish that $t_2$ can be replaced by a clone point of $ Q_2$. Before that, we first make an important observation here.
       
       \begin{observation}\label{obs-rImageinZ_mid}
           Both $\RA(t_1)$ and $\IMr(t_1)$  is a vertex of $\zm$ (similarly, $\LA(t_k)$ and $\IMl(t_k)$ for $t_k$, and consequently for any $w_i \in V(\mo) \cup \ro_{mid}$, their respective $\RA(w_i), \LA(w_i)$ are vertices of $\zm$).
       \end{observation}

\ifthenelse{\boolean{shortver}}{}{
       \begin{proof}
          {Clearly, as $\RA(t_1) \in \ro \subseteq V(\mo) = Q_0$, $\RA(t_1) \in \zm$. Now, we need to show that $\IMr(t_1)$ is also a vertex of $\zm$. We know that $t_1 \in V(\mo) \cup \ro_{mid}$. Therefore, in either case, $\RC(t_1)$ passes through two vertices of $\mo$, which are visible to each other. That is, if $t_1 \in V(\mo)$, then $\RC(t_1)$ passes through the two mutually visible vertices $t_1$ itself, and $\RA(t_1)$. Otherwise, if $t_1 \in \ro_{mid}$, it implies that $t_1, \RA(t_1)$ and $\LA(t_1)$ are collinear; so, $\RC(t_1)$ passes through two mutually visible vertices $\LA(t_1)$ and $\RA(t_1)$. As $\IMr(t_1)$ is an endpoint of $\RC(t_1)$ on $\bd(\mo)$, distinct from $t_1, \LA(t_1)$ and $\RA(t_1)$, $\IMr(t_1) \in Q_1 \subseteq \zm$.}
       \end{proof}
       }
  
        \begin{clm}\label{claim:u2}
            $u_2'  $ is  a   vertex of $\zm$. More  specifically $u_2' \in Q_2$.
        \end{clm}

        \begin{sloppypar}

        \ifthenelse{\boolean{shortver}}{}{
\begin{proof}
            Since $u_2'$ is visible from $\Rb(t_1)$, {$u_2'$ also sees some point(s) of $\Rb(t_1)$.} So, we examine the sub-region of $\Rb(t_1)$ that is  visible  from $u_2'$, namely $\vis(u_2')\vert_{\Rb(t_1)}$. Now, w.l.o.g., let us assume that $x(u_2') < x(u_2)$. The other case is symmetric.
            \begin{description}
                \item[\textbf{Case A:} $u_2'$ sees exactly one point on $\Rb(t_1)$.\label{CaseA}] Let $c$ be the unique point on $\Rb(t_1)$ that is visible from $u_2'$, i.e., $\vis(u_2')\vert_{\Rb(t_1)} = \{c\}$. We now differentiate our cases based on whether the point  $c$ is a point of the set $\{\RA(t_1), \IMr(t_1) \}$.
 \begin{figure}[ht!]\label{fig-Steiner-vertex}
     \centering
     \begin{subfigure}[b]{0.45\textwidth}
         \centering
         \includegraphics[width=\textwidth]{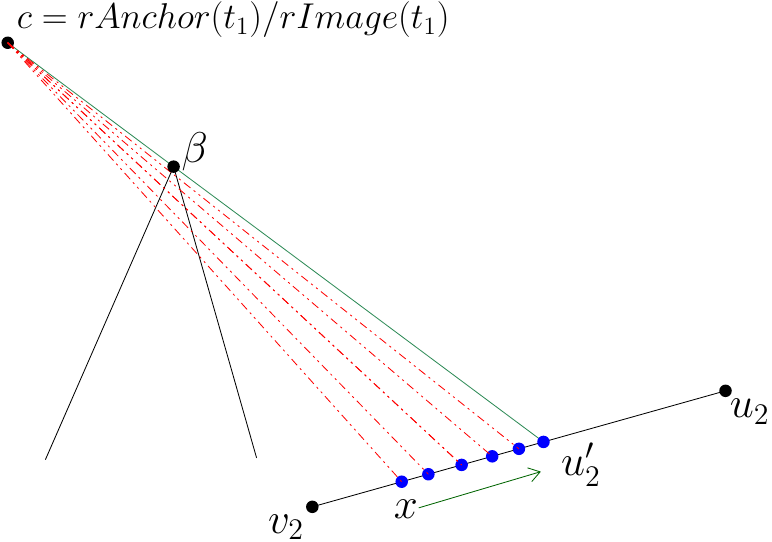}
         \subcaption{\textbf{Case} \hyperref[CaseA1]{A1}: $c = \RA(t_1)$ or $\IMr(t_1)$.}
         \label{figcaseA1}
     \end{subfigure}
     \hfill
     \begin{subfigure}[b]{0.45\textwidth}
         \centering
         \includegraphics[width=\textwidth]{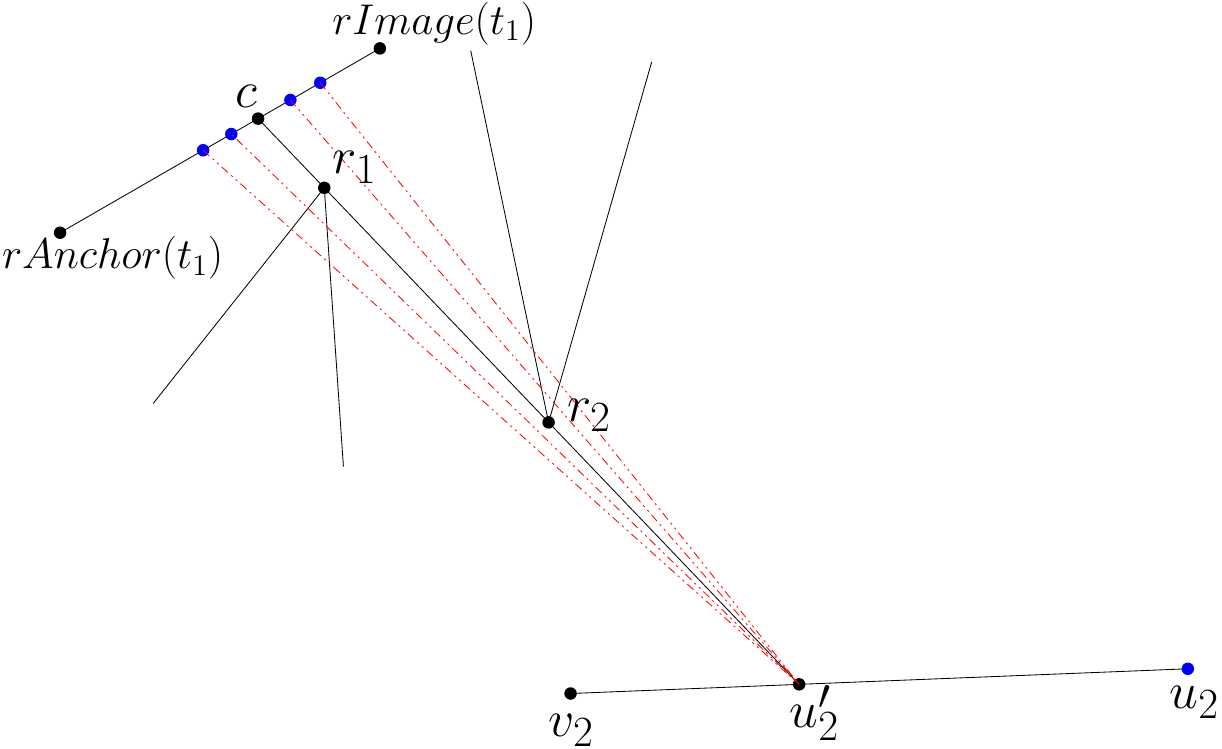}
         \subcaption{\textbf{Case} \hyperref[CaseA2]{A2}:  $c \neq \RA(t_1)$ or $\IMr(t_1)$.}
         \label{figcaseA2}
     \end{subfigure}
        \caption{$u_2'$ is a vertex of $\zm$.}
        \label{inducedvertex}
\end{figure}               
                \begin{description}
                    \item[Case A1: $c$ is either of $\RA(t_1)$ and $\IMr(t_1)$.\label{CaseA1}] $u_2$ is a vertex on the boundary \(\bd(\mo)\). Given that \(c\) is a point either from \(\RA(t_1)\) or \(\IMr(t_1)\), to demonstrate that \(u_2'\) is a {vertex} of $\zm$, it suffices to show that the segment \(\overline{cu_2'}\) intersects a reflex vertex in \(\mo\) and is wholly contained within \(\mo\) {(as $\RA(t_1)$ and $\IMr(t_1)$ lies in $\zm$, due to \Cref{obs-rImageinZ_mid})}. Observe that \(c\) is visible from \(u_2'\) and vice versa, yet no point left of \(u_2'\) is visible from \(c\) due to \(x(u_2')<x(u_2)\), and only \(\overline{u_2'u_2}\) of \(e_2\) is visible from \(\Rb(t_1)\). Thus, if \(c\) is joined to any point \(x \in e_2\) left of \(u_2'\), the segment \(\overline{cx}\) cannot be entirely inside \(\mo\). Consequently, as mentioned in \Cref{obs-hill}, there exists {an obstacle, i.e., a hill} containing a reflex vertex that obstructs visibility from \(x\) to \(c\) {(here, the hills might be different for different $x$, but that won't matter in our arguments below)}. However, since \(u_2'\) is visible from \(c\), transforming \(\overline{cx}\) towards \(\overline{cu_2'}\) as \(x\) moves along $\overline{v_2u_2}$ means the segment becomes part of \(\mo\) when \(x\) reaches \(u_2'\). This implies that \(\overline{cu_2'}\) must intersect reflex vertex \(r\) since \(\overline{cx}\) is wholly contained in \(\mo\) only when the segment \(\overline{cx}\) intersect exactly at the point {of intersection of two edges of a hill}, which is precisely a reflex vertex $\beta$(refer to \Cref{figcaseA1}). Now if \(c = \RA(t_1)\), then \(c\) itself is a reflex vertex of \(\mo\). If \(c = \IMr(t_1)\), then $c \in Q_1 \subseteq \zm$ as \(t_1 \in \ro_{mid} \cup V(\mo)\), (note that this  \(t_1\) {was obtained by replacing $w_{i_j}$ by a point in $\ro_{mid} \cup V(\mo)$)}. Therefore, as \(\overline{cu_2'}\) passes through a reflex vertex \(r\), it signifies that \(u_2' \in Q_2 \subseteq \zm\).

                   \smallskip
                    \item[{Case A2. $c$ is neither of $\RA(t_1)$ and $\IMr(t_1)$.}\label{CaseA2}]
                 In this scenario where the point $c \neq \RA(t_1)$ or $\IMr(t_1)$, we aim to establish that \(u_2'  \in \zm \). This will be shown by demonstrating that the line segment \(\overline{cu_2'}\) intersects a pair of distinct reflex vertices of \(\mo\) and remains completely contained in \(\mo\). Since $c$ is the sole point on $\Rb(t_1)$ that is visible from $u_2'$, there are no visible points to either the left or right of $c$ {on $\Rb(t_1)$} from $u_2'$. As argued in \textbf{Case} \hyperref[CaseA1]{A1}, the {non-visibility} of any point to the left of $c$ {on $\Rb(t_1)$} from $u_2'$ implies that \(\overline{cu_2'}\) passes through a reflex vertex, say $r_1$ of $\mo$. Likewise, since no point to the right of $c$ {on $\Rb(t_1)$} is visible from $u_2'$, a similar rationale suggests it also must intersect another reflex vertex, say $r_2$, distinct from $r_1$, in $\mo$. {The fact that the two reflex vertices $r_1$ and $r_2$ are distinct is guaranteed by the divergent orientation of the hills containing them,} as illustrated in \Cref{figcaseA2}. Therefore, $u_2'$ is derived by connecting two reflex vertices, $r_1$ and $r_2$, confirming {that} \(u_2' \in \zm\).
    \end{description}
            \end{description}    

                \begin{description}
                
                \item[\textbf{Case B:} $u_2'$ sees more than  one point on $\Rb(t_1)$.\label{CaseB}] Let $\vis(u_2')\vert_{\Rb(t_1)} = \overline{rs} \subseteq \Rb(t_1)$, i.e., $u_2'$ sees a part of (segment) or whole of $\Rb(t_1)$. 
\begin{figure}[ht!]
    \centering
    \includegraphics[width=0.5\linewidth]{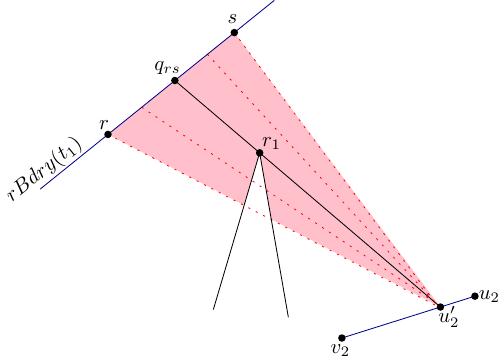}
    \caption{\textbf{Case} \hyperref[CaseB]{B}: $u_2'$ sees a part of (segment) or whole of $\Rb(t_1)$.}
         \label{figcaseB}
\end{figure}

                 In this case, we essentially show that this scenario does not arise. Assume that either $x(r) = x(s)$ and $y(r) < y(s)$ or $x(r) < x(s)$. {But $x(r) = x(s)$ is not possible because that would mean that the entire $\RC(t_1)$ is a vertical line (with same $x$-coordinates); and as $\RC(t_1) \subsetneq \mo$ intersects $\bd(\mo)$ on at least  three points, namely $\IMr(t_1),\ro(t_1)$ (the two distinct endpoints of $\RC(t_1)$) and $\RA(t_1)$, it contradicts the fact that the polygon $\mo$ is $x$-monotone.} So, we only consider the case where $x(r) < x(s)$.
                 
                 Now, we have that the entire $\overline{rs}$ is visible from $u_2'$. Therefore, if we join any point $q_{rs}$ on $\overline{rs}$ from $u_2'$, the line $\overline{u_2'q_{rs}}$ lies in $\mo$. So, essentially the solid triangle $\triangle{u_2'rs}$ lies entirely inside $\mo$. Let us now consider any point $q_{rs}$ on the open segment $(r,s)$. In particular, for the sake of clarity, let $q_{rs}$ be the mid-point of the line $\overline{rs}$. Note that this point $q_{rs}$ is visible from $u_2'$ and vice-versa. Since no point to the left of $u_2'$ is visible from $q_{rs}$, then, by the argument of \textbf{Case} \hyperref[CaseA1]{A1} we will get a reflex vertex $r_1$ through which $\overline{q_{rs}u_2'}$ must pass ($r_1$ must lie on the open segment $(q_{rs},u_2')$). But then, observe that a part of the hill which contains the reflex vertex $r_1$ must lie inside the $\triangle{u_2'rs}$ (See \Cref{figcaseB}). But this contradicts the fact that $\triangle{u_2'rs}$ lies entirely inside $\mo$. Thus, we conclude that this \textbf{Case} \hyperref[CaseB]{B} does not occur.

                 %   \item Otherwise, we are in the case that  neither $\RA(t_1)$ nor $\IMr(t_1) \in \vis(u_2')\vert_{\Rb(t_1)}$ (See \Cref{figcaseB}). We have, $\overline{xy}$ is entirely visible from $u_2'$. Suppose that $x$ lies to the left of $y$. As no point to the left of $x$ and no point to the right of $y$ is visible from $u_2'$, by the arguments provided in \textbf{Case} \hyperref[CaseA1]{A1}, we must have two reflex vertices $r_1$ and $r_2$, such that $\overline{u_2'x}$ and $\overline{u_2'y}$ passes through them, respectively. Now, since no point to the left of $u_2'$ is visible from $y$, then again, by the argument of \textbf{Case} \hyperref[CaseA1]{A1} we will get another reflex vertex $r_3$ through which $\overline{yu_2'}$ must pass. So, as $\overline{yu_2'}$ passes through two reflex vertices $r_2$ and $r_3$, $u_2'$ must be a Steiner vertex. Realistically though, this case is not possible, because, as entire $\overline{xy}$ is visible from $u_2'$, the triangle $\triangle{u_2'xy}$ must lie entirely inside $\mo$. So, the side $\overline{yu_2'}$ cannot pass through two reflex vertices where the edges containing the reflex vertex lies on either side of $\overline{yu_2'}$.
      
            \end{description}  

            This completes the proof that $u_2' \in \zm$. Now, as per our description of \Cref{algo_1}, the point $u_2'$ has been identified during the arrangement $A_2$, hence  $u_2' \in Q_2$. 
        \end{proof}}
        \end{sloppypar}

        Thus, we got a replacement of  $t_2$ in a clone point of a vertex from $Q_2 \subseteq \zm$.

  \begin{mdframed}[backgroundcolor=black!10,topline=false,bottomline=false,leftline=false,rightline=false] 
 Notice that when attempting to find an alternative for $t_2$, the key requirement is having $\RA(t_1)$ and $\IMr(t_1)$ as a vertex of $Q_{1} $. {Here, this fact is ensured by \Cref{obs-rImageinZ_mid}}. Similarly, in the search for a substitute for $t_3$, it is crucial to have $\RA(t_2)$ and $\IMr(t_2)$ as vertices of $Q_{i} $ for some $i \in [2k]$, although acquiring these is not as straightforward as obtaining  $\IMr(t_1)$ since the substitute of $t_2$ may be a clone point which is not a vertex of $V(\mo) \cup  \zm$.
\end{mdframed}

Let $t_2^*$ (a clone point) be a valid substitute for $t_2$ within $W$. In the following section, we will determine the substitute for $t_3$ by considering all possible positions of $t_2^*$.

          \paragraph{Determining Substitutes for $\boldsymbol{t_3}$.}
          
Here, we show that $t_3$ can be substituted with a clone point of a vertex from $Q_4 $. Recall that $t_3 \in \wbb$ and $t_2^* = u_2'^+/u_2'^-$, where $u_2'$ is a vertex of $Q_2 $.  Also $t_3$ lies on the edge $e_3 = \overline{u_3v_3}$ of $\mo$.  We first observe the following.
%Now consider the sub-region {$\reg_3 =$} $\vis(\Rb(t_2^*))\vert_{e_3} \cup \vis(\Lb(t_4))\vert_{e_3}$ in $\reg(e_3)$.   

\begin{observation}\label{obs-t3clone}
	$  \vis(\Rb(u_2'))\vert_{e_3} \cup \vis(\Lb(t_4))\vert_{e_3} \neq \reg(e_3)$ 
\end{observation}

\begin{proof}
	If $\vis(\Rb(u_2'))\vert_{e_3} \cup \vis(\Lb(t_4))\vert_{e_3} = \reg(e_3)$, then, it means that $t_3$ is visible from either $\vis(\Rb(u_2'))\vert_{e_3}$ or $\vis(\Lb(t_4))\vert_{e_3}$. But $t_3$ cannot be visible from any point of $\Lb(t_4)$. So, $t_3$ must be visible from some point of $\Rb(u_2')$. Using similar kind of {argument in} \Cref{open-connectedw2} we get that  $\Rb(u_2') \cap \Lb(t_3) \neq \emptyset.$ So, it means $u_2'$ is visible from $\Lb(t_3)$. Also, $u_2'$ is visible from $\Rb(t_1)$. But then, we recall the case where we needed $t_2$ to be replaced by a clone point. In that case, $\Rb(t_1)$ sees a vertex of $e_2$ and $\Lb(t_3)$ sees the other vertex of $e_2$. But as both intersect in at least one point, i.e., $u_2'$, and, as $\vis(\Rb(t_1)\vert_{e_2}, \vis(\Lb(t_3)\vert_{e_2}$ are connected, this violates \Cref{edgevis}. Thus, we get a contradiction. {[Note that as $\vis(\Rb(u_2')\vert_{e_3}, \vis(\Lb(t_4)\vert_{e_3}$ are closed, the non visible region of $\reg(e_3)$ from $\Rb(u_2')$ and $\Lb(t_4)$ is open.]}
\end{proof}

Using a similar argument, we can claim the following as well.
\begin{observation}\label{obs-t3readjustment}
	$  \vis(\Rb(t_2^*))\vert_{e_3} \cup \vis(\Lb(t_4))\vert_{e_3} \neq \reg(e_3)$ 
\end{observation}

Let $\reg_{2,4} \coloneqq \vis(\Rb(t_2^*))\vert_{e_3} \cup \vis(\Lb(t_4))\vert_{e_3}$ {and $\reg_{2,4}^c \coloneqq \reg(e_3) \setminus (\vis(\Rb(t_2^*))\vert_{e_3} \cup \vis(\Lb(t_4))\vert_{e_3})$}. If at most one of $\{u_3, v_3\}$ is contained in $\reg_{2,4}$, then  the region $\tra(t_3, W)$ contains either $u_3$ or $v_3$. Hence, we have shown that $V(\mo)$ contains a replacement of $t_3$. {But as we have already assumed that $t_3 \in \wbb$, this case doesn't arise.} So, we consider that $\{u_3, v_3\} \subseteq \reg_{2,4}$.

\medskip

 \noindent We know that $t_2^* = u_2'^+/u_2'^-$. For the time being, let us assume that $t_2^*$ is $u_2'$, a vertex of $ Q_2 \subseteq \zm$. Then $t_3$ can be replaced in a manner analogous to our replacement of $t_2$, where we already knew the substituted position of $t_1=w_{i_j}$. {So, temporarily, we do that itself {by assuming $t_2^* = u_2'$}. That is, we look at the closed, connected region on $\reg(e_3), \vis(\Rb(u_2')\vert_{e_3} = \overline{v_3v_3'}$ (say). Then, we replace $t_3$ by a clone of $v_3'$, either $v_3'^+$ or $v_3'^-$, whichever lies on the open connected region $\reg_{2,4}^c$ (Note that $\reg_{2,4}^c$ is an open connected region due to \Cref{obs-t3clone}). So, w.l.o.g., we assume that the substitute of $t_3$ is $v_3'^+$ (the other case will be analogous, as our arguments provided below will evidently show). Now, this arrangement of $u_2'$, $v_3'^+$ and $t_4$ will satisfy $\vis(u_2') \cap \vis(v_3'^+) = \emptyset$ and $\vis(v_3'^+) \cap \vis(t_4) = \emptyset.$ In particular, as $\vis(u_2') \cap \vis(v_3'^+) = \emptyset$, there exists an open region on $\reg(e_2)$ where any replacement of $u_2'$ would successfully preserve the non-intersection property with $\vis(v_3'^+)$. So, effectively, we can choose a clone point of $u_2'$, such that it still preserves the non-intersection property with $\vis(v_3'^+)$ (due to \Cref{obs-t3readjustment}). Since we can choose this clone point from an open segment around $u_2'$ on $\reg(e_2)$, we particularly choose one such that it lies to the left of $u_2'$ (recall that our original substitute of $t_2$ was also positioned on the left of $u_2'$). So, we can follow the same procedure of replacement as we had done for $t_2$, without affecting the non-intersection property with both $\vis(t_1)$ (as any point to the left of $u_2'$ that lies on $\reg(e_2)$ is not visible from $\Rb(t_1)$), as well as $\vis(v_3'^+)$. It may happen that after choosing $v_3'^+$ as our replacement for $t_3$, we have to update the previous choice of replacement of $t_2$, by choosing a point closer to $u_2'$, but even then it can still be considered as a clone point of $u_2'$.}
 {Thus, ultimately, this allows us to effectively replace $t_3$ with a clone point of a vertex of $ Q_4$, illustrated in \Cref{t3_replacement}.}

\begin{figure}
    \centering
    \includegraphics[width=.8\linewidth]{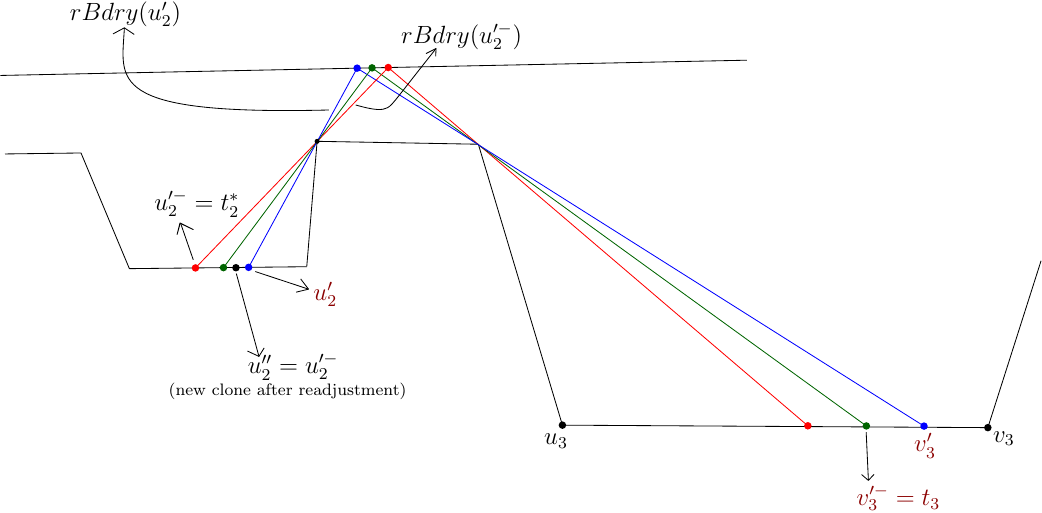}
    \caption{Readjusting the position of $t_2^*$ after replacing $t_3$ by a clone point. Here, our initial choice of substitute for $t_2$ was $u_2'^-$, which changes to $u_2''$ after replacing $t_3$ by $v_3'^-$.}
    \label{t3_replacement}
\end{figure}

   \paragraph{Determining Substitutes for $\boldsymbol{t_4}$ and onwards} 

 So, proceeding similarly, we replace the next set of witnesses $t_4$ onward. 
Recall that our goal is to demonstrate that each witness in $\wbb$ has an appropriate replacement in $\zm$. So, to do that, we had chosen a chain of such consecutive witnesses in $\wbb$, $t_2, t_3, \ldots , t_{k'-1}$ ($t_1$ and $t_k'$ lies in $V(\mo) \cup \ro_{mid}$). We then tried to find a suitable substitute for each $t_i ~(1 < i < k')$ in $\wbb$ in $\zm$. However, as described thus far, the replacement is not consistent with vertices from $V(\mo) \cup \zm$ and includes clone points. Some witnesses may find their replacements in clone points, yet our algorithm never produces clone points. Our next aim is to show that we can obtain a viable alternative for these clone points of the vertices of $z_{mid}$, in $\zm$ itself.
 
 \paragraph{Left $\rightarrow$ Right: Find a  replacement for $\boldsymbol{t_i}$ and updating all previous Clones for $\boldsymbol{t_{j}, j<i}.$}
 
 We have seen that when we try to find a replacement of $t_3$ as a clone point of some vertices of $Q_4 $, after finding the suitable clone, we {might have to readjust the} position of the clone point which {was initially considered to be a} replacement for $t_2$. Similarly, we proceed {to the next set of witnesses from left to right}. Assume that we are at $A_{2i}$\textsuperscript{th} arrangement where we aim to find a replacement of $t_i$. Once we find the clone point, which is a replacement $t_i$, we {might} need to update all the clone points corresponding to the replacement of each $t_j$ where $1 < j <i$. But we know that we do not have to find a replacement for $t_k'$ as $t_k'$ {has already been substituted by a point in} $V(\mo) \cup \ro_{mid}$. So in $A_{2i}$\textsuperscript{th} arrangement we actually find all the replacement of $t_i$ for each $i \in \{2, 3, \ldots, k'-1\}$ in some clone point, say $u_i^*$ (either $ u_i^+ $ or $ u_i^- $) where $u_i$ is a vertex obtained $A_{2i}$\textsuperscript{th} arrangement.
 
                   \begin{mdframed}[backgroundcolor=red!10,topline=false,bottomline=false,leftline=false,rightline=false] 
 	%\centering
 	Recall that our primary goal is to demonstrate that each vertex in $W$ has an appropriate replacement in $\zm$. In the following, we illustrate that witnesses with potential replacements in clone points also have substitutes in  $\zm$. More precisely, every clone point has a viable alternative in $ \zm$. Further details are provided below.
 \end{mdframed}
 
 %\todo{This box is probably not needed now as i've written a similar "what we have done so far above"}

  \paragraph{Right $\rightarrow$ Left: Generating replacement of Clones in $\boldsymbol{\zm}.$}

        Note that $t_k'$ is a {point} in $V(\mo) \cup \ro_{mid}$. Now, focus on the witness $t_{k'-1}$. According to the previously described method, $t_{k'-1}$ was initially replaced by a clone point of a vertex {of $\zm$}, say $p_{k'-1} (\in  Q_{2k'})$, {while moving from left to right}. Note that the vertex {of $\zm,$} $p_{k'-1}$, was determined by analyzing the visibility regions from $\Rb(t_{k'-2})$ on the edge $e_{k'-1} = \overline{u_{k'-1}v_{k'-1}}$ of $\mo$, where $t_{k'-1}$ lies. However, since $t_{k'}$ belongs to $V(\mo) \cup \ro_{mid}$, we can opt for a clone of another vertex {of $\zm$}, say  $q_{k'-1}$ from $ Q_{1}$ to do   a replacement of $t_{k'-1}$ where  $q_{k'-1}$ lies on the edge $e_{k'-1}$ and obtained in the arrangement $A_1$. Specifically, $q_{k'-1}$ is determined by examining the part of $e_{k'-1}$ that is visible from $\Lb(t_{k'})$.  The clone of $q_{k'-1}$ would adequately replace $t_{k'-1}$.  This approach is symmetric to the replacement method of $t_2$ {by using the fact that $t_1$ lies in $V(\mo) \cup \ro_{mid}$}. Now we can claim the following.

        %  \sj{UPTO HERE CLEAN - SATYA}

            \begin{clm}\label{midpt}
                These two choices of replacement of $t_{k'-1}$, i.e., clone of $p_{k'-1}$ (obtained from $\Rb(t_{k'-2})\vert_{e_{k'-1}}$) and clone of $q_{k'-1}$ (obtained from $\Lb(t_{k'})\vert_{e_{k'-1}}$) are distinct.
            \end{clm}
            
\ifthenelse{\boolean{shortver}}{}{
            \begin{proof}
                As we were not able to replace $t_{k'-1}$ with a vertex of $\mo$, it implies that $\Rb(t_{k'-2})$ sees exactly one vertex of $e_{k'-1}$ while $\Lb(t_{k'})$ sees the other. If both $p_{k'-1}, q_{k'-1} \in V(\mo)$, then they are the two endpoints of the edge $e_{k'-1}$ (since $p_{k'-1} \in V(\mo)$ is visible from $\Rb(t_{k'-2})$ while $q_{k'-1} \in V(\mo)$ is visible from $\Lb(t_{k'})$), and hence, are distinct. Otherwise, if $p_{k'-1} = q_{k'-1}$, then, $\overline{u_{k'-1}p_{k'-1}}$ is visible from $\Rb(t_{k'-2})$ and $\overline{v_{k'-1}q_{k'-1}}$ is visible from $\Lb(t_{k'})$, or, $\overline{v_{k'-1}p_{k'-1}}$ is visible from $\Rb(t_{k'-2})$ and $\overline{u_{k'-1}q_{k'-1}}$ is visible from $\Lb(t_{k'})$. But then, ${\reg_{k'-1} =} \vis(\Rb(t_{k'-2})\vert_{e_{k'-1}} \cup \vis(\Lb(t_{k'})\vert_{e_{k'-1}} = \reg(e_{k'-1})$, as $p_{k'-1} = q_{k'-1}$ {and the respective regions $\vis(\Rb(t_{k'-2})\vert_{e_{k'-1}}$ and $\vis(\Lb(t_{k'})\vert_{e_{k'-1}}$ are connected,} which contradicts \Cref{edgevis}. So, as $p_{k'-1}\neq q_{k'-1}$ we can claim that their respective clones are also distinct. 
            \end{proof}
}

Thus, we can replace $t_{k'-1}$ by a clone of $p_{k'-1}$, or a clone of $q_{k'-1}$ {with $p_{k'-1} \neq q_{k'-1}$}. Since we cannot replace the witness with a vertex of the edge $e_{k'-1}$, the region in which we can place a witness, i.e., $\reg_{k'-1}$ is open and connected (due to \Cref{obs-openconnected}). This means that we can replace $t_{k'-1}$ by any point that lie in between $\overline{p_{k'-1}q_{k'-1}}$. Specifically, we can replace $t_{k'-1}$ by the midpoint of $\overline{p_{k'-1}q_{k'-1}}$. (See \Cref{lastmid}) Now, as the mid-points of two {adjacent} vertices {of $\zm$} on the same edge of $\mo$ {in the $A_i\textsuperscript{th}$ arrangement} get added to our set $\zm$, we can consider $t_{k'-1}$ to be replaced by a vertex from $\zm$ which is the midpoint (say, $\alpha_{k'-1}$) between two adjacent vertices of $Q_{i}$ generated in the preceding step of \Cref{algo_1}) itself. Then, by following the same argument as we replaced $t_{k'-1}$, we move towards the set of witnesses to the left of $t_{k'-1}$ again. In more detail, once we try to find a replacement for $t_{k'-1}$ in $\zm$, specifically in $Q_{k'+1}^{\mathtt{midpt}}$, we use the position of the {point} $t_k'$. We call $t_k'$ is the  {\em important} {point} for $t_{k'-1}$. In the next steps, when trying to find a replacement of $t_{k'-2}$, we need such an important point. For that, we apply the same argument as before, considering the vertex $\alpha_{k'-1}$ as an important point for $t_{k'-2}$. Continuing this {way} we can get a set of {points} $\alpha_{k'-2}, \alpha_{k'-3}, ...$, and so on, which are distinct from the choices we had obtained by replacing the witnesses while moving from left to right. {Since choices $q_i$ and $p_i$ are distinct (due to \Cref{midpt}), existence of such an $\alpha_i \in \zm$ is always ensured.} Proceeding this way, we can get a choice of a substitute for any witness of $W$ by a vertex from $\zm$. From the above discussion, we obtain the following conclusions for each $i \in \{2,3, \ldots, k'-1\}$:
\begin{itemize}
	\item $\alpha_{i}$ is a mid-point of two vertices, say $p_i$ and $q_i$ of $\zm$ where $p_i \in Q_i$, $q_i \in Q_{2k-i}$.
	\item $\alpha_{i} \in \zm$. More specifically,  $\alpha_{i} \in Q_{2k-i}^{\mathtt{midpt}}$. 
	\item the point $\alpha_{i}$ is a replacement of the witness $t_i$ in $\zm$.
	\item the point $\alpha_{i+1}$ is an important point for $t_i$ ($\alpha_k'=t_k'$).
\end{itemize} 

 \begin{proof}[\textbf{Proof of \Cref{Type-2-witness}}]
{We have already shown that for a witness $w \in \wint \cup \wgb, ~w$ has a valid substitute in $V(\mo) \cup \ro_{mid}$. And, for any chain of consecutive witnesses $t_i \in \wbb (1<i<k')$, we have a valid substitute in $\zm$.} Therefore, every $w \in W$ has a valid substitute in $V(\mo) \cup \ro_{mid} \cup \zm$. This shows the correctness of our \Cref{algo_1}. Hence, the proof follows.
\end{proof}
            
	\begin{figure}[ht!]
	    \centering
	    \includegraphics[width=1\linewidth]{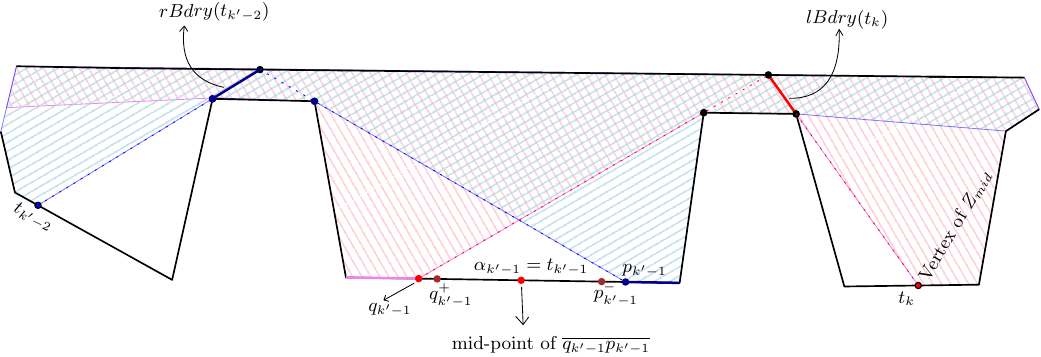}
	    \caption{Replacement of $t_{k'-1}$ with the mid-point of two adjacent vertices of $\zm$.}
	    \label{lastmid}
	\end{figure}
    }
	% \begin{corollary} \label{cor-optimum}
	% 	If $i \le k$ then $Q_{2i-2}^*$ contains a $ \ws(\mo) $ of size at least $ i $.  So size of the solution in $\wsp(Q_{2i-2}^*)$ is at least $i$.
	% \end{corollary}	

	% \begin{corollary} \label{cor-ptas}
	% 	Let $W^*$ be any $\ws(\mo)$ of size $k_1$ and $t_{\alpha} \in W^*$ be a witness which is $\wld$ in $Q_j^*$ with respect to $ W^*$. Now for $(j+s) \le k_1$,  $t_{\alpha+1}, t_{\alpha+2}, \cdots, t_{\alpha+s} \in W$ becomes  $\wld$ $\PW$ in $Q_{j+2}^*, Q_{j+4}^*, \cdots, Q_{j+2s}^*$, respectively. So we can replace $t_{\alpha+1}, t_{\alpha+2}, \cdots, t_{\alpha+s}$ by their respective $\PW$ in $Q_{j+2s}^*$ to obtain an alternative $\ws(\mo)$, say $W'$, such that  $|W'| =|W^*|$.
	% \end{corollary}	
	
	\paragraph{Running Time.} \label{sec-exact}
	
	The running time of the \Cref{algo_1} is influenced by both the number of arrangements and the cardinality of points generated in each arrangement. According to \Cref{obs:r3}, the number of points in $\zm$ is bounded by $r^{\OO(k)}n$, where $|V(\mo)|=n$ and $r$ represent the count of reflex vertices in $\mo$. Given that the number of arrangements is limited to $2k$, and each arrangement is quadratically dependent on the number of existing points, the running time is bounded by  $r^{\OO(k)} \cdot n^{\OO(1)}$.

   % \thfpt*

 \begin{proof}[\textbf{Proof of \Cref{exact-algorithm}}]
		Consider a hypothetical solution $W$ for the {\sc Witness Set} of $\mo$, divided into parts $W= \wint \cup \wgb \cup \wbb$, as defined in \Cref{def:witness}. According to \Cref{Type-4-witness}, we understand that $V(\mo)$ functions as a $(W, \wgb)$-\pws. Furthermore, as detailed in \Cref{lem:intwitness}, there is a $(W, \wint)$-\pws~with a size of at most $r^2$, which can be determined in $n^{\OO(1)}$ time. Additionally, by \Cref{lem:bdrywitness}, it is possible to generate a point set $\zm$ with a maximum size of $n \cdot (2+r)^{2k-1}$, ensuring that $\zm$ is a $(W, \wbb)$-\pws, and this set can be constructed within $r^{\OO(k)} n^{\OO(1)}$ time. Hence, we have shown that if the largest witness set has size $k$, then the set $V(\mo) \cup \ro_{mid} \cup \zm$ (as derived by \Cref{algo_1}) contains a witness set of the same size, $k$.
        
        Now, the algorithm for the {\sc Witness Set} in $\mo$ works as follows: Define $O_k$ as the output (the set of vertices) from \texttt{WitGen}$(\mo,k)$. We aim to identify the minimum $k$ (by guessing the value $k$) such that: (i) $V(\mo) \cup \ro_{mid} \cup O_k$ encompasses a witness set of size $k$, while (ii) $V(\mo) \cup \ro_{mid} \cup O_{k+1}$ does not contain any witness set of size $k+1$. The correctness of this approach is assured because if the maximum possible size of a witness set exceeds $k$, then $V(\mo) \cup \ro_{mid} \cup O_{k+1}$ must include a witness set of at least size $k+1$. Since the {\sc Discrete Witness Set} problem in a monotone polygon $\mo$ can be solved in $\OO(|S|^2 + |V(\mo)||S|)$ time (as demonstrated by \Cref{theo-finite-witness-in-M}), we thus infer that the {\sc Witness Set} problem in a monotone polygon is solvable in $r^{\OO(k)} \cdot n^{\OO(1)}$ time.
  	\end{proof}

\section{Approximation Algorithms for {\sc Witness Set} in Monotone Polygons} \label{sec-ptas}

\ifthenelse{\boolean{shortver}}{We present a factor-2 approximation and a \textsf{PTAS} for the {\sc Witness Set} problem in monotone polygons, both running in polynomial time. These algorithms adopt a discretization approach similar to the optimal algorithm, with a key difference: we perform a \emph{vertical decomposition} by drawing vertical chords $\ell(v)$ through each vertex $v \in V(\mo)$. Let $H = \bigcup_{v} \{h_v\}$ denote the set of endpoints $h_v$ of these chords.\footnote{For $\XMi(\mo)$ or $\XM(\mo)$, the chord may degenerate to a point (i.e., the vertex itself).} The rest of the arrangement generation follows as in the optimal case. Sorting the chords by $x$-coordinate yields $\ell(v_1) < \ldots < \ell(v_n)$.}{}

\ifthenelse{\boolean{shortver}}{}{Here we first obtain a polynomial-time factor-2 approximation and  a polynomial-time approximation scheme ({\sf PTAS}) for the {\sc Witness Set} in Monotone Polygons. Both algorithms employ a discretization method similar to that of the optimal algorithm.. The only difference in this discretization is as follows:  we first draw the vertical chords $\ell(v)$ through each vertex $v \in V(\mo)$ (vertical decomposition). Let $H=\cup_{v} \{h_v\}$ be the set of other end points $h_v$\footnote{The chord through
	$\XMi(\mo)$ or $\XM(\mo)$ could be a point which is the vertex itself.} of $\ell(v)$.  After this, the generation of the $i$-th iteration of the arrangement remains the same as the generation of arrangements in the optimal solution. We can sort all the vertical chords according to their $x$-coordinate. Let the ordering be $\ell(v_1) <  \ldots < \ell(v_n)$. The following observation is immediate.}

\begin{observation} \label{obs-two-chords}
	If $\ell(v_j)$ and $\ell(v_{j+1})$ are the two consecutive chords and $ p $ is a point in $ \mo $ that lies in the vertical strip $(\{z \in \mo: x(v_j) \leq x(z) \leq x(v_{j+1}) \})$ then $\ell(v_j),~\ell(v_{j+1}) \subseteq \vis(p)$.
\end{observation}

	\begin{definition}[$ \bm \wld $] \label{def-well}
	{\em Let $W$ be a witness set in $\mo$. For a witness $w \in W$, we say $w$ is {\em well-defined} in a set $Q $ with respect to $ W$ if $w$ has a  replacement  in $ Q$, that is, there exists a point   $ q \in   Q$ such that $(W \smallsetminus \{w\}) \cup \{q\} $ is a witness set in $\mo$.} 
	
	\end{definition}

\noindent 	Now, let $Q_1$ be the set of all points generated from the monotone polygon $\mo$  as follows: 
\begin{itemize}
	\item For each \( v \in V(\mo) \cup H \) and \( x \in \ro \), if \( v \) is visible to \( x \), add a maximal chord \( L \) in \( \mo \) joining \( v \) and \( x \), where "maximal" means no other chord in \( \mo \) contains \( L \).
 %For each pair of vertex $v \in V(\mo) \cup H$ and reflex vertex $x \in \ro$, if  $v$ is visible to $ x $ then add a maximal chord $ L $ in $ \mo $ joining $ v$ and $ x $. Here maximal means no other chord in $ \mo $ contains $ L $. 

	\item Add the point $(L \cap \bd(\mo)) \smallsetminus \{v\} $ into $ Q_{1}$.
\end{itemize}

\ifthenelse{\boolean{shortver}}{\subsection{Factor-2 Approximation and PTAS}}{

}

\ifthenelse{\boolean{shortver}}{}{
\subsection{Factor-2 Approximation}
}

\begin{lemma} \label{theo-2-factor}
	
	Any witness set $W$ of size $ k $ contains a subset $ W_{odd} \subseteqq W$  such that $ |W_{odd}| \geq  \frac{k}{2}$ and each witness in $ W_{odd} $ are $\wld$ in $Q_1$ with respect to $ W_{odd} $.
\end{lemma}

\ifthenelse{\boolean{shortver}}{}{

\begin{proof}
	Let $W=\{w_1,\cdots,w_{k}\}$ be a witness set  of size $ k $, with $ x(w_i) <x (w_{i+1}) $, where $ 1 \leq i < k $. Let $W_{odd}=\{ w_i ~|~ w_i \in W, 1\leq i \leq k$ and $i$ is odd$\}$. Obviously $ W_{odd} $ is a witness set in $\mo$ of size at least $ \frac{k}{2}$. Now we show each witness in $ W_{odd} $ is $\wld$ in $Q_1$ with respect to $ W_{odd} $, i.e., $Q_1$ is a  $ W_{odd}\text{-}\pws$ in $ Q_1 $. Surely, the witness $w_{1}$ and $w_k$ is $\wld$ in $V(\mo)$ and so has replacemnt  in $ Q_1 $. Now we consider the witnesses that are in between $w_{1}$ and $w_k$, i.e, witnesses from $W_{odd} \smallsetminus \{w_1, w_k\}$.

	 Let  $w_{2t+1}$ be an arbitrary vertex in $ W_{odd} $ for some positive integer $t$ where $ 1 < 2t+1 < k$.  Now there must exist a pair of consecutive vertical chords, say $\ell(v_{j})$ and $\ell(v_{j+1})$ such that $w_{2t}$ lies in the vertical strip formed by  $\ell(v_{j})$ and $ \ell(v_{j+1})$ with  $\ell(v_{j}), \ell(v_{j+1}) \subset \vis(w_{2t})$ (follows from \Cref{obs-two-chords}). Note that $w_{2t+1}$ may not be well-defined in $Q_1$ with respect to the witness set $W$. But  if we consider the witness $W \smallsetminus \{ w_{2t}, w_{2t+2}\}$, then we can show  that $w_{2t+1}$ is well-defined in $Q_1$ with respect to the witness set $W \smallsetminus \{w_{2t}, w_{2t+2}\}$. Below, we prove that.

Notice  that   $\vis(w_{2t-1}) \cap \ell(v_{j+1}) = \emptyset$ in fact  $\vis(w_{2t-1}) \cap \ell(v_{j}) = \emptyset$. Also $\vis(w_{2t+1}) \cap \ell(v_{j+1}) = \emptyset$. Now consider  $\LC(w_{2t+1})$. It must pass through some reflex vertex, say $x^*$. Now consider the line  $L$ (not segment) passing through  $w_{2t+1}$ and $x^*$. If we rotate the line $L$ anti-clockwise through $x^*$ (i.e.,  the point $x^*$ remains in the line when rotating),  it must hit some vertices in $V(\mo) \cup H$. Now consider the point $p$ where the rotating line hits some vertices in $V(\mo) \cup H$ at the very first time. Either $p \in V(\mo)$ or $p$  is a vertex of $H $, more specifically a vertex from the line segment $\ell(v_{j+1})$ or $\ell(v_{j+2})$ (this is true because while rotating it may happen that the $\LC(w_{2t+1})$ hits a vertex of $V(\mo) \cup H$ on the left side or right side). Now we can find a vertex $w^*$ in the rotated line such that $w^*$ is a replacement of $w_{2t+1}$ in $Q_1$ with respect to the witness set   $W \smallsetminus \{w_{2t}, w_{2t+2}\}$. This is true due the fact of combining (i) As $\vis(w_{2t-1}) \cap \ell(v_{j}) = \emptyset$ and $w^*$ can see at most one point of $ \ell(v_{j+1}) \cup  \ell(v_{j+2}) $ that is also from $H$, $\ell(v_{j}) \cap \vis(w^*) = \emptyset$ that indeed imply $\vis(w_{2t-1}) \cap \vis(w^*) = \emptyset$, similarly $\ell(v_{j+3}) \cap \vis(w^*) = \emptyset$ that indeed imply $\vis(w_{2t+3}) \cap \vis(w^*) = \emptyset$, (ii) as we rotate the line $L$  anti-clockwise through $x^*$, $\vis(w_{2t-1}) \cap \vis(w^*) = \emptyset$ and $\vis(w_{2t+3}) \cap \vis(w^*) = \emptyset$. 	
%\todo[inline]{While rotating it may happen that the $\LC(w_{2t+1})$ hits a vertex of $V(\mo) \cup H$ on the right side. We need to find a replacement of $w$ w.r.t the witness $W \smallsetminus \{w_{2t}, w_{2t+2} \}$}	
\end{proof}
}
%\sasanka{Spelling check, "conatined, withness" in the following para?}

% \todo[inline]{\textbf{REVIEWER'S COMMENT:} for the factor-2
% approximation, the authors consider only the witnesses with odd index and claim that for this set (which contains at
% least k/2 witnesses) each witness can be substituted by a witness from the point set ($Q_1$) defined by an
% arrangement. They claim that witness $w_{2t+1}$ can be substituted by a witness from $Q_1$, they consider the two
% polygon vertices between which $w_{2t}$ lies horizontally ($v_j, v_{j+1}$), and claim that the maximum vertical line
% segments through these vertices are not intersected by the visibility polygons of $w_{2t-1}$ and $w_{2t+1}$. Then,
% they consider a structure called the $\ell$ Chord of $w_{2t+1}$, which runs through a reflex vertex $x*$. They then
% rotate a line through $w_{2t+1}$ and $x*$ cow until they hit a vertex or a boundary point vertically above/below a
% vertex. They claim that this must lie on the vertical line through $v_{j+1}$ (in the way they describe the process the
% line could also hit a vertex first on the other side of $x*$. On this rotated line a replacement for $w_{2t+1}$ is claimed
% to lie, but intersection is not checked with the visibility polygon of $w_{2t+1}$'s neighbouring witnesses visibility
% polygons, but with those of $w_{2t}$, and the new witness is claimed to not even intersect $w_{2t+1}$'s visibility
% polygon. I cannot follow these arguments, and the same holds for the PTAS, which is based on the same idea.}

\noindent 	We have shown that $Q_1$ contains a witness set of size at least $k/2$ where  $k$ is the size of the maximum witness set. So even for the optimum value, the argument works. Hence, $Q_1$ contains a factor-2 solution (witness set). It remains to find a maximum witness set that is contained in $Q_1$. Now since $|Q_1|= \OO(nr)$ and the {\sc Discrete Witness Set} in monotone polygons is solvable in $\OO(|Q_1|^2 + |V(\po)||Q_1|)$ time (by \Cref{theo-finite-witness-in-M}).  We have the following theorem.

%\sasanka{Check spellings, Why 2-apximation? AI is showing some corrections for the following theorem, which I found interesting!!}

\begin{theorem} \label{theo-2-factor-time}
	There is a 2-approximation algorithm for the {\sc Witness Set} problem in monotone polygons running in time   $\OO(r^2n^2)$, where $ r$ and $n $ denote the number of reflex vertices and vertices, respectively, of the input monotone polygon.
\end{theorem}

\ifthenelse{\boolean{shortver}}{We then extend our algorithm for 2-factor to get a PTAS. And obtained the following.}{}

\ifthenelse{\boolean{shortver}}{}{
\subsection{Polynomial-time Approximation Scheme}

The polynomial-time approximation scheme algorithm follows a similar idea to the factor-2 approximation algorithm. Intuitively, when we try to design 2-factor  algorithm we show that for any witness set $W= \{w_i : i \in [k]\}$ if we consider the witness set $W_{odd}=\{ w_i ~|~ w_i \in W, 1\leq i \leq k\}$ the set $Q_1$ forms a $W_{odd}$-\pws. Here we show that when we try $(1+\epsilon)$-factor approximation algorithm we show that for any witness set $W= \{w_i : i \in [k]\}$ if we consider the witness set $W_{apx}=\{w_j ~|~w_j \in W, j$ mod $(i+1) \neq 0\}$ then  $Q_{2i}$ forms a $W_{apx}$-\pws, where $ i = \frac{1}{\epsilon}$.

\begin{lemma} \label{theo-PTAS}
	Any witness set $W$ of size $ k $ contains a subset $ W_{apx} \subseteq W $, such that $ | W_{apx}| \geq  \frac{i  k}{i+1}$ and each witness in $  W_{apx} $ is $\wld$ in $Q_{2i}$ with respect to $  W_{apx} $.

\end{lemma}

\begin{proof}
	
	Let $W_{apx}=\{w_j ~|~w_j \in W, j$ mod $(i+1) \neq 0\}$. $  W_{apx} $ is a witness set of size at least $\frac{i k}{i+1}$. Now we show each witness in $  W_{apx} $ is $\wld$ in $Q_{2i}$ with respect to $  W_{apx} $, i.e., has a replacement in $Q_{2i}$. Surely, the witness $w_{1}$ and $w_k$ is $\wld$ in $Q_{2i}$ with respect to $  W_{apx} $.  Now fix a $j$ and  consider the witnesses $w_{j \times (i+1)+1}, w_{j \times (i+1)+2}, \cdots w_{j \times (i+1) +i} \in  W_{apx}$ whose indexes are of the form $\{(j \times (i+1)+1),(j \times (i+1)+2),\cdots,(j \times (i+1)+i)\}$, (except $ w_{1}$ and $w_k$).  If we can show that the vertex $w_{j \times (i+1)+1} \in  W_{apx}$ is well-defined on $Q_{2}$ then it trivially follows that 
	$\{(j \times (i+1)+1),\cdots,(j \times (i+1)+i)\}$ is $\wld$ in $Q_{2\times 1}^*, Q_{2\times 3}^*, \cdots, Q_{2 \times i -2}$  is well-defined on $Q_{2i}$. Towards proving that we apply a similar kind of argument as in \Cref{theo-2-factor}. In \Cref{theo-2-factor}, we show that if we consider $w_{2t-1}$ and $w_{2t+1}$ not as a witness, then we can find a replacement of $w_{2t}$ in $Q_1$. Here, we also do the exact same thing.
	
		 Consider the witness  $w_{j \times (i+1)}$.  Now there must exists a pair of consecutive vertical chords, say $\ell(v_{l})$ and $\ell(v_{l+1})$ such that  $w_{j \times (i+1)}$ lies in the vertical strip formed by  $\ell(v_{l})$ and $ \ell(v_{l+1})$ with  $\ell(v_{l}), \ell(v_{l+1}) \subset \vis(w_{j \times (i+1)})$ (follows from \Cref{obs-two-chords}). Note that $w_{j \times (i+1)+1}$ may not be well-defined in $Q_1$ with respect to the witness set $W$. But  if we consider the witness $W \smallsetminus \{ w_{j \times (i+1)}\}$, then we can show  that $w_{j \times (i+1)+1}$ is well-defined in $Q_1$ with respect to the witness set $W \smallsetminus \{w_{j \times (i+1)}\}$. Below, we prove that.

	Notice  that   $\vis(w_{j \times (i+1)-1}) \cap \ell(v_{l+1}) = \emptyset$ in fact  $\vis(w_{j \times (i+1)-1})\cap \ell(v_{l}) = \emptyset$. Also $\vis(w_{j \times (i+1)+1}) \cap \ell(v_{l+1}) = \emptyset$. Now consider  $\LC(w_{j \times (i+1)+1})$. It must pass through some reflex vertex, say $x^*$. Now consider the line  $L$ (not segment) passing through  $w_{w_{j \times (i+1)+1}}$ and $x^*$. If we rotate the line $L$ anti-clockwise through $x^*$ (i.e.,  the point $x^*$ must be in the line when rotating), it must hit some vertices in $V(\mo) \cup H$. Now consider the point $p$ where the rotating line hits some vertices in $V(\mo) \cup H$ for the first time. Either $p \in V(\mo)$ or $p$ is a vertex of $H$; more specifically, a vertex from the line segment $\ell(v_{l+1})$. Now we can find a vertex $w^*$ in the rotated line such that $w^*$ is a replacement of $w_{j \times (i+1)+1}$ in $Q_1$ with respect to the witness set   $W \smallsetminus \{w_{j \times (i+1)}\}$. This is true due the fact of combining (i) As $\vis(w_{j \times (i+1)-1}) \cap \ell(v_{l}) = \emptyset$ and $w^*$ can see at most one point of $ \ell(v_{l+1})$ that is also from $H$, $\ell(v_{l}) \cap \vis(w^*) = \emptyset$ that indeed imply $\vis(w_{j \times (i+1)-1}) \cap \vis(w^*) = \emptyset$, (ii) as we rotate the line $L$  anti-clock wise through $x^*$, {it may happen that $L$ hits a vertex of $V(\mo) \cup H$ for the first time on the right side. Even then, by a symmetric argument we can conclude that} $\vis(w_{j \times (i+1)+2}) \cap \vis(w^*) = \emptyset$. 	
	Given that $w_{j \times (i+1)+1}$ is $\wld$ in $Q_1$, we can apply the reasoning used to prove the correctness of \Cref{algo_1} to show that witnesses $w_{j \times (i+1)+2}, \cdots, w_{j \times (i+1) +i} \in  W_{apx}$ whose indexes are of the form $\{(j \times (i+1)+1),\cdots,(j \times (i+1)+i)\}$ is $\wld$ in $Q_{3}, Q_{5}, \cdots, Q_{2i}$. This completes the proof.		
\end{proof}

	In \Cref{theo-PTAS}, we are able to show that  $Q_{2i}$ contains a witness set of size at least $\frac{ik}{i+1}$ where the value $k$ can be arbitrary. So even for the optimum value, the argument works. Hence $Q_{2i}$ contains a factor $\frac{i+1}{i}$ solution (witness set). It remains to find a maximum witness set that is contained in $Q_{2i}$. Now since $|Q_{2i}|= \OO(nr^{\OO(i)})$ and the {\sc Discrete Witness Set} in monotone polygons is solvable in $\OO(|Q|^2 + |V(\po)||Q|)$ time (by \Cref{theo-finite-witness-in-M}).  Thus, we have the following result.}

\begin{theorem} \label{theo-PTAS-time}
	For every fixed  $\epsilon >1$, there is a  $(1+\epsilon)$-approximation algorithm for the {\sc Witness Set} problem in monotone polygons running in time $r^{\OO(1/\epsilon)} \cdot n^2$, where $ r$ and $n $ denote the number of reflex vertices and vertices, respectively, of input polygon.
\end{theorem}

%\sasanka{Conclusion needs to be written afresh. NPhardness does not hold for WS for general polygons, not only for monotone polygons...}

%\ifthenelse{\boolean{shortver}}{}{

	\section{Conclusion}\label{sec-conclusion}

   In the literature, although Amit et al.\cite{DBLP:journals/ijcga/AmitMP10} claim the problem to be \nph, there was no {\sf NP}-hardness result for the {\sc Witness Set} problem in any kind of polygons. We roughly argue that it is difficult to build polynomial-size gadgets to show {\sf NP}-hardness reduction for this problem. The reason is the following: let $ \Pi' $ be a gadget that we construct for the {\sc Witness Set}  problem in a simple polygon $\po$  from an {\sf NP}-complete problem $ \Pi $. Now $ \Pi' $ should consist of a polynomial size point set $ Q $ where witnesses of $ \po $ in the desired solution are constrained to lie on. As shown in \Cref{theo-finite-witness-in-M}, the {\sc Discrete Witness Set} problem can be solved in polynomial time for any polygon, so we will be unable to show the reduction between $ \Pi $ and $ \Pi' $. That is why we feel that {\sf NP}-hardness of this problem does not hold for simple polygons, though we do not rule out the hardness result for simple polygons. The preceding discussion implies that any {\sf NP}-hardness proof for this problem would likely require methods beyond standard polynomial-size discretization techniques. 
   %}

% Authors are thankful to Prof.~Joseph S. B. Mitchell for introducing this problem to us. This {\sf NP}-hard statement thing, according to Prof.~Joseph, is an inadvertent error. He strongly believes that {\sc Witness Set}  problem in a simple polygon should admit a polynomial-time algorithm. Our outerstring reduction also hints towards the observation in the sense that if we try to discretize the problem with polynomial-size points (gadgets), then that gadgets are more likely to have a polynomial time solution.

 \medskip

   \noindent
\textbf{Funding:} 
Satyabrata Jana: Supported by the Engineering and Physical Sciences Research Council
(EPSRC) via the project MULTIPROCESS (grant no. EP/V044621/1).

%\medskip
%\noindent
%\textbf{Author Contributions:} All authors contributed equally at every stage of this work.

\medskip
\noindent
\textbf{Acknowledgment:} Authors are sincerely thankful to Haim Kaplan, Matya Katz,  Joseph S. B. Mitchell  and Micha Sharir for several discussions and valuable comments which has truly helped us in writing this article.

	%optional

	%%
	%% Bibliography
	%%
	
	%% Please use bibtex, 

    \bibliographystyle{alpha}
	\bibliography{satya}
	
	\appendix
	
	%	\section{help}\label{sec:help}

\end{document}